\documentclass[a4paper, UKenglish, cleveref, autoref, thm-restate, numberwithinsect]{lipics-v2021}

\title{Tight Approximation Algorithms for Two-dimensional Guillotine Strip Packing}
\titlerunning{Guillotine Strip Packing} 

\author{Arindam Khan}{Department of Computer Science and Automation, Indian Institute of Science, Bangalore, India \and \url{https://www.csa.iisc.ac.in/~arindamkhan/} }{arindamkhan@iisc.ac.in}{https://orcid.org/0000-0001-7505-1687}{Arindam Khan was supported in part by Pratiksha Trust Young Investigator
Award, Google CSExplore Award, and Google India Research Award.}
\author{Aditya Lonkar}{Department of Computer Science and Automation, Indian Institute of Science, Bangalore, India \and \url{http://www.myhomepage.edu} }{adityaabhay@iisc.ac.in}{}{}
\author{Arnab Maiti}{Indian Institute of Technology, Kharagpur, India \and \url{https://sites.google.com/view/arnab-maiti/home} }{arnabmaiti@iitkgp.ac.in}{}{}
\author{Amatya Sharma}{Indian Institute of Technology, Kharagpur, India \and \url{https://aaysharma.github.io}}{amatya65555@iitkgp.ac.in}{}{}
\author{Andreas Wiese}{Technical University of Munich, Germany \and \url{https://discrete.ma.tum.de/people/professors/andreas-wiese.html} }{andreas.wiese@tum.de}{https://orcid.org/0000-0003-3705-016X}{Andreas Wiese was partially supported by the Fondecyt Regular grant 1200173.}

\authorrunning{A. Khan, A. Lonkar, A. Maiti, A. Sharma and A. Wiese} 

\Copyright{Arindam Khan, Aditya Lonkar,  Arnab Maiti, Amatya Sharma and Andreas Wiese} 

\ccsdesc[100]{Theory of computation~ Design and analysis of algorithms} 

\keywords{Approximation Algorithms, Two-Dimensional Packing, Rectangle Packing, Guillotine Cuts, Computational Geometry} 

\category{} 

\acknowledgements{A part of this work was done when Arnab Maiti and Amatya Sharma were  undergraduate interns at Indian Institute of Science.}

\nolinenumbers 

\EventEditors{John Q. Open and Joan R. Access}
\EventNoEds{2}
\EventLongTitle{The 49th EATCS International Colloquium on Automata, Languages, and Programming}
\EventShortTitle{ICALP 2022}
\EventAcronym{ICALP}
\EventYear{2022}
\EventDate{July 3--8, 2022}
\EventLocation{Paris, France}
\EventLogo{}
\SeriesVolume{42}
\ArticleNo{23}

\usepackage[T1]{fontenc}
\usepackage{amsthm}
\usepackage{anyfontsize}
\usepackage{amssymb}
\usepackage[dvipsnames]{xcolor}
\usepackage[utf8]{inputenc}
\usepackage{amsmath, amsbsy}
\usepackage{amsthm}
\usepackage{amsfonts}
\usepackage{graphicx}
\usepackage[colorinlistoftodos]{todonotes}
\usepackage{verbatim} 
\usepackage{enumitem}

\usetikzlibrary{calc}
\usepackage{tkz-euclide}
\usepackage{tikz,pgfplots}
\usepackage{forest}
\usetikzlibrary{patterns, arrows,positioning,shadows, trees, shapes}
\usetikzlibrary{decorations.pathreplacing}
\usepackage{caption, setspace, subcaption}
\usepackage{xspace}
\graphicspath{ {./image/} }
\usepackage{algorithm}	
\usepackage{lipsum}
\usepackage{algorithmic}

\makeatletter

\newcommand{\eps}{\varepsilon}

\newcommand{\epsau}{\varepsilon_{ra}}

\newcommand{\opt}{OPT}

\newcommand{\Rsm}{I_{small}}

\newcommand{\Rho}{I_{hor}}
\newcommand{\Rve}{I_{ver}}

\newcommand{\Rhard}{I_{hard}}
\newcommand{\R}{I}
\newcommand{\bottomc}{bottom}
\newcommand{\topc}{top}
\newcommand{\leftc}{left}
\newcommand{\rightc}{right}
\newcommand{\height}{h}
\newcommand{\width}{w}
\newcommand{\profit}{p}
\newcommand{\area}{a}

\newcommand{\Vc}{\boldsymbol{\mathsf{V}}}
\newcommand{\Hc}{\boldsymbol{\mathsf{H}}}
\newcommand{\Lc}{\boldsymbol{\mathsf{L}}}

\newcommand{\ts}{SP}
\newcommand{\tsg}{GSP}
\newcommand{\meu}{\mu}
\newcommand{\dlta}{\delta}
\newcommand{\W}{W}
\newcommand{\Rbg}{I_{large}}
\newcommand{\Rme}{I_{medium}}

\global\long\def\B{\mathcal{B}}%
\global\long\def\N{\mathbb{N}}%
\global\long\def\opt{\mathrm{OPT}}%
\global\long\def\height{h}%
\global\long\def\H{h_{\max}}%
\global\long\def\L{\mathcal{L}}%
	
\newcommand{\Rit}{I_{tall}}
\newcommand{\Ip}{{I}^{\prime}}

\newcommand{\Mi}{\textit{Mirror }}

\makeatother

\usepackage{babel}

\def\DEBUG{true} 
\ifdefined\DEBUG

  \def\rem#1{{\marginpar{\raggedright\scriptsize #1}}}
  \newcommand{\arir}[1]{\rem{\textcolor{Red}{$\bullet$ #1}}}
  \newcommand{\amar}[1]{\rem{\textcolor{Green}{$\bullet$ #1}}}
  \newcommand{\adir}[1]{\rem{\textcolor{Maroon}{$\bullet$ #1}}}
  \newcommand{\arnr}[1]{\rem{\textcolor{Violet}{$\bullet$ #1}}}
  \newcommand{\awr}[1]{\rem{\textcolor{NavyBlue}{$\bullet$ #1}}}
\else

  \newcommand{\arir}[1]{}
  \newcommand{\adir}[1]{}
  \newcommand{\arnr}[1]{}
  \newcommand{\amar}[1]{}
  \newcommand{\awr}[1]{}
\fi

\begin{document}	
	\maketitle

\begin{abstract}
In the \textsc{Strip Packing} problem (\ts), we are given a vertical
half-strip $[0,W]\times[0,\infty)$ and a set of $n$ axis-aligned
rectangles of width at most $W$. The goal is to find a non-overlapping
packing of all rectangles into the strip such that the height of the
packing is minimized. A well-studied and frequently used practical
constraint is to allow only those packings that are guillotine separable,
i.e., every rectangle in the packing can be obtained by recursively
applying a sequence of edge-to-edge axis-parallel cuts (guillotine
cuts) that do not intersect any item of the solution. In this paper,
we study approximation algorithms for the \textsc{Guillotine Strip Packing} problem
(\tsg), i.e., the \textsc{Strip Packing} problem where we require additionally that the packing
needs to be guillotine separable. This problem generalizes the classical
\textsc{Bin Packing} problem and also makespan minimization on identical
machines, and thus it is already strongly $\mathsf{NP}$-hard. Moreover,
due to a reduction from the \textsc{Partition} problem, it is $\mathsf{NP}$-hard
to obtain a polynomial-time $(3/2-\eps)$-approximation algorithm
for GSP  for any $\eps>0$ (exactly as \textsc{Strip Packing}). We provide a matching polynomial time
$(3/2+\eps)$-approximation algorithm for \tsg. Furthermore, we present
a pseudo-polynomial time $(1+\eps)$-approximation algorithm for \tsg.
This is surprising as it is $\mathsf{NP}$-hard to obtain a $(5/4-\eps)$-approximation
algorithm for (general) \textsc{Strip Packing} in pseudo-polynomial
time. Thus, our results essentially settle the approximability of
GSP
 for both the polynomial and the
pseudo-polynomial settings. 
\end{abstract}

\section{Introduction}

Two-dimensional packing problems form a fundamental research area
in combinatorial optimization, computational geometry, and approximation
algorithms. They find numerous practical applications in logistics
\cite{CKPT17}, cutting stock \cite{gilmore1965multistage}, VLSI
design \cite{Hochbaum1985}, smart-grids \cite{GGIK16}, etc. The
\textsc{Strip Packing} problem (\ts), a generalization of the classical
\textsc{Bin Packing} problem and also the makespan minimization problem
on identical machines, is one of the central problems in this
area. We are given an axis-aligned vertical half-strip $[0,W]\times[0,\infty)$
and a set of $n$ axis-aligned rectangles (also called {\em items})
$I:=\{1,2,\dots,n\}$, where for each rectangle $i$ we are given an integral width $\width_{i}\le W$,
and an integral height $\height_{i}$; we assume the rectangles
to be open sets. The goal is to pack all items such that the maximum
height of the top edge of a packed item is minimized. The packing needs to be
{\em non-overlapping}, i.e., such a packing into a strip of height
$H$ maps each rectangle $i\in I$ to a new translated open rectangle
$R(i):=(\leftc(i),\rightc(i))\times(\bottomc(i),\topc(i))$ where
$\rightc(i)=\leftc(i)+\width_{i}$, $\topc(i)=\bottomc(i)+\height_{i}$, $\leftc(i)\ge0$, $\bottomc(i)\ge0$, $\rightc(i)\le W$, $\topc(i)\le H$
and for any $i,j\in I$, we must have $R(i)\cap R(j)=\emptyset$.
We assume that items are not allowed to be rotated.

The best known polynomial time approximation algorithm for SP has
an approximation ratio of ($5/3+\eps$) for any constant $\eps>0$ \cite{HarrenJPS14}
and a straight-forward reduction from \textsc{Partition }shows that
it is $\mathsf{NP}$-hard to approximate the problem with a ratio
of $(3/2-\eps)$ for any $\eps>0$. Maybe surprisingly, one can
approximate SP better in pseudo-polynomial time: there is a pseudo-polynomial
time $(5/4+\eps)$-approximation algorithm \cite{JansenR19} and it
is $\mathsf{NP}$-hard to obtain a $(5/4-\eps)$-approximation algorithm
with this running time \cite{henning2020complexity}. Hence, it remains open
to close the gap between $(5/3+\eps)$ and $(3/2-\eps)$ for polynomial
time algorithms, and even in pseudo-polynomial time, there can be no
$(1+\eps)$-approximation for the problem for arbitrarily small $\eps>0$.

\ts~is particularly motivated from applications in which we want
to cut out rectangular pieces of a sheet or stock unit of raw material,
i.e., metal, glass, wood, or, cloth, and we want to minimize the amount
of wasted material. For cutting out these pieces in practice, axis-parallel
end-to-end cuts, called {\em guillotine cuts}, are popular due
to their simplicity of operation~\cite{sweeney1992cutting}. In this context, we look for solutions to cut out the
individual objects by a recursive application of guillotine cuts that
do not intersect any item of the solution. 
Applications of guillotine cutting are found in crepe-rubber mills \cite{schneider1988trim}, glass industry \cite{puchinger2004solving},  paper cutting \cite{mchale1999cutting}, etc. 
In particular, this motivates studying geometric packing problems
with the additional constraint that the placed objects need to be
separable by a sequence of guillotine cuts (see Figure~\ref{fig_guillo}). 
Starting from the classical work by Christofides et al.~\cite{christofides1977algorithm} in 1970s, settings with such guillotine
cuts are widely studied in the literature \cite{dolatabadi2012exact,wei2015bidirectional,borgulya2019eda,lodi2017partial,di2013algorithms,clautiaux2018combining,furini2016modeling,clautiaux2019pattern}.
In fact, many heuristics for guillotine packing have been developed to efficiently solve benchmark instances, based on tree-search, branch-and-bound, dynamic optimization, tabu search, genetic algorithms, etc. Khan et al. \cite{khan2021guillotine} mentions ``a staggering number of recent experimental papers” on guillotine packing and lists several such recent experimental papers.

A related notion is $k$-stage packing, originally introduced by Gilmore and Gomory~\cite{gilmore1965multistage}.
Here, each stage consists of either vertical or horizontal guillotine cuts (but not both). In
each stage, each of the pieces obtained in the previous stage is considered separately
and can be cut again by using either horizontal or vertical guillotine cuts. In $k$-stage packing,
the number of cuts to obtain each rectangle from the initial packing is at most $k$, plus an
additional cut to trim (i.e., separate the rectangles itself from a waste area). Intuitively, this
means that in the cutting process we change the orientation of the cuts $k -1$ times. 

\begin{figure}[h]
\centering \captionsetup[subfigure]{justification=centering}
\begin{subfigure}[b]{0.3\linewidth} \centering \resizebox{3cm}{3cm}{ \begin{tikzpicture}
\draw [very thick] (0,0) rectangle (8,10);
\draw [ultra thick, blue] (3,0) -- (3,10);
\draw [ultra thick, blue] (6.5,0) -- (6.5,10);
\draw [ultra thick, red]  (0,5.5) -- (3,5.5);
\draw [ultra thick, red]  (3,6.8) -- (6.5,6.8);
\draw [ultra thick, red]  (3,2.9)-- (6.5,2.9);
\draw [ultra thick, red]  (6.5,8) -- (8,8);
\draw[ultra thick, red] (6.5,6.7)--(8,6.7);

\draw [ultra thick, green] (0.8, 5.5) -- (0.8,10) ;
\draw [ultra thick, green] (1.45,5.5) -- (1.45,10) ;
\draw [ultra thick, green] (2.15,5.5) -- (2.15,10) ;
\draw [ultra thick, green] (2,0) -- (2,5.5) ;
\draw [ultra thick, green] (4.25,6.8) -- (4.25,10);
\draw [ultra thick, green] (5.4,2.9) -- (5.4,6.8);
\draw [ultra thick, green] (4.7,2.9) -- (4.7,6.8);
\draw [ultra thick, green] (3.9,2.9) -- (3.9,6.8);
\draw [ultra thick, green] (4.5,0) -- (4.5,2.9);
\draw [ultra thick, green] (7.3, 8) -- (7.3, 10);
 
\draw [ultra thick, yellow] (3,1.3)-- (4.5,1.3);
\draw [ultra thick, yellow] (5.4,4)-- (6.5,4);
\draw [ultra thick, yellow] (4.25,8.8)-- (6.5,8.8);
\draw[ultra thick, yellow] (0,4.5)--(2,4.5);
\draw[ultra thick, yellow] (0,3.4)--(2,3.4);
\draw[ultra thick, yellow] (0,2)--(2,2);

\draw [ultra thick, brown] (5.1,6.8)-- (5.1,8.8);
\draw [ultra thick, brown] (5.8,6.8) -- (5.8,8.8);

\draw [fill=gray!50, very thick]   (7,0) rectangle (7.9,6.6);
\draw [fill=gray!50, very thick]   (6.75,6.8) rectangle (8,7.8);
\draw [fill=gray!50, very thick] (6.6,8) rectangle (7.2,9.8);
\draw [fill=gray!50, very thick] (7.5,9) rectangle (8,9.9);
\draw [fill=gray!50, very thick] (4.6,0) rectangle (6.3,2.7);
\draw [fill=gray!50, very thick] (3.05,0) rectangle (4.4,1.2);
\draw [fill=gray!50, very thick] (3.2,1.4) rectangle (4.2,2.85);
\draw [fill=gray!50, very thick] (0,5.55) rectangle (0.75,9.5) ;
\draw [fill=gray!50, very thick] (0.85,5.55) rectangle (1.4,9.3) ;
\draw [fill=gray!50, very thick] (1.5,5.55) rectangle (2.1,9);
\draw [fill=gray!50, very thick] (2.2,5.55) rectangle (2.95,8.5) ;
\draw [fill=gray!50, very thick]  (2.1,0) rectangle (2.85,5);

\draw [fill=gray!50, very thick]  (0,4.55) rectangle (1.9,5.4);
\draw [fill=gray!50, very thick]  (0,3.45) rectangle (1.8,4.45);
\draw [fill=gray!50, very thick]  (0,2.1) rectangle (1.65,3.35);
\draw [fill=gray!50, very thick]  (0,0.2) rectangle (1.5,1.9);
\draw [fill=gray!50, very thick]  (3.1,6.85) rectangle (4.2,9.9);
\draw [fill=gray!50, very thick]  (4.4,8.9) rectangle (6.4,9.9);
\draw [fill=gray!50, very thick]  (4.3,6.85) rectangle (5.05,8.7);
\draw [fill=gray!50, very thick]  (5.15,6.85) rectangle (5.75,8.5);
\draw [fill=gray!50, very thick]  (5.85,6.85) rectangle (6.45,8.2);
\draw [fill=gray!50, very thick]  (5.5,2.95) rectangle (6.45,3.9);
\draw [fill=gray!50, very thick]  (5.7,4.05) rectangle (6.4,6);
\draw [fill=gray!50, very thick]  (3.05,2.95) rectangle (3.85,6.6);
\draw [fill=gray!50, very thick]  (3.95,2.95) rectangle (4.65,6.3);
\draw [fill=gray!50, very thick]  (4.75,2.95) rectangle (5.35,6);
\end{tikzpicture}} \caption{}
\end{subfigure} \begin{subfigure}[b]{0.3\linewidth} \centering
\resizebox{3cm}{3cm}{ \begin{tikzpicture}
\draw [very thick] (0,0) rectangle (10,10);
\draw [fill=gray!50, very thick] (0,0) rectangle (10,0.6);
\draw [fill=black!30, very thick] (0,0.6) rectangle (0.6,10);
\draw [fill=gray!90, very thick] (0.6,0.6) rectangle (10,1.2);
\draw [fill=black!65, very thick] (0.6,1.2) rectangle (1.2,10);
\draw [fill=gray!50, very thick] (1.2,1.2) rectangle (10,1.8);
\draw [fill=black!30, very thick] (1.2,1.8) rectangle (1.8,10);
\draw [fill=gray!90, very thick] (1.8,1.8) rectangle (10,2.4);
\draw [fill=black!65, very thick] (1.8,2.4) rectangle (2.4,10);
\draw [fill=gray!50, very thick] (2.4,2.4) rectangle (10,3);
\draw [fill=black!30, very thick] (2.4,3) rectangle (3,10);
\draw [fill=gray!90, very thick] (3,3) rectangle (10,3.6);
\draw [fill=black!65, very thick] (3,3.6) rectangle (3.6,10);
\draw [fill=gray!50, very thick] (3.6,3.6) rectangle (10,4.2);
\draw [fill=black!30, very thick] (3.6,4.2) rectangle (4.2,10);
\draw [fill=gray!90, very thick] (4.2,4.2) rectangle (10,4.8);
\end{tikzpicture}} \caption{}
\label{subfig_L_nice} \end{subfigure} \begin{subfigure}[b]{0.3\linewidth}
\centering \resizebox{3cm}{3cm}{ \begin{tikzpicture}
\draw [very thick] (0,0) rectangle (8,10);

\draw [ultra thick, red, dashed]  (0,6.7) -- (8,6.7);

\draw [fill=black!50, very thick, fill opacity=0.7]  (0,3.4) rectangle (2.6,9.9);
\draw [fill=gray!50, very thick]  (0,0) rectangle (5.3,3.3);
\draw [fill=gray!50, very thick]  (5.4,0) rectangle (7.95,6.55);
\draw [fill=gray!50, very thick]  (2.7,6.85) rectangle (7.95,10);

\end{tikzpicture}} \caption{}
\end{subfigure} \caption{Packing (a) is a $5$-stage guillotine separable packing, packing
(b) is a $(n-1)$-stage guillotine separable packing, packing (c) is
not guillotine separable as any end-to-end cut in the strip intersects
a rectangle.}
\label{fig_guillo} 
\end{figure}
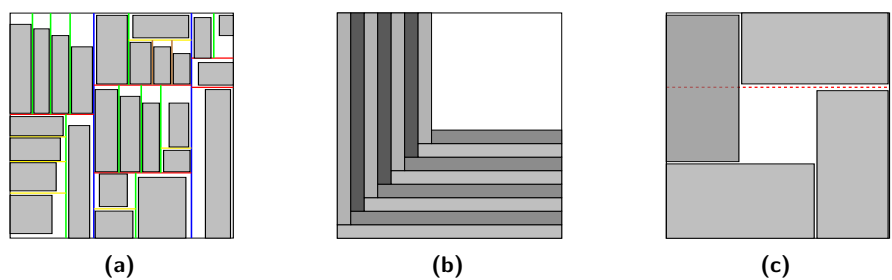

Therefore, in this paper, we study the \textsc{Guillotine Strip Packing}
problem (\tsg). The input is the same as for \ts, but we require
additionally that the items in the solution can be separated by a
sequence of guillotine cuts, and we say then that  they are \emph{guillotine
separable}. Like general \ts~without requiring the items to be guillotine
separable, \tsg~generalizes \textsc{Bin Packing }(when all items
have the same height) and makespan minimization on identical machines
(when all items have the same width). Thus, it is strongly $\mathsf{NP}$-hard,
and the same reduction from \textsc{Partition} mentioned above yields
a lower bound of $(3/2-\eps)$ for polynomial time algorithms (see Appendix \ref{section_hardness} for more details). 
For asymptotic approximation, \tsg~is well understood. 
Kenyon and R{é}mila \cite{kenyon2000near} gave an asymptotic polynomial time approximation scheme (APTAS)
for (general) \ts. Their algorithm produces a 5-stage packing (hence, guillotine
separable), and thus yields an APTAS for \tsg~as well.
Later, Seiden et al.~\cite{seiden2005two} settled the asymptotic approximation status of \tsg~under $k$-stage packing.
They gave an APTAS for \tsg~using 4-stage guillotine cuts, and showed 
$k = 2$ stages cannot guarantee any bounded asymptotic performance ratio, and $k = 3$ stages lead to
asymptotic performance ratios close to $1.691$.
However, in the non-asymptotic setting,  approximation ratio of \tsg~is not yet settled. 
Steinberg's algorithm \cite{steinberg1997strip} yields a 2-approximation algorithm
for \tsg~and this is the best known polynomial time approximation algorithm for the problem.

In this paper we present approximation algorithms for \tsg~which
have strictly better approximation ratios than the best known algorithms
for \ts, and in the setting of pseudo-polynomial time algorithms
we even beat the lower bound that holds for \ts. Moreover, we show
that all our approximation ratios are essentially the best possible.

\subsection{Our Contribution}

We present a polynomial time $(3/2+\eps)$-approximation algorithm
for \tsg. Due to the mentioned lower bound of $(3/2-\eps)$, our approximation
ratio is essentially tight. Also, we present a pseudo-polynomial time
$(1+\eps)$-approximation algorithm, which is also essentially
tight since GSP is strongly $\mathsf{NP}$-hard.

For the pseudo-polynomial time
$(1+\eps)$-approximation, we first prove that there exists a structured solution with height
at most $(1+\eps)\opt$ ($\opt$ denotes the height of the optimal solution) in which the strip is divided into $O(1)$
rectangular boxes inside which the items are {\em nicely packed},
e.g., horizontal items are stacked on top of each other, vertical
items are placed side by side, and small items are packed greedily
with the Next-Fit-Decreasing-Height algorithm \cite{coffman1980performance}
(see Figure~\ref{fig:pptas_poly}(\subref{fig:pptas_poly_1}) and also Figure~\ref{fig_skew_nice}. Also, refer to Section~\ref{sec:PPTAS} for item classification). 
This result starkly contrasts \ts~(i.e.,
where we do not require the items to be guillotine separable): for
that problem, it is already unlikely that we can prove that there
always exists such a packing with a height of less than $5/4\cdot\opt$.
If we could prove this, we could approximate the problem in pseudo-polynomial
time with a better ratio than $5/4$, which is $\mathsf{NP}$-hard~\cite{henning2020complexity}.

To construct our structured packing, we start with an optimal packing
and use the techniques in \cite{khan2021guillotine} to obtain
a packing in which each item is {\em nicely packed} in one of a
constant number of boxes and L-shaped {\em compartments}. We increase
the height of our packing by $\eps\opt$ in order to round the heights
of the packed items and get some leeway within the packing. Then,
we rearrange the items placed inside the L-shaped compartments.
Here, we crucially exploit that the items in the initial packing  are guillotine
separable. In particular, this property allows us to identify certain
sets of items that we can swap, e.g., items on the left and the
right of a vertical guillotine cut to simplify the packing, and reduce
the number of boxes to $O(1)$. Then, using standard techniques, we
compute a solution with this structure in pseudo-polynomial time and
hence with a packing height of at most $(1+\eps)\opt$ (see Figure~\ref{fig:pptas_poly} (\subref{fig:pptas_poly_1})).

\begin{figure}[tb]

\hspace*{-7em}
\captionsetup[subfigure]{justification=centering,aboveskip=1em}
\centering
\begin{subfigure}[b]{0.28\linewidth}
\centering
\resizebox{4.4cm}{4.21cm}{
\begin{tikzpicture}
\draw [ultra thick](0,0) rectangle (8,10);
%

\draw [ultra thick](0,0) rectangle (1.8,6.5);
\draw [ultra thick](1.85,0) rectangle (2.6,5.5);
\draw [ultra thick](2.65,0) rectangle (3.6,3.2);
\draw [ultra thick](3.65,0) rectangle (5.15,3.3);
\draw [ultra thick](2.65,3.3) rectangle (3.6,5.6);
\draw [ultra thick](3.65,3.3) rectangle (5.15,6.8);
\draw [ultra thick](5.2,0) rectangle (6.1,7.1);
\draw [ultra thick](7.15,0) rectangle (8,8);
\draw [ultra thick](6.1,0) rectangle (7.1,4);
\draw [ultra thick](6.1,4) rectangle (7.05,8.1);

%
\draw [ultra thick](0.1,8.5) rectangle (5,9.9);
\draw [ultra thick](6.1,8.2) rectangle (8,9.8);
\draw [ultra thick](0.1,6.6) rectangle (1.7,8.4);
\draw [ultra thick](1.85,5.65) rectangle (3.6,8.3);
\draw [ultra thick](5.25,7.1) rectangle (6.05,9.9);
\draw [ultra thick](3.65,6.85) rectangle (5.1,8.4);

\draw [ultra thick, fill=gray!60] (0,0) rectangle (0.85,6.45);
\draw [ultra thick, fill=gray!60] (0.9,0) rectangle (1.3,6.1);
\draw [ultra thick, fill=gray!60] (1.35,0) rectangle (1.75,5.9);
\draw [ultra thick, fill=gray!60] (1.85,0) rectangle (2.15,5.45);
\draw [ultra thick, fill=gray!60] (2.2,0) rectangle (2.55,5.1);
\draw [ultra thick, fill=gray!60] (5.2,0) rectangle (5.45,7.05);
\draw [ultra thick, fill=gray!60] (5.5,0) rectangle (6,6);
\draw [ultra thick, fill=gray!60] (7.15,0) rectangle (7.6,7.95);
\draw [ultra thick, fill=gray!60] (7.65,0) rectangle (7.95,7);

\draw [ultra thick, fill=cyan] (0.1,9.85) rectangle (4.9,9.55);
\draw [ultra thick, fill=cyan] (0.1,9.5) rectangle (4.65,9.05);
\draw [ultra thick, fill=cyan] (0.1,9) rectangle (4.4,8.6);

\draw [ultra thick, fill=orange] (3.7,0) rectangle (4,3.2);
\draw [ultra thick, fill=orange] (4.05,0) rectangle (4.6,3.05);
\draw [ultra thick, fill=orange] (4.65,0) rectangle (5.1,2.9);

\draw [ultra thick, fill=orange] (3.65,3.35) rectangle (4.4,6.7);
\draw [ultra thick, fill=orange] (4.45,3.35) rectangle (5,6.5);

\draw [ultra thick, fill=orange] (5.3,7.15) rectangle (5.5,9.6);
\draw [ultra thick, fill=orange] (5.55,7.15) rectangle (6,9.2);

\draw [ultra thick, fill=orange] (6.15,0) rectangle (6.35,3.8);
\draw [ultra thick, fill=orange] (6.4,0) rectangle (6.7,3.5);
\draw [ultra thick, fill=orange] (6.75,0) rectangle (7.05,3.2);

\draw [ultra thick, fill=orange] (6.15,4.05) rectangle (6.5,8);
\draw [ultra thick, fill=orange] (6.5,4.05) rectangle (6.8,7.5);

\draw [ultra thick, fill=orange] (2.7,0) rectangle (3.2,3.15);
\draw [ultra thick, fill=orange] (3.25,0) rectangle (3.55,2.9);

\draw [ultra thick, fill=orange] (2.7,3.35) rectangle (2.95,5.5);
\draw [ultra thick, fill=orange] (3,3.35) rectangle (3.25,5.2);
\draw [ultra thick, fill=orange] (3.3,3.35) rectangle (3.6,4.9);

\draw [ultra thick, fill=brown] (0.15,6.65) rectangle (1.65,8.35);
\draw [ultra thick, fill=brown] (3.7,6.9) rectangle (5.05,8.35);

\draw [ultra thick, fill=magenta] (1.9,5.7) rectangle (2.2,6.7);
\draw [ultra thick, fill=magenta] (2.25,5.7) rectangle (2.9,6.6);
\draw [ultra thick, fill=magenta] (2.95,5.7) rectangle (3.4,6.5);

\draw [ultra thick, fill=magenta] (1.9,6.8) rectangle (2.45,7.55);
\draw [ultra thick, fill=magenta] (2.5,6.8) rectangle (2.7,7.45);
\draw [ultra thick, fill=magenta] (2.75,6.8) rectangle (3.2,7.35);

\draw [ultra thick, fill=magenta] (1.9,7.65) rectangle (2.35,8.15);
\draw [ultra thick, fill=magenta] (2.4,7.65) rectangle (2.6,8.05);
\draw [ultra thick, fill=magenta] (2.65,7.65) rectangle (2.9,7.9);
\draw [ultra thick, fill=magenta] (2.95,7.65) rectangle (3.5,7.8);

\draw [ultra thick, fill=magenta] (6.15,8.2) rectangle (6.4,8.95);
\draw [ultra thick, fill=magenta] (6.45,8.2) rectangle (7,8.85);
\draw [ultra thick, fill=magenta] (7.05,8.2) rectangle (7.8,8.85);

\draw [ultra thick, fill=magenta] (6.15,9.1) rectangle (6.5,9.7);
\draw [ultra thick, fill=magenta] (6.55,9.1) rectangle (7.2,9.65);
\draw [ultra thick, fill=magenta] (7.25,9.1) rectangle (7.9,9.6);

\draw (-1,5) node {\LARGE $\boldsymbol{\frac{1}{2}\textbf{OPT}}$};
\draw (-1,10) node {\LARGE $\boldsymbol{\textbf{OPT}}$};

\draw (-1.75,11) node[inner sep=0pt] {\LARGE $\boldsymbol{(1+\eps)\textbf{OPT}}$};
\filldraw [ultra thick,color=red!60,pattern=north east lines,pattern color =GreenYellow] (0,10) rectangle (8,11) ;
\draw (4,10.5) node {\LARGE $\boldsymbol{S}$};
\draw [dashed] (0,5) -- (8,5);
\end{tikzpicture}}
\caption{\hspace*{-20mm}}
\label{fig:pptas_poly_1}
\end{subfigure}
\hspace*{1.3in}
\begin{subfigure}[b]{0.3\linewidth}
\resizebox{4.7cm}{6.1cm}{
\begin{tikzpicture}
 \draw [ultra thick](0,0) rectangle (8,15);
%

\draw [ultra thick](4.45,0) rectangle (5.4,3.2);
\draw [ultra thick](5.45,0) rectangle (6.95,3.3);
\draw [ultra thick](4.45,3.3) rectangle (5.4,5.6);
\draw [ultra thick](5.45,6.55) rectangle (6.95,10.05);

\draw [ultra thick](7,0) rectangle (8,4);
\draw [ultra thick](7,6.55) rectangle (8,10.65);

\draw [ultra thick](0.1,13.5) rectangle (5,14.9);
\draw [ultra thick](6.1,13.2) rectangle (8,14.8);
\draw [ultra thick](0.1,11.6) rectangle (1.7,13.4);

%
\draw [ultra thick](3.65,7.2) rectangle (5.4,9.85);
\draw [ultra thick](5.25,12.1) rectangle (6.05,14.9);
\draw [ultra thick](3.65,11.85) rectangle (5.1,13.4);

\draw [ultra thick, fill=gray!60] (0,0) rectangle (0.45,7.95);
\draw [ultra thick, fill=gray!60] (0.5,0) rectangle (0.75,7.05);
\draw [ultra thick, fill=gray!60] (0.8,0) rectangle (1.1,7);
\draw [ultra thick, fill=gray!60] (1.15,0) rectangle (2,6.45);
\draw [ultra thick, fill=gray!60] (2.05,0) rectangle (2.45,6.1);
\draw [ultra thick, fill=gray!60] (2.5,0) rectangle (3,6);
\draw [ultra thick, fill=gray!60] (3.05,0) rectangle (3.45,5.9);
\draw [ultra thick, fill=gray!60] (3.5,0) rectangle (3.8,5.45);
\draw [ultra thick, fill=gray!60] (3.85,0) rectangle (4.2,5.1);

\draw [ultra thick, fill=cyan] (0.1,14.85) rectangle (4.9,14.55);
\draw [ultra thick, fill=cyan] (0.1,14.5) rectangle (4.65,14.05);
\draw [ultra thick, fill=cyan] (0.1,14) rectangle (4.4,13.6);


\draw [ultra thick, fill=brown] (0.15,11.65) rectangle (1.65,13.35);
\draw [ultra thick, fill=brown] (3.7,11.9) rectangle (5.05,13.35);

\draw [ultra thick, fill=magenta] (3.7,7.25) rectangle (4,8.25);
\draw [ultra thick, fill=magenta] (4.05,7.25) rectangle (4.7,8.15);
\draw [ultra thick, fill=magenta] (4.75,7.25) rectangle (5.2,8.05);

\draw [ultra thick, fill=magenta] (3.7,8.35) rectangle (4.25,9.1);
\draw [ultra thick, fill=magenta] (4.3,8.35) rectangle (4.5,9);
\draw [ultra thick, fill=magenta] (4.55,8.35) rectangle (5,8.9);

\draw [ultra thick, fill=magenta] (3.7,9.2) rectangle (4.15,9.7);
\draw [ultra thick, fill=magenta] (4.2,9.2) rectangle (4.4,9.6);
\draw [ultra thick, fill=magenta] (4.45,9.2) rectangle (4.7,9.45);
\draw [ultra thick, fill=magenta] (4.75,9.2) rectangle (5.3,9.35);

\draw [ultra thick, fill=magenta] (6.15,13.2) rectangle (6.4,13.95);
\draw [ultra thick, fill=magenta] (6.45,13.2) rectangle (7,13.85);
\draw [ultra thick, fill=magenta] (7.05,13.2) rectangle (7.8,13.85);

\draw [ultra thick, fill=magenta] (6.15,14.1) rectangle (6.5,14.7);
\draw [ultra thick, fill=magenta] (6.55,14.1) rectangle (7.2,14.65);
\draw [ultra thick, fill=magenta] (7.25,14.1) rectangle (7.9,14.6);

\draw [ultra thick, fill=orange] (5.5,0) rectangle (5.8,3.2);
\draw [ultra thick, fill=orange] (5.85,0) rectangle (6.4,3.05);
\draw [ultra thick, fill=orange] (6.45,0) rectangle (6.9,2.9);

\draw [ultra thick, fill=orange] (5.3,12.15) rectangle (5.5,14.6);
\draw [ultra thick, fill=orange] (5.55,12.15) rectangle (6,14.2);

\draw [ultra thick, fill=orange] (7.05,0) rectangle (7.25,3.8);
\draw [ultra thick, fill=orange] (7.3,0) rectangle (7.6,3.5);
\draw [ultra thick, fill=orange] (7.65,0) rectangle (7.95,3.2);

\draw [ultra thick, fill=orange] (7.05,6.55) rectangle (7.4,10.5);
\draw [ultra thick, fill=orange] (7.4,6.55) rectangle (7.7,10);

\draw [ultra thick, fill=orange] (4.5,0) rectangle (5,3.15);
\draw [ultra thick, fill=orange] (5.05,0) rectangle (5.35,2.9);

\draw [ultra thick, fill=orange] (4.5,3.35) rectangle (4.75,5.5);
\draw [ultra thick, fill=orange] (4.8,3.35) rectangle (5.05,5.2);
\draw [ultra thick, fill=orange] (5.1,3.35) rectangle (5.4,4.9);

\draw [ultra thick, fill=orange] (5.5,6.6) rectangle (6.2,9.95);
\draw [ultra thick, fill=orange] (6.25,6.6) rectangle (6.8,9.75);

\filldraw [ultra thick,color=red!60,pattern=north west lines,pattern color =green] (2,6.5) rectangle (3.4,11.5) ;
\filldraw [ultra thick,color=red!60,pattern=north east lines,pattern color =GreenYellow] (0,15) rectangle (8,16) ;

\draw (-1,5) node {\LARGE $\boldsymbol{\frac{1}{2}\textbf{OPT}}$};
\draw (-1,10) node {\LARGE $\boldsymbol{\textbf{OPT}}$};
\draw (-1,15) node {\LARGE $\boldsymbol{\frac{3}{2}\textbf{OPT}}$};
\draw (2.7,9) node {\LARGE $\boldsymbol{B^{*}}$};
\draw (4,15.5) node {\LARGE $\boldsymbol{S}$};
\draw (-1.75,16) node[inner sep=0pt] {\LARGE $\boldsymbol{(\frac{3}{2}+\eps)\textbf{OPT}}$};
\draw [ultra thick,decorate,
    decoration = { brace,amplitude=10pt}] (1.9,6.55) --  (1.9,11.5) node[pos=0.5,left=6pt,black]{\Large$\boldsymbol{\frac{1}{2}\textbf{OPT}}$};

\draw [dashed] (0,5) -- (8,5);
\draw [dashed] (0,10) -- (8,10);
\end{tikzpicture}}
\caption{\hspace*{-23mm}}
\label{fig:pptas_poly_2}
\end{subfigure}
\caption{(a) A guillotine separable structured packing (for the pseudo-polynomial time approximation scheme) where all the items are packed nicely in containers. The  tall items (dark-gray) are stacked next to each other just like the vertical items (orange); the horizontal items (blue) are stacked on top of each other, the small items (pink) are packed according to NFDH, and the large containers contain single large items (brown). \;(b) A guillotine separable structured packing for the polynomial time $(3/2+\eps)$-approximation, where the packing from $(a)$ is rearranged such that the tall items are \textit{bottom-left-flushed} and there is an extra empty box $B^*$ to accommodate some of the vertical items which we are unable to pack in polynomial time in the rest of the guessed boxes. This arises from the $\mathsf{NP}$-hardness of the \textsc{Partition} problem.  The yellow rectangular strip $S$ on top of both the packings is used for packing the medium and leftover horizontal and small items.\vspace{-3.5mm}}
\label{fig:pptas_poly}

\end{figure}
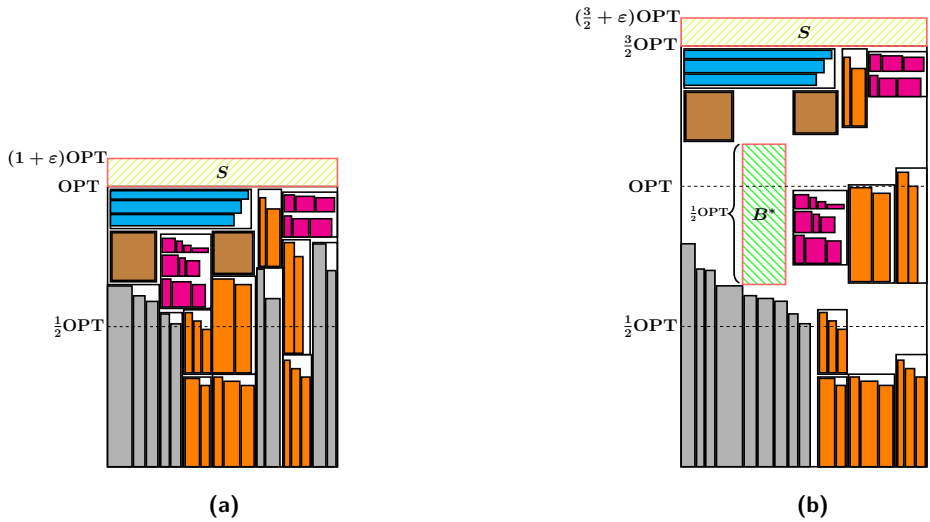

Note that we do not obtain a $(1+\eps)$-approximation algorithm
in polynomial time in this way. The reason is that when we pack the
items into the rectangular boxes, we need to solve a generalization
of \textsc{partition:} there can be several boxes in which vertical
items are placed side by side, and we need that the widths of the
items in each box sum up to at most the width of the box. If there
is only a single item that we cannot place, then we would need to
place it on top of the packing, which can increase our packing height
by up to $\text{OPT}$. 

For our polynomial time $(3/2+\eps)$-approximation algorithm, we, therefore,
need to be particularly careful with the items whose height is larger
than $\opt/2$, which we call the \emph{tall items}. We prove a different
structural result which is the main technical contribution of this
paper: we show that there is always a $(3/2+\eps)$-approximate packing
in which the tall items are packed together in a {\em bottom-left-flushed}
way, i.e., they are ordered non-increasingly by height and stacked
next to each other with their bottom edges touching the base of the
strip. All remaining items are nicely packed into $O_{\eps}(1)$ boxes,
and there is also an empty strip of height $\opt/2$ and width $\Omega_{\eps}(W)$,
see Figure~\ref{fig:pptas_poly}~(\subref{fig:pptas_poly_2}). Thus, it is very easy to pack the tall items correctly
according to this packing. We pack the remaining items with standard
techniques into the boxes. In particular, the mentioned empty strip
allows us to make slight mistakes while we pack the vertical items
that are not tall; without this, we would still need to solve a generalization
of \textsc{partition}. 

In order to obtain our structural packing for our polynomial time
$(3/2+\eps)$-approximation algorithm, we build on the idea of the packing for
the pseudo-polynomial time $(1+\eps)$-approximation. Using
that it is guillotine separable, we rearrange its items further. First,
we move the items such that all tall items are at the bottom. To achieve
this, we again argue that we can swap certain sets of items,
guided by the guillotine cuts. Then, we shift certain items up by
$\opt/2$, which leaves empty space between the shifted and the not-shifted
items,
 see Figure~\ref{fig:pptas_poly}~(\subref{fig:pptas_poly_2}).
Inside this empty space, we place the empty box of height $\opt/2$.
Also, we use this empty space in order to be able to reorder the tall
items on the bottom by their respective heights. During these
changes, we ensure carefully that the resulting packing stays guillotine
separable.

It is possible that also for (general) SP there always exists a structured
packing of height at most $(3/2+\eps)\opt$, similar to our packing. This would
yield an essentially tight polynomial time $(3/2+\eps)$-approximation
 for SP and thus solve the long-standing open problem to
find the best possible polynomial time approximation ratio for \ts.
We leave this as an open question.
\subsection{Other related work}
In the 1980s, Baker et al.~\cite{baker1980orthogonal} initiated
the study of approximation algorithms for strip packing, by giving
a 3-approximation algorithm. After a sequence of improved approximations
 \cite{coffman1980performance,sleator19802}, Steinberg
\cite{steinberg1997strip} and Schiermeyer \cite{schiermeyer1994reverse}
independently gave 2-approximation algorithms. 
For asymptotic approximation,
Kenyon and R{é}mila \cite{kenyon2000near} settled \ts~by
providing an APTAS.

\ts~has rich connections with  important geometric
packing problems \cite{CKPT17,khan2015approximation} such as 2D bin
packing (2BP) \cite{bansal2014binpacking,KS21}, 2D geometric knapsack (2GK) \cite{GalvezGIHKW21,Jansen2004},
dynamic storage allocation \cite{BuchsbaumKKRT04}, maximum independent
set of rectangles (MISR) \cite{Galvez22,AHW19}, sliced packing \cite{DeppertJ0RT21,Galvez0AK21},
etc.

In 2BP, we are given a set of rectangles and square bins, and the goal is to find an axis-aligned
non-overlapping packing of all items into a minimum number of bins.
The problem admits no APTAS \cite{bansal2006bin}, and
the present best approximation ratio is $1.406$ \cite{bansal2014binpacking}.
In 2GK, we are given a set of
rectangular items and a square knapsack. Each item has an associated
profit, and the goal is to pack a subset of items in the knapsack
such that the profit is maximized. The present best polynomial time
approximation ratio is $1.89$ \cite{GalvezGIHKW21}. There is a pseudo-polynomial
time $(4/3+\eps)$-approximation \cite{Galvez00RW21} for
2GK.
In MISR, we are given a set of (possibly overlapping) rectangles we need to find the
maximum cardinality non-overlapping set of rectangles. Recently, Mitchell~\cite{mitchellFocs}
gave the first constant approximation algorithm for the problem. Then
G{á}lvez et al. \cite{Galvez22} obtained a $(2+\eps)$-approximation
algorithm for MISR. Their algorithms are based on a recursive geometric
decomposition of the plane, which can be viewed as a generalization
of guillotine cuts, more precisely, to cuts with $O(1)$ bends. Pach and Tardos \cite{pach2000cutting} even conjectured that for any
set of $n$ non-overlapping axis-parallel rectangles, there is a guillotine
cutting sequence separating $\Omega(n)$ of them.

%

2BP and 2GK are also well-studied in the guillotine setting \cite{pietrobuoni2015two}.
Caprara et al.~\cite{caprara2005fast} gave an APTAS for 2-stage
\ts~and 2-stage BP. 
Later, Bansal et al.~\cite{BansalLS05} showed an
APTAS for guillotine 2BP. Bansal et al.~\cite{bansal2014binpacking}
conjectured that the worst-case ratio between the best guillotine
2BP and the best general 2BP is $4/3$. If true, this would imply
a $(\frac{4}{3}+\eps)$-approximation algorithm for 2BP. %
For guillotine 2GK, Khan et al.~\cite{khan2021guillotine} recently
gave a pseudo-polynomial time approximation scheme. 
\section{\label{sec:PPTAS}Pseudo-polynomial time approximation scheme}

In this section, we present our pseudo-polynomial time approximation scheme (PPTAS) for~\tsg. 

Let $\eps>0$ and assume w.l.o.g.~that $1/\eps\in\N$. We
denote by $\opt$ the height of the optimal solution. We classify
the input items into a few groups according to their heights
and widths similar to the classification in \cite{khan2021guillotine}. For two constants $1\geq\dlta>\meu>0$ to be defined later,
we classify each item $i\in\R$ as: 
\begin{itemize}
\item \textit{tall} if $h_{i}>{\opt}/{2}$;
\item {\em large} if $\width_{i}>\dlta\W$ and ${\opt}/{2}\geq\height_{i}>\dlta\opt$; 
\item {\em horizontal} if $\width_{i}>\dlta\W$ and $\height_{i}\leq\meu\opt$; 
\item {\em vertical} $\width_{i}\leq\delta\W$ and ${\opt}/{2}\geq\height_{i}>\dlta\opt$; 
\item {\em medium if} 
\begin{itemize}
\item either $\dlta\opt\geq\height_{i}>\meu\opt$; 
\item or $\dlta\W\geq\width_{i}>\meu\W$ and $\height_{i}\leq\meu\opt$;
\end{itemize}
\item {\em small } if $\width_{i}\leq\meu\W$ and $\height_{i}\leq\meu\opt$; 
\end{itemize}

\begin{figure}[!thb]
		\centering
		\resizebox{!}{4.5cm}{
		\begin{tikzpicture}		
			\draw (0,0) rectangle (6,6);
			\draw [thick, pattern=north west lines, pattern color=black] (4,2) rectangle (6,3);			
			\draw [thick, pattern=north east lines, pattern color=red] (0,2) rectangle (4,3);			
			\draw [thick, pattern=north west lines] (0,0) rectangle (2,1);			
			\draw [thick, pattern=north west lines, pattern color=Brown] (4,0) rectangle (6,1);				
			\draw [thick, pattern=north west lines, pattern color=yellow] (0,1)--(0,2)--(6,2)--(6,1)-
			-(4,1)--(4,0)--(2,0)--(2,1)--(0,1);			
			\draw [thick, pattern=north west lines, pattern color=OliveGreen] (0,3) rectangle (6,6);
			
			\draw[ultra thick] (6,0)--(0,0)--(0,6);
			\draw[ultra thick] (-.15,1)--(.15,1);
			\draw[ultra thick] (-.15,2)--(.15,2);
                                \draw[ultra thick] (-.15,3)--(.15,3);
			
			\draw[ultra thick] (2,-.15)--(2,.15);
			\draw[ultra thick] (4, -.15)--(4,.15);			
			
			\draw (5,2.5) node {\large \textbf{$Large$}};				
			\draw (1,0.5) node {\large \textbf{$Small$}};				
			\draw (2,2.5) node {\large \textbf{$Vertical$}};				
			\draw (5,0.5) node {\large \textbf{$Horizontal$}};				
			\draw (3,1.5) node {\large \textbf{$Medium$}};
                                \draw (3,4.5) node {\large \textbf{$Tall$}};	
			
			\draw (-.75,1) node {\Large $\boldsymbol{\meu\textbf{OPT}}$};				
			\draw (-.75,2) node {\Large $\boldsymbol{\dlta\textbf{OPT}}$};
 			\draw (-.75,3) node {\Large $\boldsymbol{\frac{1}{2}\textbf{OPT}}$};	
			\draw (2,-.3) node {\Large $\boldsymbol{\meu W}$};				
			\draw (4,-.3) node {\Large $\boldsymbol{\dlta W}$};				
					
			\draw (6,-.3) node {\Large $\boldsymbol{W}$};				
			\draw (-.75,6) node {\Large $\boldsymbol{\textbf{OPT}}$};				
			
		\end{tikzpicture}}
		\caption{Item Classification: x-axis represents width 
and  y-axis represents height.} 
		\label{fig:classif}
	\end{figure}
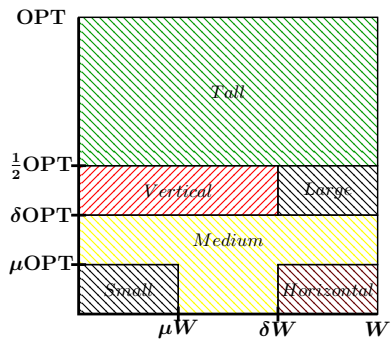

See Figure~\ref{fig:classif}  for a picture of item classification. Let $\Rit,\Rbg, \Rho,\Rve,\\\Rme,\Rsm$ be the set of tall, large, horizontal,
medium, and small rectangles in $\R$, respectively. 

Using the following lemma, one can appropriately choose $\mu,\delta$ such that the medium items
occupy a  marginal area. This effectively allows us to ignore them
in our main argumentation. 
\begin{lemma}[\cite{nadiradze2016approximating}]
\label{class_1} Let $\eps>0$
and $f(.)$ be any positive increasing function such that $f(x)<x$
for all $x\in(0,1]$. Then we can efficiently find $\dlta,\meu\in\Omega_{\eps}(1)$,
with $\eps\ge f(\eps)\ge\dlta\ge f(\dlta)\ge\meu$ so that the total
area of medium rectangles is at most $\eps(\opt\cdot\W)$. 
\end{lemma}
We will specify how we choose the function $f(x)$ later.
In our PPTAS, we will use a packing, which is defined solely via boxes. 

\begin{definition}
A \emph{box }$B$ is an axis-aligned open rectangle that satisfies
$B\subseteq[0,\W]\times[0,\infty)$. We denote by $h(B)$ and $w(B)$
the height and the width of $B$, respectively.
\end{definition}
Inside each box $B$, we will place the items \emph{nicely, }meaning
that they are either stacked horizontally or vertically,
or $B$ contains a single large item, or only small items, or only medium items. This is useful since in the first
two cases, it is trivial to place a given set of items into $B$, and
in the last two cases, it will turn out that it suffices to pack the items
greedily using the Next Fit Decreasing Height (NFDH) algorithm \cite{coffman1980performance}
and Steinberg's algorithm \cite{steinberg1997strip}, respectively.
There will be one box with height at most $2\eps\opt$
that contains all medium items.

\begin{definition}
[Nice packing] \label{def:structured-boxes} Let $B$ be a box and
let $\R_{B}\subseteq\R$ be a set of items that are placed non-overlappingly
inside $B$. We say that the packing of $\R_{B}$ in $B$ is \emph{nice}
if the items in $\R_{B}$ are guillotine separable and additionally 

\begin{itemize}
\item $\R_{B}$ contains only one item, or 
\item $\R_{B}\subseteq\Rho$ and the items in $\R_{B}$ are stacked on top
of each other inside $B$, or 
\item $\R_{B}\subseteq\Rit\cup\Rve$ and the items in $\R_{B}$ are placed
side by side inside $B$, or 
\item $\R_{B}\subseteq\Rme$, or 
\item $\R_{B}\subseteq\Rsm$ and for each item $i\in\R_{B}$ it holds that
$w_{i}\le\eps\cdot w(B)$ and $h_{i}\le\eps\cdot h(B)$.
\end{itemize}
\end{definition}
We will use the term \emph{container} to refer to a box $B$ that contains a nice packing of some  set of items $I_B$.
See Figure~\ref{fig_skew_nice} for nice packings in different types of containers.
We say that a set of boxes $\B$ is \emph{guillotine separable} if
there exists a sequence of guillotine cuts that separates them and
that does not intersect any box in $\B$. 

We now state the structural lemma for the PPTAS. Intuitively, it states
that there exists a $(1+\eps)$-approximate solution in which
the input items are placed into $O_{\eps}(1)$ boxes
such that within each box the packing is nice.
 We remark that we will crucially use that in the optimal packing the
items in $\R$ are guillotine separable. In fact, if one could prove
that there exists such a packing with $O_{\eps}(1)$ boxes and
a height of $\alpha\opt$ for some $\alpha<\frac{5}{4}$ also in the
non-guillotine case (where neither the optimal solution nor the
computed solution needs to  be guillotine separable), then one would
obtain a pseudo--polynomial time $(\alpha+\eps)$-approximation
algorithm also in this case, by using 
straightforward adaptations of the algorithms in, e.g.,
\cite{nadiradze2016approximating, GGIK16, JansenR19} %
or our algorithm in section \ref{subsec:algo-PPTAS}. 
However,
this is not possible for $\alpha<\frac{5}{4}$, unless~$\mathsf{P=NP}$
\cite{henning2020complexity}.

\begin{lemma}[Structural lemma $1$]
\label{lem:structural}
Assume that $\mu$ is sufficiently small compared to $\delta$. Then
there exists a set $\B$ of $O_{\eps}(1)$
pairwise non-overlapping and guillotine separable boxes all
placed inside $[0,\W]\times[0,(1+16\eps)\opt)$ and a partition
$\R=\bigcup_{B\in\B}\R_{B}$ such that for each $B\in\B$ the items
in $\R_{B}$ can be placed nicely into $B$.
\end{lemma}
We choose our function $f$ due to Lemma~\ref{class_1} such that $\mu$ is sufficiently small compared to $\delta$, as required by Lemma~\ref{lem:structural}.
We will prove Lemma~\ref{lem:structural} in the next subsection. In its packing, let $\B_{hor}, \B_{ver}, \B_{tall}, \B_{large}, \B_{small}$ and $\B_{med}$ denote the set of boxes for the horizontal, vertical, tall, large, small and medium items, respectively. Let $\B_{tall+ver}:=\B_{tall}\cup\B_{ver}$. 

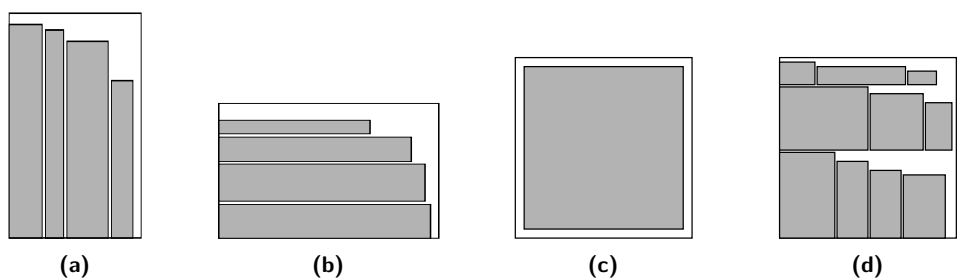
\begin{figure}[t]
\centering
\captionsetup[subfigure]{justification=centering}
\begin{subfigure}[b]{0.2\linewidth}
\centering
\resizebox{1.8cm}{3cm}{
\begin{tikzpicture}
\draw [ultra thick ] (0,0) rectangle (8,8);
\draw [ultra thick,fill=gray!60] (0,0) rectangle (2,7.6);
\draw [ultra thick, fill=gray!60] (2.2,0) rectangle (3.3,7.4);
\draw [ultra thick, fill=gray!60] (2.2,0) rectangle (3.3,7.4);
\draw [ultra thick, fill=gray!60] (3.5,0) rectangle (6,7);
\draw [ultra thick, fill=gray!60] (6.2,0) rectangle (7.5,5.6);
\end{tikzpicture}
}
\caption{}
\label{subfig_vert_nice}
\end{subfigure}
\begin{subfigure}[b]{0.26\linewidth}
\centering
\resizebox{3cm}{1.8cm}{
\begin{tikzpicture}
\draw [ultra thick ] (0,0) rectangle (8,8);
\draw [ultra thick,fill=gray!60] (0,0) rectangle (7.7,2);
\draw [ultra thick, fill=gray!60] (0,2.2) rectangle (7.5,4.4);
\draw [ultra thick, fill=gray!60] (0,4.55) rectangle (7,6);
\draw [ultra thick, fill=gray!60] (0,6.2) rectangle (5.5,7);
\end{tikzpicture}
}
\caption{}
\label{subfig_hor_nice}
\end{subfigure}
\begin{subfigure}[b]{0.24\linewidth}
\centering
\resizebox{2.41cm}{2.41cm}{
\begin{tikzpicture}
\draw [ultra thick ] (0,0) rectangle (8,8);
\draw [ultra thick,fill=gray!60] (0.4,0.4) rectangle (7.6,7.6);
\end{tikzpicture}
}
\caption{}
\label{subfig_large_nice}
\end{subfigure}
\begin{subfigure}[b]{0.24\linewidth}
\centering
\resizebox{2.41cm}{2.41cm}{
\begin{tikzpicture}
\draw [ultra thick ] (0,0) rectangle (8,8);
\draw [ultra thick,fill=gray!60] (0,0) rectangle (2.5,3.8);
\draw [ultra thick, fill=gray!60] (2.6,0) rectangle (4,3.4);
\draw [ultra thick, fill=gray!60] (4.1,0) rectangle (5.5,3);
\draw [ultra thick, fill=gray!60] (5.6,0) rectangle (7.5,2.8);
\draw [ultra thick, fill=gray!60] (0,3.9) rectangle (4,6.7);
\draw [ultra thick, fill=gray!60] (4.1,3.9) rectangle (6.5,6.4);
\draw [ultra thick, fill=gray!60] (6.6,3.9) rectangle (7.8,6);
\draw [ultra thick, fill=gray!60] (0,6.8) rectangle (1.6,7.8);
\draw [ultra thick, fill=gray!60] (1.7,6.8) rectangle (5.7,7.6);
\draw [ultra thick, fill=gray!60] (5.8,6.8) rectangle (7.1,7.4);
\end{tikzpicture}
}
\caption{}
\label{subfig_vert_nice}
\end{subfigure}
\caption{Nice packing of vertical, horizontal, large and small items in their respective containers }
\label{fig_skew_nice}
\end{figure}

\subsection{Proof of  Structural Lemma 1}
\label{subsec_proof_of_struc_lem}

In this section we prove Lemma~\ref{lem:structural}. We have omitted a few proofs due to space constraints which can be found in Appendix \ref{sec:algo_new}. Our strategy
is to start with a structural lemma from \cite{khan2021guillotine} that guarantees the
existence of a structured packing of all items in $\Rhard:=\Rit\cup\Rbg\cup\Rho\cup\Rve$. 
This packing uses boxes and $\Lc$-compartments. Note that, for now we ignore the items $\Rsm$. We will show how to pack them later.
\begin{definition}[$\Lc$-compartment]
\label{def:Lc} An $\Lc$-compartment $L$
is an open sub-region of $[0,\W]\times[0,\infty)$ bounded by a simple
rectilinear polygon with six edges $e_{0},e_{1},\dots,e_{5}$ such
that for each pair of horizontal (resp. vertical) edges $e_{i},e_{6-i}$
with $i\in\{1,2\}$ there exists a vertical (resp. horizontal) line
segment $\ell_{i}$ of length less than $\dlta\frac{\opt}{2}$ (resp.
$\dlta\frac{\W}{2}$) such that both $e_{i}$ and $e_{6-i}$ intersect
$\ell_{i}$ but no other edges intersect $\ell_{i}$. 
\end{definition}

Note that for an $\Lc$-compartment, no item $i\in \Rho$ can be packed in its vertical arm and similarly, no item $i\in\Rve\cup I_{tall}$ can be packed in its horizontal arm.
 

The next lemma follows immediately from a structural insight in \cite{khan2021guillotine}
for the guillotine two-dimensional knapsack problem. It partitions
the region $[0,\W]\times[0,\opt]$ into non-overlapping boxes and
$\Lc$-compartments that admit a \emph{pseudo-guillotine cutting sequence}.
This is a sequence of cuts in which each cut is either a (normal)
guillotine cut, or a special cut that cuts out an $\Lc$-compartment
$L$ from the current rectangular piece $R$ in the cutting sequence,
such that $R\setminus L$ is a rectangle, see 
 Figure~\ref{fig:pseudonice}.
 So intuitively
$L$ lies at the boundary of $R$.

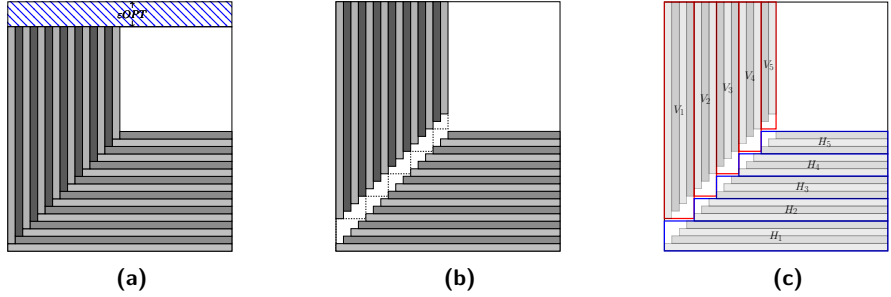
\begin{figure}[h]
\centering
\captionsetup[subfigure]{justification=centering}
\begin{subfigure}[b]{0.3\linewidth}
\resizebox{2.975cm}{3.315cm}{
\begin{tikzpicture}
\draw [very thick] (0,0) rectangle (18,18);
\draw [very thick,pattern=north west lines, pattern color=blue] (0,18) rectangle (18,20);
\draw [fill=gray!50, very thick] (0,0) rectangle (18,0.6);
\draw [fill=black!30, very thick] (0,0.6) rectangle (0.6,18);
\draw [fill=gray!90, very thick] (0.6,0.6) rectangle (18,1.2);
\draw [fill=black!65, very thick] (0.6,1.2) rectangle (1.2,18);
\draw [fill=gray!50, very thick] (1.2,1.2) rectangle (18,1.8);
\draw [fill=black!30, very thick] (1.2,1.8) rectangle (1.8,18);
\draw [fill=gray!90, very thick] (1.8,1.8) rectangle (18,2.4);
\draw [fill=black!65, very thick] (1.8,2.4) rectangle (2.4,18);
\draw [fill=gray!50, very thick] (2.4,2.4) rectangle (18,3);
\draw [fill=black!30, very thick] (2.4,3) rectangle (3,18);
\draw [fill=gray!90, very thick] (3,3) rectangle (18,3.6);
\draw [fill=black!65, very thick] (3,3.6) rectangle (3.6,18);
\draw [fill=gray!50, very thick] (3.6,3.6) rectangle (18,4.2);
\draw [fill=black!30, very thick] (3.6,4.2) rectangle (4.2,18);
\draw [fill=gray!90, very thick] (4.2,4.2) rectangle (18,4.8);
\foreach \x in {1,...,4} {\draw [fill=black!65, very thick](3+\x*1.2,3.6+\x *1.2 ) rectangle (3.6+\x *1.2,18);
\draw [fill=gray!50, very thick](3.6+\x*1.2,3.6+\x*1.2 ) rectangle (18,4.2+\x *1.2);
\draw [fill=black!30, very thick](3.6+\x*1.2,4.2+\x*1.2 ) rectangle (4.2+\x *1.2,18);
\draw [fill=gray!90, very thick](4.2+\x*1.2,4.2+\x*1.2 ) rectangle (18,4.8+\x *1.2);
\draw (10,19) node {\Huge \textbf{$\varepsilon OPT$}} ;
}
\draw[ultra thick,->] (10,19.5)--(10,20);
\draw[ultra thick,->] (10,18.5)--(10,18);
\end{tikzpicture}}
\caption{\hspace*{2.5em}\label{fig_Lcomp}}
\end{subfigure}
\begin{subfigure}[b]{0.3\linewidth}
\resizebox{2.975cm}{3.315cm}{
\begin{tikzpicture}
\draw [very thick] (0,0) rectangle (18,20);
\draw[ultra thick,dashed] (0,0)--(0,2.6)--(2.4,2.6)--(2.4,4.4)--(4.2,4.4)--(4.2,6.2)--(6,6.2)--(6,8)--(7.8,8)--(7.8,9.8)--(9,9.8)--(9,11);
\draw [fill=gray!50, very thick] (0,0) rectangle (18,0.6);
\draw [fill=black!30, very thick] (0,2.6) rectangle (0.6,20);
\draw [fill=gray!90, very thick] (0.6,0.6) rectangle (18,1.2);
\draw [fill=black!65, very thick] (0.6,3.2) rectangle (1.2,20);
\draw [fill=gray!50, very thick] (1.2,1.2) rectangle (18,1.8);
\draw [fill=black!30, very thick] (1.2,3.8) rectangle (1.8,20);
\draw [fill=gray!90, very thick] (1.8,1.8) rectangle (18,2.4);
\draw [fill=black!65, very thick] (1.8,4.4) rectangle (2.4,20);
\draw [fill=gray!50, very thick] (2.4,2.4) rectangle (18,3);
\draw [fill=black!30, very thick] (2.4,5) rectangle (3,20);
\draw [fill=gray!90, very thick] (3,3) rectangle (18,3.6);
\draw [fill=black!65, very thick] (3,5.6) rectangle (3.6,20);
\draw [fill=gray!50, very thick] (3.6,3.6) rectangle (18,4.2);
\draw [fill=black!30, very thick] (3.6,6.2) rectangle (4.2,20);
\draw [fill=gray!90, very thick] (4.2,4.2) rectangle (18,4.8);
\foreach \x in {1,...,4} {\draw [fill=black!65, very thick](3+\x*1.2,5.6+\x *1.2 ) rectangle (3.6+\x *1.2,20);
\draw [fill=gray!50, very thick](3.6+\x*1.2,3.6+\x*1.2 ) rectangle (18,4.2+\x *1.2);
\draw [fill=black!30, very thick](3.6+\x*1.2,6.2+\x*1.2 ) rectangle (4.2+\x *1.2,20);
\draw [fill=gray!90, very thick](4.2+\x*1.2,4.2+\x*1.2 ) rectangle (18,4.8+\x *1.2);
}
\end{tikzpicture}}
\caption{\hspace*{2.5em}}
\end{subfigure}
\begin{subfigure}[b]{0.3\linewidth}
\resizebox{2.975cm}{3.315cm}{
\begin{tikzpicture}
\draw [very thick] (0,0) rectangle (18,20);
\draw[line width=0.1 cm, red](0,2.6)  rectangle (2.4,20);
\draw[line width=0.1 cm, red](2.4,4.4)  rectangle (4.2,20);
\draw[line width=0.1 cm, red](4.2,6.2)  rectangle (6,20);
\draw[line width=0.1 cm, red](6,8)  rectangle (7.8,20);
\draw[line width=0.1 cm, red](7.8,9.8)  rectangle (9,20);
\draw[line width=0.1 cm, blue](0,0)  rectangle (18,2.4);
\draw[line width=0.1 cm, blue](2.4,2.4)  rectangle (18,4.2);
\draw[line width=0.1 cm, blue](4.2,4.2)  rectangle (18,6);
\draw[line width=0.1 cm, blue](6,6)  rectangle (18,7.8);
\draw[line width=0.1 cm, blue](7.8,7.8)  rectangle (18,9.6);

\draw [fill=gray!50, very thin,opacity=0.3] (0,0) rectangle (18,0.6);
\draw [fill=black!30, very thin,opacity=0.3] (0,2.6) rectangle (0.6,20);
\draw [fill=gray!90, very thin,opacity=0.3] (0.6,0.6) rectangle (18,1.2);
\draw [fill=black!65, very thin,opacity=0.3] (0.6,3.2) rectangle (1.2,20);
\draw [fill=gray!50, very thin,opacity=0.3] (1.2,1.2) rectangle (18,1.8);
\draw [fill=black!30, very thin,opacity=0.3] (1.2,3.8) rectangle (1.8,20);
\draw [fill=gray!90, very thin,opacity=0.3] (1.8,1.8) rectangle (18,2.4);
\draw [fill=black!65, very thin,opacity=0.3] (1.8,4.4) rectangle (2.4,20);
\draw [fill=gray!50, very thin,opacity=0.3] (2.4,2.4) rectangle (18,3);
\draw [fill=black!30, very thin,opacity=0.3] (2.4,5) rectangle (3,20);
\draw [fill=gray!90, very thin,opacity=0.3] (3,3) rectangle (18,3.6);
\draw [fill=black!65, very thin,opacity=0.3] (3,5.6) rectangle (3.6,20);
\draw [fill=gray!50, very thin,opacity=0.3] (3.6,3.6) rectangle (18,4.2);
\draw [fill=black!30, very thin,opacity=0.3] (3.6,6.2) rectangle (4.2,20);
\draw [fill=gray!90, very thin,opacity=0.3] (4.2,4.2) rectangle (18,4.8);
\foreach \x in {1,...,4} {\draw [fill=black!65, very thin,opacity=0.3](3+\x*1.2,5.6+\x *1.2 ) rectangle (3.6+\x *1.2,20);
\draw [fill=gray!50, very thin,opacity=0.3](3.6+\x*1.2,3.6+\x*1.2 ) rectangle (18,4.2+\x *1.2);
\draw [fill=black!30, very thin,opacity=0.3](3.6+\x*1.2,6.2+\x*1.2 ) rectangle (4.2+\x *1.2,20);
\draw [fill=gray!90, very thin,opacity=0.3](4.2+\x*1.2,4.2+\x*1.2 ) rectangle (18,4.8+\x *1.2);
}
\draw (1.2,11.3) node {\Huge \textbf{$V_1$}};
\draw (3.3,12.2) node {\Huge \textbf{$V_2$}};
\draw (5.1,13.1) node {\Huge \textbf{$V_3$}};
\draw (6.9,14) node {\Huge \textbf{$V_4$}};
\draw (8.4,14.9) node {\Huge \textbf{$V_5$}};
\draw (9,1.2) node {\Huge \textbf{$H_1$}};
\draw (10.2,3.3) node {\Huge \textbf{$H_2$}};
\draw (11.1,5.1) node {\Huge \textbf{$H_3$}};
\draw (12,6.9) node {\Huge \textbf{$H_4$}};
\draw (12.9,8.7) node {\Huge \textbf{$H_5$}};
\end{tikzpicture}}
\caption{\hspace*{2.5em}}
\end{subfigure}
\caption{Using an extra $\eps \text{OPT}$ height, we convert a packing of items $I$ in an $\Lc$-compartment into another packing such that the items in $I$ are packed in boxes $\mathcal{C}^{\prime}={\Vc}\cup {\Hc}$, which are guillotine separable and $|\mathcal{C}^{\prime}|=O_{\eps}(1)$, where ${\Vc}=\cup_{i=1}^{i=5} V_i$ and ${\Hc}=\cup_{i=1}^{i=5} H_i$}
\label{fig_rayshoot}
\end{figure}

\begin{lemma}
[\cite{khan2021guillotine}] \label{lem:rearrange}There exists
a partition of $[0,\W]\times[0,\opt]$ into a set $\B_1$ of $O_{\eps}(1)$
boxes and a set $\L$ of $O_{\eps}(1)$ $\Lc$-compartments such
that 
\begin{itemize}
\item the boxes and $\Lc$-compartments in $\B_1\cup\L$ are pairwise non-overlapping,
\item $\B_1\cup\L$ admits a pseudo-guillotine cutting sequence,
\item the items in $\Rhard$ can be packed into $\B_1\cup\L$ such that for
each $B\in\B_1$ it either contains only items $i\in\Rit\cup\Rbg\cup\Rve$ or it contains only items  $i\in \Rho$.
\end{itemize}
\end{lemma}
Our strategy is to take the packing due to Lemma~\ref{lem:rearrange}
and transform it step by step until we obtain a packing that corresponds
to Lemma~\ref{lem:structural}. First, we round the heights of the
tall, large, and vertical items such that they are integral multiples
of $\delta^{2}\opt$. Formally, for each item $i\in\Rit\cup\Rbg\cup\Rve$
we  round its height to $\height'_{i}:=\left\lceil \frac{\height_{i}}{\delta^{2}\opt}\right\rceil \delta^{2}\opt$.
Let $\Rhard$ denote the resulting set of items.
 By a shifting argument, we will show that we can still pack $\Rhard$
into $O_{\eps}(1)$ guillotine separable boxes and $\Lc$-compartments
if we can increase the height of the packing by a factor $1+\eps$ which also does not violate guillotine separability.
Then, we increase the height of the packing by another factor $1+\eps$.
Using this additional space, we shift the items inside each $\Lc$-compartment
$L$ such that we can separate the vertical items from the horizontal items
(see Figure \ref{fig_rayshoot}). Due to this separation, we can partition $L$ into
$O_{\eps}(1)$ boxes such that each box contains only horizontal
or only vertical and tall items. Note however, that they might not be packed
nicely inside these boxes. 
\begin{lemma}
\label{lem:rearrange-2}There exists a partition of $[0,\W]\times[0,(1+2\eps)\opt]$
into a set $\B_2$ of $O_{\eps}(1)$ boxes such that 
\begin{itemize}
\item the boxes in $\B_2$ are pairwise non-overlapping and admit a guillotine
cutting sequence,
\item the items in $\Rhard$ can be packed into $\B_2$ such that they are
guillotine separable and each box $B\in\B_2$ either  contains only items from $\Rit\cup\Rbg\cup\Rve$, or contains only items from $\Rho$.
\item Any item $i\in \Rit\cup\Rbg\cup\Rve$ has height $h_i'=k_i \delta^2 \opt$ for integer $k_i$, $k_i\leq 1/\delta^2+1$. 
\end{itemize}
\end{lemma}
Let $\B_{2}$ be the set of boxes due to Lemma~\ref{lem:rearrange-2}.
Consider a box $B\in \B_{2}$ and let $\Rhard(B)$ denote the items from
$\Rhard$ that are placed inside $B$ in the packing due to Lemma~\ref{lem:rearrange-2}.
Our goal is to partition $B$ into $O_{\eps}(1)$ smaller containers, i.e.,
 the items in $\Rhard(B)$ are packed nicely into these
smaller boxes. If $B$ contains horizontal items, then this can be
done using standard techniques, e.g., by 1D resource augmentation (only in height) in {\cite{khan2021guillotine}}. This resource augmentation procedure maintains guillotine separability (see Appendix~\ref{subsec_resource_aug}).
\begin{lemma}[\cite{khan2021guillotine}]
\label{lem:hor-boxes}Given a box $B\in\B_2$ such that $B$ contains
a set of items $\Rhard(B)\subseteq\Rho$. There exists a partition
of $B$ into $O_{\eps'}(1)$ containers $\B'$ and one additional box
$B'$ of height at most $\eps' h(B)$ and width $w(B)$ such
that
 the containers $\B'$ are guillotine separable
 and the containers $\B' \cup \{B'\}$ 
  contain $\Rhard(B)$.
\end{lemma}

\begin{figure}[t]
	\captionsetup[subfigure]{justification=centering}
	\hspace{-10pt}
\centering
	\begin{subfigure}[b]{.47\textwidth}
		\centering
		\scalebox{-1}[1]{\resizebox{!}{4.5cm}{
		\begin{tikzpicture}			
			\draw[thick] (0,0) rectangle (6,6);	

			\draw[ultra thick] (3,-0.2)--(3,6.05);
			\draw (3,-.5) node {\scalebox{-1}[1]{\large \textbf{$\ell_1$}}};

			\draw[ultra thick] (-.2,3.2)--(3,3.2);
			\draw (-.4, 3.2) node {\scalebox{-1}[1]{\large \textbf{$\ell_6$}}};

			\draw[ultra thick] (0, .5)--(2.5,.5)--(2.5,3.2);
			\draw (2.3,.7) node {\scalebox{-1}[1]{\large \textbf{$\ell_7$}}};

			\draw[ultra thick] (3,2)--(6.2,2);
			\draw (6.5,2) node {\scalebox{-1}[1]{\large \textbf{$\ell_2$}}};

			\draw[ultra thick] (3.7, 6)--(3.7,2.7)--(6,2.7);

			\draw (3.9,2.9) node {\scalebox{-1}[1]{\large \textbf{$\ell_3$}}};

			\draw[ultra thick] (3.7,4.5)--(6.2, 4.5);
			\draw (6.5,4.5) node {\scalebox{-1}[1]{\large \textbf{$\ell_4$}}};

	\end{tikzpicture}}}
	\label{fig:pseudocut_1}
\caption{}
	\end{subfigure}
	\begin{subfigure}[b]{.47\textwidth}
		\centering
		\scalebox{-1}[1]{\resizebox{!}{7.5cm}{
		\begin{tikzpicture}
			\draw[thick] (-3, -3) rectangle (3,3);	
			\draw[dashed] (0,-3)--(0,3);
			\draw (0,0) node {\scalebox{-1}[1]{\Huge \textbf{$\ell_1$}}};
			\draw[-{Latex[length=7mm, width=6mm]}](0,-3.2)--(-6, -4.9);
			\draw[-{Latex[length=7mm, width=6mm]}](0,-3.2)--(6, -4.9);

			\draw[thick] (-7.5,-5) rectangle (-4.5, -11);
			\draw[dashed] (-7.5,-7.8)--(-4.5,-7.8);
			\draw (-6,-7.8) node {\scalebox{-1}[1]{\Huge \textbf{$\ell_6$}}};
			\draw[-{Latex[length=7mm, width=6mm]}](-6,-11.2)--(-4.5, -12.9);
			\draw[-{Latex[length=7mm, width=6mm]}](-6,-11.2)--(-8.5, -12.9);

			\draw[thick] (4.5, -5) rectangle (7.5, -11);
			\draw[dashed] (7.5,-9)--(4.5,-9);
			\draw (6,-9) node {\scalebox{-1}[1]{\Huge \textbf{$\ell_2$}}};
			\draw[-{Latex[length=7mm, width=6mm]}](6,-11.2)--(8.5, -12.9);
			\draw[-{Latex[length=7mm, width=6mm]}](6,-11.2)--(5, -12.9);

			\draw[thick] (-6,-13) rectangle (-3, -16.2);
			\draw[-{Latex[length=7mm, width=6mm]}](-4.5,-16.4)--(-6.25, -18.9);
			\draw[-{Latex[length=7mm, width=6mm]}](-4.5,-16.4)--(-3, -18.9);
			\draw[dashed](-6,-16.2)--(-6, -15.7)--(-3.5, -15.7)--(-3.5, -13);
			\draw (-3.7,-15.5) node {\scalebox{-1}[1]{\Huge \textbf{$\ell_7$}}};

			\draw[thick] (-10,-13) rectangle (-7, -15.8);

			\draw[thick] (6,-13) rectangle (3, -15);

			\draw[thick] (10,-13) rectangle (7, -17);
			\draw[-{Latex[length=7mm, width=6mm]}](8.5,-17.2)--(5.5, -18.9);
			\draw[-{Latex[length=7mm, width=6mm]}](8.5,-17.2)--(11.15, -18.9);
			\draw[dashed](7,-13)--(7.7, -13)--(7.7, -16.3)--(10, -16.3);
			\draw (7,-12.5) node {\scalebox{-1}[1]{\Huge \textbf{$\ell_3$}}};

			\draw[thick] (-7.5, -19)rectangle(-5,-21.5);
			\draw[thick] (-.5, -19)--(-.5,-22.2)--(-3.5,-22.2)--(-3.5,-21.7)--(-1, -21.7)--(-1,-19)--(-.5,-19);

			\draw[thick] (4, -19)--(4, -23)--(7, -23)--(7, -22.3)--(4.7, -22.3)--(4.7,-19)--(4,-19);

			\draw[thick] (10, -19) rectangle (12.3, -22.3);
			\draw[dashed](10,-20.5)--(12.3, -20.5);
			\draw (11.15,-20.5) node {\scalebox{-1}[1]{\Huge \textbf{$\ell_4$}}};



	\end{tikzpicture}}}
	\label{fig:tree}
	\caption{}
	\end{subfigure}
	\caption{(a) A pseudo-guillotine cutting sequence. The first cut is $l_1$, and then the resulting left piece is further subdivided by $\ell_2$, $\ell_3$ and $\ell_4$. Similarly, $\ell_6$, $\ell_7$ subdivide the right piece. Note that $\ell_3$ and $\ell_7$ are not guillotine cuts, but they cut out the corresponding $\Lc$-compartments. (b) Step by step pseudo-guillotine cutting sequence corresponding to Figure (a). Dashed line at each level indicates a partition of a rectangle into two regions (two boxes, or one box and one $\Lc$).}
	\label{fig:pseudonice}
\end{figure}
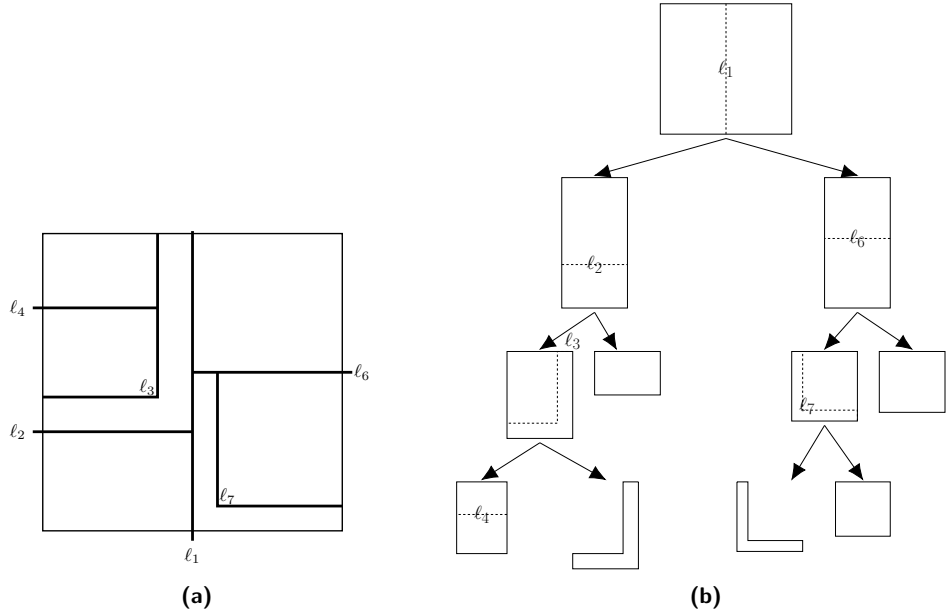

We apply Lemma~\ref{lem:hor-boxes} to each box $B\in\B_2$ that contains
a horizontal item. Consider the items which are contained in their respective boxes $B'$. In order to avoid any confusions between constants of our algorithm and resource augmentation, we denote the constant used for resource augmentation as $\eps'$. We choose $\eps'=\eps$ and then their
total area is at most $\eps\opt\cdot W$ and therefore,  all such items can be packed in a box  of height at most $2\eps\opt$ and width $W$ using Steinberg's algorithm~\cite{steinberg1997strip}. But since this will possibly not result in a nice packing we apply resource augmentation (only along height) again to ensure that we get a nice packing of such horizontal items in $O_{\eps}(1)$ containers which can all be packed in a box of height at most $3\eps\text{OPT}$ and width $W$ (see Section~\ref{subsec_resource_aug}).

Consider now a box $B\in\B_2$ that contains at least one item from
$\Rit\cup\Rbg\cup\Rve$. Let $\Rhard(B)\subseteq\Rit\cup\Rbg\cup\Rve$
denote the items packed inside $B$. We argue that we can rearrange
the items in $\Rhard(B)$ such that they are nicely placed inside
$O_{\eps}(1)$ containers. In this step we crucially use that the items
in $\Rhard(B)$ are guillotine separable. 

Consider the guillotine cutting sequence for $\Rhard(B)$. It is
useful to think of these cuts as being organized in \emph{stages:}
in the first stage we do vertical cuts (possibly zero cuts). In the
following stage, we take each resulting piece and apply horizontal cuts.
In the next stage, we again take each resulting piece and apply vertical
cuts, and so on.  Since the heights of the items in $\Rhard(B)$
are rounded to multiples of $\delta^2\opt$ we can assume w.l.o.g.~that
the $y$-coordinates of the horizontal cuts are all integral multiples
of $\delta^{2}\opt$ (possibly moving the items a little bit). Assume here for the sake of simplicity that $t=1/\delta^2$ is an integer. Because of the rounding of heights of the items in $I_{hard}'(B)$, there are at most $(1/\delta^2 -1)$  $y$-coordinates for making a horizontal cut. For a horizontal stage of cuts, for a rectangular piece we define a \emph{configuration vector} $(x_1,...,x_{t-1})$: For each $i\in[t-1]$ if there is a horizontal cut in the piece at $y=t\cdot i$, then $x_i=1$, otherwise $x_i=0$. Consider $y=0$ to be the bottom of the rectangular piece.
Therefore, in each horizontal stage, for each piece there are at most $K:=(2^{(1/\delta^{2})})$
 ~possible configurations. Consider
the first stage (which has vertical cuts). If there are more than
$K$ vertical cuts then in two of the resulting pieces, in the second
stage the same configuration of horizontal cuts is applied (see Figure \ref{fig_vert_cont}).

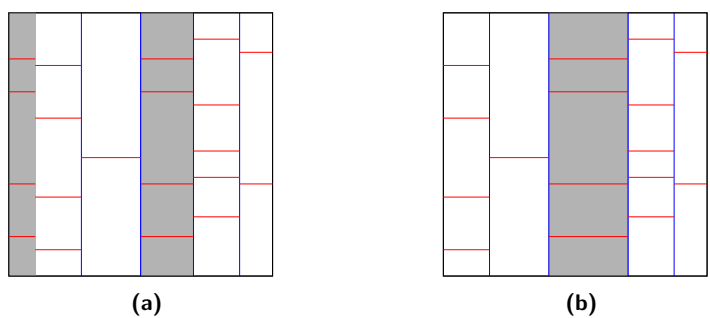
\begin{figure}[t]
\centering
\captionsetup[subfigure]{justification=centering}
\begin{subfigure}[t]{0.4\linewidth}
\centering
\resizebox{3.5cm}{3.5cm}{
\begin{tikzpicture}
\draw [ultra thick] (0,0) rectangle (20,20);
\draw [ultra thick,blue] (2,0) -- (2,20);
\draw [ultra thick,blue] (5.5,0) -- (5.5,20);
\draw [ultra thick,blue] (10,0) -- (10,20);
\draw [ultra thick,blue] (14,0) -- (14,20);
\draw [ultra thick,blue] (17.5,0) -- (17.5,20);
\fill [gray!60] (0,0) rectangle (2,20);
\fill [gray!60] (10,0) rectangle (14,20);
\draw [ultra thick,blue] (10,0) -- (10,20);
\draw [ultra thick,blue] (14,0) -- (14,20);
\draw [ultra thick] (0,0) rectangle (20,20);

\draw [ultra thick,red] (0,3) -- (2,3);
\draw [ultra thick,red] (0,7) -- (2,7);
\draw [ultra thick,red] (0,14) -- (2,14);
\draw [ultra thick,red] (0,16.5) -- (2,16.5);
\draw [ultra thick,red] (2,2) -- (5.5,2);
\draw [ultra thick,red] (2,6) -- (5.5,6);
\draw [ultra thick,red] (2,12) -- (5.5,12);
\draw [ultra thick,red] (2,16) -- (5.5,16);
\draw [ultra thick,red] (5.5,9) -- (10,9);

\draw [ultra thick,red] (10,3) -- (14,3);
\draw [ultra thick,red] (10,7) -- (14,7);
\draw [ultra thick,red] (10,14) -- (14,14);
\draw [ultra thick,red] (10,16.5) -- (14,16.5);
\draw [ultra thick,red] (14,4.5) -- (17.5,4.5);
\draw [ultra thick,red] (14,7.5) -- (17.5,7.5);
\draw [ultra thick,red] (14,9.5) -- (17.5,9.5);
\draw [ultra thick,red] (14,13) -- (17.5,13);
\draw [ultra thick,red] (14,18) -- (17.5,18);
\draw [ultra thick,red] (17.5,7) -- (20,7);
\draw [ultra thick,red] (17.5,17) -- (20,17);
\end{tikzpicture}}
\caption{\hspace*{-1em}}
\label{fig_vert_cont_1}
\end{subfigure}
\begin{subfigure}[t]{0.4\linewidth}
\centering
\resizebox{3.5cm}{3.5cm}{
\begin{tikzpicture}
\draw [ultra thick] (0,0) rectangle (20,20);
\draw [ultra thick,blue] (3.5,0) -- (3.5,20);
\draw [ultra thick,blue] (8,0) -- (8,20);
\draw [ultra thick,blue] (14,0) -- (14,20);
\draw [ultra thick,blue] (17.5,0) -- (17.5,20);
\fill [gray!60] (8,0) rectangle (14,20);
\draw [ultra thick,blue] (8,0) -- (8,20);
\draw [ultra thick,blue] (14,0) -- (14,20);
\draw [ultra thick] (0,0) rectangle (20,20);

\draw [ultra thick,red] (0,2) -- (3.5,2);
\draw [ultra thick,red] (0,6) -- (3.5,6);
\draw [ultra thick,red] (0,12) -- (3.5,12);
\draw [ultra thick,red] (0,16) -- (3.5,16);
\draw [ultra thick,red] (3.5,9) -- (8,9);

\draw [ultra thick,red] (8,3) -- (14,3);
\draw [ultra thick,red] (8,7) -- (14,7);
\draw [ultra thick,red] (8,14) -- (14,14);
\draw [ultra thick,red] (8,16.5) -- (14,16.5);
\draw [ultra thick,red] (14,4.5) -- (17.5,4.5);
\draw [ultra thick,red] (14,7.5) -- (17.5,7.5);
\draw [ultra thick,red] (14,9.5) -- (17.5,9.5);
\draw [ultra thick,red] (14,13) -- (17.5,13);
\draw [ultra thick,red] (14,18) -- (17.5,18);
\draw [ultra thick,red] (17.5,7) -- (20,7);
\draw [ultra thick,red] (17.5,17) -- (20,17);
\end{tikzpicture}}
\caption{\hspace*{-1em}}
\label{fig_vert_cont_2}
\end{subfigure}
\caption{(a) $2$ stages of guillotine cuts for a box containing vertical rectangles.\: (b) Since rounded heights of vertical rectangles are integral multiples of $\delta^2$, merge configurations with same set of horizontal cuts to get $O_\delta(1)$ configurations.}
\label{fig_vert_cont}
\end{figure}

We reorder the resulting pieces and their items such that pieces with
the same configuration of horizontal cuts are placed consecutively.
Therefore, in the first stage we need only $K$ vertical cuts and we can have at most $(\frac{1}{\delta}2^{(1/\delta^{2})})$ resulting pieces. We
apply the same transformation to each stage with vertical cuts. Now
observe that there can be at most $O(1/\delta)$ stages since
there are at most $1/\delta$ possible tall, vertical or large items stacked on top of the other and thus at most $1/\delta$ stages with
horizontal cuts. Therefore, after our transformations, we apply only
$(\frac{1}{\delta}2^{(1/\delta^{2})})^{\frac{1}{\delta}}$ cuts in total, in all stages in all resulting pieces.
Thus, we obtain $O_{\eps}(1)$ boxes at the end, in which the
items are nicely packed. This leads to the following lemma.
\begin{lemma}
\label{lem:hor-boxes-1}Given a box $B\in\B_2$ such that $B$ contains
a set of items $\Rhard(B)\subseteq\Rit\cup\Rbg\cup\Rve$. There exists
a partition of $B$ into $O_{\eps}(1)$ containers $\B'$ such that
 the containers $\B'$ are guillotine separable and contain the items $\Rhard(B)$.
\end{lemma}
We apply Lemma~\ref{lem:hor-boxes-1} to each box $B\in\B_2$ that
contains an item from $\Rit\cup\Rbg\cup\Rve$. Thus, we obtain a packing
of $\Rhard$ into a set of $O_{\eps}(1)$ guillotine separable
containers in which these items are nicely placed; we denote these containers
by $\B_{hard}$. This yields directly a packing for the (original)
items $\Rhard$ (without rounding). Finally, we partition the empty space of the resulting
packing into more boxes, and one additional box that we place on top
of the current packing. We pack the items in $\Rsm$ inside all these boxes.
We might not be able to use some parts of the empty space, e.g., if
two boxes are closer than $\mu W$ to each other horizontally; however, if 
$\mu$ is sufficiently small compared to the number of boxes, this space is small and compensated by the additional box.

\begin{lemma}
\label{lem_small_rearrange}
Assume that $\mu$ is sufficiently small compared to $\delta$.
There exists a set of $O_{\eps}(1)$ boxes $\B_{small}$, all
contained in $[0,\W]\times[0,(1+14\eps)\opt]$, such that the
boxes in $\B_{hard}\cup\B_{small}$ are non-overlapping and guillotine
separable and the items in $\Rsm$ can be placed nicely into the boxes
$\B_{small}$.
\end{lemma}

Finally, we show the following lemma by using the fact that the medium items have area at most $\eps(\opt\cdot\W)$ and by applying Theorem~\ref{steinberg_alg}. This completes the proof of Lemma~\ref{lem:structural}.
\begin{lemma}
\label{lem:pack-med}In time $n^{O(1)}$ we can find a nice placement
of all items in $\Rme$ inside one container $B_{med}$ of height $2\eps\opt$
and width $\W$.
\end{lemma}

\subsection{Algorithm}

We describe now our algorithm that computes a packing of height at
most $(1+O(\eps))\opt$. First, we guess $\opt$ and observe that
there are at most $n\cdot\H$ possibilities, where $\H:=\max_{i\in\R}\height_{i}$.
Then, we guess the set of containers $\B$ due to Lemma~\ref{lem:structural}
and their placement inside $[0,\W]\times[0,(1+O(\eps))\opt)$. For
each container $B\in\B$ we guess which case of Definition~\ref{def:structured-boxes}
 applies to $B$, i.e., whether $\R_{B}$ contains only one item, $\R_{B}\subseteq\Rho$,
$\R_{B}\subseteq\Rit\cup\Rve$, $\R_{B}\subseteq\Rme$, or $\R_{B}\subseteq\Rsm$.
For each box $B\in\B$ for which $\R_{B}$ contains only one item
$i\in\R$, we guess $i$. Observe that for the remaining containers this
yields independent subproblems for the sets $\Rho$, $\Rit\cup\Rve$,
$\Rme$, and $\Rsm$. 
We solve these subproblems via similar routines as in \cite{nadiradze2016approximating, GGIK16, JansenR19}.

We pack all medium items in $\Rme$ into one single
container $B_{med}$ of height $2\eps\opt$ by Lemma~\ref{lem:pack-med} (see \Cref{subsec:pptasmed}).
Then, for the sets $\Rho$ and $\Rit\cup\Rve$ we pack their respective
items into their containers using a standard pseudo-polynomial time dynamic
program; we denote these containers by $\B_{hor}$ and $\B_{tall+ver}$,
respectively. We crucially use that $\left|\B_{hor}\right|\le O_{\eps}(1)$
and $\left|\B_{tall+ver}\right|\le O_{\eps}(1)$. See  \Cref{subsec:pptasskew} for details of packing of items in $\Rho$ and $\Rit\cup\Rve$. 

Finally, we pack small items (see \Cref{subsec:pptassmall} for details). 
From the proof of Lemma~\ref{lem_small_rearrange}, apart from some items $I_{small}'\subset I_{small}$  which have area at most $\eps\text{OPT}\cdot W$, the other items can be packed nicely in the containers in $\B_{small}\setminus B_{small}$, where $B_{small}$ has height $9\eps\text{OPT}$ and width $W$. Thus, we use NFDH for packing the remaining small items. It can be shown that the small items which remain unpacked  can be packed nicely in $B_{small}$, which is placed on the top of our packing.

\begin{theorem}\label{thm:pptas}
There is a $(1+\eps)$-approximation algorithm for the guillotine
strip packing problem with a running time of $(n\W)^{O_{\eps}(1)}$.
\end{theorem}
See Section~\ref{subsec:algo-PPTAS} for the rest of the details


\section{Polynomial time $(\boldsymbol{\frac{3}{2}+\eps})$-approximation}	
In this section, we first present the structural lemma for our polynomial time $({3}/{2}+\eps)$-approximation
 algorithm for guillotine strip packing. Then we describe our algorithm.
~We have omitted a few proofs due to space constraints which can be found in Appendix \ref{sec:omitted_proofs2}.

To derive our structural lemma, we start with the
packing due to Lemma~\ref{lem:structural}. The problem is that with a polynomial time algorithm (rather than a pseudo-polynomial time algorithm)
we might not be able to pack all tall items in their respective boxes. 
If there is even one single tall item~$i$ that
we cannot pack, then we need to place $i$ on top of our packing,
which can increase the height of the packing by up to $\opt$.


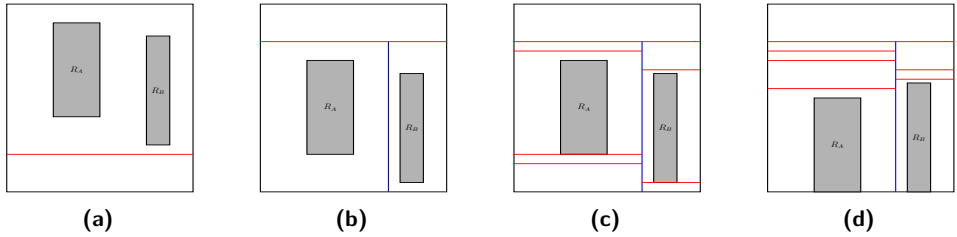
\begin{figure}[!htb]
\centering
\captionsetup[subfigure]{justification=centering}
\hspace{-0.85cm}
\begin{subfigure}[t]{0.23\linewidth}
\centering
\resizebox{2.5cm}{2.5cm}{
\begin{tikzpicture}
\draw [very thick] (0,0) rectangle (8,10);
\draw [fill=gray!60] (2,4) rectangle (4,9);
\draw [thick, red] (0,2) -- (8,2);
\draw [fill=gray!60] (6,2.5) rectangle (7,8.3);
\draw (3,6.5) node {\Large \textbf{$R_A$}};
\draw (6.5,5.4) node {\Large \textbf{$R_B$}};
\end{tikzpicture}}
\caption{}
\label{}
\end{subfigure}
\begin{subfigure}[t]{0.23\linewidth}
\centering
\resizebox{2.5cm}{2.5cm}{
\begin{tikzpicture}
\draw [very thick] (0,0) rectangle (8,10);
\draw [fill=gray!60] (2,2) rectangle (4,7);
\draw [thick, red] (0,8) -- (8,8);
\draw[thick, blue] (5.5,0) -- (5.5,8);
\draw [fill=gray!60] (6,0.5) rectangle (7,6.3);
\draw (3,4.5) node {\Large \textbf{$R_A$}};
\draw (6.5,3.4) node {\Large \textbf{$R_B$}};
\end{tikzpicture}}
\caption{}
\label{}
\end{subfigure}
\begin{subfigure}[t]{0.23\linewidth}
\centering
\resizebox{2.5cm}{2.5cm}{
\begin{tikzpicture}
\draw [very thick] (0,0) rectangle (8,10);
\draw [fill=gray!60] (2,2) rectangle (4,7);
\draw [thick, red] (0,8) -- (8,8);
\draw[thick, blue] (5.5,0) -- (5.5,8);
\draw[thick, red] (0,7.5)--(5.5,7.5); 
\draw[thick, red] (0,1.5)--(5.5,1.5);
\draw[thick, red] (0,2)--(5.5,2);
\draw [fill=gray!60] (6,0.5) rectangle (7,6.3);
\draw[thick, red] (5.5,6.5) --(8,6.5);
\draw[thick, red] (5.5,0.5) --(8,0.5);
\draw (3,4.5) node {\Large \textbf{$R_A$}};
\draw (6.5,3.4) node {\Large \textbf{$R_B$}};
\end{tikzpicture}}
\caption{}
\label{}
\end{subfigure}
\begin{subfigure}[t]{0.23\linewidth}
\centering
\resizebox{2.5cm}{2.5cm}{
\begin{tikzpicture}
\draw [very thick] (0,0) rectangle (8,10);
\draw [fill=gray!60] (2,0) rectangle (4,5);
\draw [thick, red] (0,8) -- (8,8);
\draw[thick, blue] (5.5,0) -- (5.5,8);
\draw[thick, red] (0,7.5)--(5.5,7.5);
\draw[thick, red] (0,7)--(5.5,7);
\draw[thick, red] (0,5.5)--(5.5,5.5);
\draw [fill=gray!60] (6,0) rectangle (7,5.8);
\draw[thick, red] (5.5,6.5) --(8,6.5);
\draw[thick, red] (5.5,6) --(8,6);
\draw (3,2.5) node {\Large \textbf{$R_A$}};
\draw (6.5,2.9) node {\Large \textbf{$R_B$}};
\end{tikzpicture}}
\caption{}
\label{}
\end{subfigure}
\caption{$R_A$ and $R_B$ are tall containers and by swapping the respective boxes (forming as a results of guillotine cuts) that contain them, they can be packed such that the bottoms of both containers intersect the bottom of the strip}
\label{fig_mirror_1}
\end{figure}

Therefore, we make our packing more robust to small errors when we pack the items into their
boxes.
In our changed packing, the
tall items are \emph{bottom-left-flushed
}(see Figure~\ref{fig_3/2_approx}(\subref{fig_3/2_approx_f})), the remaining items are packed into $O_{\eps}(1)$
boxes, and there is one extra box $B^{*}$ of height $\opt/2$ and
width $\Omega_{\eps}(W)$ which is empty. 
We will use the extra box $B^{*}$ in order to compensate small errors when we pack the vertical items.

Formally, we say that in a packing, a set of items $\R'$ is \emph{bottom-left-flushed} if they are ordered non-increasingly by height and stacked next to
each other in this order within the strip $[0,\W]\times[0,\infty)$
starting at the left edge of the strip, such that the bottom edge
of each item $i\in\R'$ touches the line segment $[0,\W]\times\{0\}$.
 We now state the modified structural lemma for our polynomial time $({3}/{2}+\eps)$-approximation algorithm formally.

\begin{lemma}[Structural lemma $2$]
\label{lem:structural-3/2}There exists a packing of the items $\R$
within $[0,\W]\times[0,(3/2+O(\eps))\opt)$ such that
\begin{itemize}
\item The items $\Rit$ are bottom-left-flushed,
\item There is a set $\B$ of $O_{\eps}(1)$ containers that are pairwise non-overlapping
and do not intersect the items in $\Rit$,
\item There is a partition of $\R\setminus\Rit=\bigcup_{B\in\B}\R_{B}$
such that for each $B\in\B$ the items in $\R_{B}$ can be placed
nicely into $B$, 
\item There is a container $B^{*}\in\B$ of height $\opt/2$ and width $\eps_1\W$ such that $\R_{B^{*}}=\emptyset$,
\item The items $\Rit$ and the containers $\B$ together are guillotine separable.
\end{itemize}
\end{lemma}

We now prove Lemma~\ref{lem:structural-3/2} in the following subsection.

\subsection{Proof of Structural Lemma 2}
\label{subsec:structural_2}

We start with the packing due to Lemma~\ref{lem:structural} and
transform it step by step.
To obtain our packing, we first argue that we can ensure that all tall items are placed on the bottom of the strip, 
i.e., their bottom edges touch the bottom edge of the strip. Here we use that the initial packing is guillotine separable.
Then we place the box $B^{*}$ as follows. Suppose that there are initially 
$C$ containers that cross the horizontal line with $y=\opt/2$. Note that $C=O_{\eps}(1)$ and $C\geq 1$ since at least one container is required to pack given non-zero number of items.
Then, by an averaging argument we can show that there is a line segment $l^{*}$ of length at least $\Omega(\frac{W}{C})$ which is the top edge of one of the containers $B$ in the packing at some height $h^{*}\geq \opt/2$. We push all the containers which completely lie above the line $y=h^{*}$ vertically upward by $\opt/2$ and this creates enough space to pack $B^{*}$ on top of $B$. After that, we take advantage of the gained extra space in order to
ensure that the tall items are bottom-left-flushed.

Now we describe the proof formally. First, we define some constants. 
Let $g(\delta,\eps)=O_{\eps}(1)$ denote an upper bound on the number of containers in the packing obtained using Lemma~\ref{lem:structural}, depending on $\eps$ and $\delta$. 
Let $\eps_1=\frac{1}{3g(\delta,\eps)}$, $\eps_2=\frac{\eps}{4\left|\B_{hor}\right|}$,  $\eps_3=\frac{\eps_1}{4\left|\B_{ver}\right|}$, $\eps_4=\mu$, $\eps_5=\frac{\eps_1\delta}{6}$, $\eps_6=\frac{\eps\delta}{6}$.

 Our first goal is to make sure
that the tall items are all placed on the bottom of the strip $[0,\W]\times[0,\infty)$.
For this, we observe the following: suppose that in the guillotine
cutting sequence a horizontal cut is placed. This cut separates the
current rectangular piece $R$ into two smaller pieces $R_{1}$ and
$R_{2}$. Suppose that $R_{1}$ lies on top of $R_{2}$. Then only
one of the two pieces $R_{1},R_{2}$ can contain a tall item. Also,
we obtain an alternative guillotine separable packing if we swap $R_{1}$
and $R_{2}$---together with the items contained in them---within
$R$. We perform this swap if $R_{1}$ contains a tall item. We apply
this operation to each horizontal cut in the guillotine cutting sequence.
As a result, we obtain a new packing in which all tall items are placed
on the bottom of the strip (but possibly not yet bottom-left-flushed) as shown in Fig \ref{fig_mirror_1}.
\begin{lemma}
\label{lem:structural-1}There exists a set $\B$ of $O_{\eps}(1)$
pairwise non-overlapping and guillotine separable boxes that are all
placed inside $[0,\W]\times[0,(1+16\eps)\opt)$ and a partition
$\R=\bigcup_{B\in\B}\R_{B}$ such that for each $B\in\B$ the items
in $\R_{B}$ can be placed nicely into $B$.
Also, for each box $B\in\B$ with $I_B \cap \Rit \ne \emptyset$
we have that the bottom edge of $B$
intersects the line segment $[0,\W]\times\{0\}$. 
\end{lemma}

\begin{figure}[tb]
\centering
\hspace*{-0.8em}
\captionsetup[subfigure]{justification=centering}
\begin{subfigure}[b]{0.28\linewidth}
\centering
\resizebox{3cm}{3cm}{
\begin{tikzpicture}
\draw [ultra thick](0,0) rectangle (8,10);
\draw [line width=1 mm,blue] (0,6.6) -- (1.85,6.6);

\draw [ultra thick](0,0) rectangle (1.8,6.5);
\draw [ultra thick](1.85,0) rectangle (2.6,5.5);
\draw [ultra thick,fill=gray!20](2.65,0) rectangle (3.6,1.5);
\draw [ultra thick,fill=gray!20](2.65,1.5) rectangle (3.5,3.2);
\draw [ultra thick,fill=gray!20](3.65,0) rectangle (5.15,3.3);
\draw [ultra thick,fill=gray!20](4.1,3.4) rectangle (5.15,5.1);
\draw [ultra thick,fill=gray!20](4.1,5.1) rectangle (4.7,6);
\draw [ultra thick,fill=gray!20](4.7,5.1) rectangle (5.15,6.6);
\draw [ultra thick,fill=gray!20](2.65,3.3) rectangle (3.1,5.2);
\draw [ultra thick,fill=gray!20](3.1,3.3) rectangle (3.6,5.6);
\draw [ultra thick,fill=gray!20](3.65,3.3) rectangle (4.1,6.8);
\draw [ultra thick](5.2,0) rectangle (6.1,7.1);
\draw [ultra thick](7.15,0) rectangle (8,8);
\draw [ultra thick,fill=gray!20](6.1,0) rectangle (7.1,4);
\draw [ultra thick,fill=gray!20](6.1,4) rectangle (6.5,6);
\draw [ultra thick,fill=gray!20](6.5,4) rectangle (7.1,6.9);

\draw [ultra thick,fill=gray!20](0.1,8.2) rectangle (5,9);
\draw [ultra thick,fill=gray!20](0.1,9.05) rectangle (5.9,9.9);
\draw [ultra thick,fill=gray!20](6.1,8.3) rectangle (8,9.8);
\draw [ultra thick,fill=gray!20](0.1,6.75) rectangle (1.7,8.1);
\draw [ultra thick,fill=gray!20](1.85,5.65) rectangle (3.6,6.3);
\draw [ultra thick,fill=gray!20](1.85,6.3) rectangle (3.2,8.1);
\draw [ultra thick,fill=gray!20](3.2,6.3) rectangle (3.6,7.9);
\draw [ultra thick,fill=gray!20](5.25,7.1) rectangle (6.05,9);
\draw [ultra thick,fill=gray!20](6.15,6) rectangle (6.5,8.1);
\draw [ultra thick,fill=gray!20](6.55,6.9) rectangle (7.1,8.1);
\draw [ultra thick,fill=gray!20](3.65,6.85) rectangle (5.1,7.5);
\draw [ultra thick,fill=gray!20](3.8,7.55) rectangle (5,8.15);

\draw [ultra thick, fill=gray!60] (0,0) rectangle (0.85,6.45);
\draw [ultra thick, fill=gray!60] (0.9,0) rectangle (1.3,6.1);
\draw [ultra thick, fill=gray!60] (1.35,0) rectangle (1.75,5.9);
\draw [ultra thick, fill=gray!60] (1.85,0) rectangle (2.15,5.45);
\draw [ultra thick, fill=gray!60] (2.2,0) rectangle (2.55,5.1);
\draw [ultra thick, fill=gray!60] (5.2,0) rectangle (5.45,7.05);
\draw [ultra thick, fill=gray!60] (5.5,0) rectangle (6,6);
\draw [ultra thick, fill=gray!60] (7.15,0) rectangle (7.6,7.95);
\draw [ultra thick, fill=gray!60] (7.65,0) rectangle (7.95,7);

\draw (-0.75,5) node {\Large $\boldsymbol{\frac{1}{2}\textbf{OPT}}$};
\draw (-0.75,10) node {\Large $\boldsymbol{\textbf{OPT}}$};
\draw (-0.75,6.5) node {\Large $\boldsymbol{h^{*}}$};
\draw [dashed] (0,5) -- (8,5);

\end{tikzpicture}}
\caption{\hspace*{-1.5em}}
\label{fig_3/2_approx_a}
\end{subfigure}
\hspace*{1.9em}
\begin{subfigure}[b]{0.3\linewidth}
\resizebox{3cm}{4.5cm}{
\begin{tikzpicture}
\draw [ultra thick](0,0) rectangle (8,15);
%

\draw [ultra thick](0,0) rectangle (1.8,6.5);
\draw [ultra thick](1.85,0) rectangle (2.6,5.5);
\draw [ultra thick,fill=gray!20](2.65,0) rectangle (3.6,1.5);
\draw [ultra thick,fill=gray!20](2.65,1.5) rectangle (3.5,3.2);
\draw [ultra thick,fill=gray!20](3.65,0) rectangle (5.15,3.3);
\draw [ultra thick,fill=gray!20](4.1,3.4) rectangle (5.15,5.1);
\draw [ultra thick,fill=gray!20](4.1,5.1) rectangle (4.7,6);
\draw [ultra thick,fill=gray!20](4.7,5.1) rectangle (5.15,6.6);
\draw [ultra thick,fill=gray!20](2.65,3.3) rectangle (3.1,5.2);
\draw [ultra thick,fill=gray!20](3.1,3.3) rectangle (3.6,5.6);
\draw [ultra thick,fill=gray!20](3.65,3.3) rectangle (4.1,6.8);
\draw [ultra thick](5.2,0) rectangle (6.1,7.1);
\draw [ultra thick](7.15,0) rectangle (8,8);
\draw [ultra thick,fill=gray!20](6.1,0) rectangle (7.1,4);
\draw [ultra thick,fill=gray!20](6.1,4) rectangle (6.5,6);
\draw [ultra thick,fill=gray!20](6.5,4) rectangle (7.1,6.9);

\draw [ultra thick,fill=gray!20](0.1,13.2) rectangle (5,14);
\draw [ultra thick,fill=gray!20](0.1,14.05) rectangle (5.9,14.9);
\draw [ultra thick,fill=gray!20](6.1,13.3) rectangle (8,14.8);
\draw [ultra thick,fill=gray!20](0.1,11.65) rectangle (1.7,13.1);
\draw [ultra thick,fill=gray!20](1.85,5.65) rectangle (3.6,6.3);
\draw [ultra thick,fill=gray!20](1.85,6.3) rectangle (3.2,8.1);
\draw [ultra thick,fill=gray!20](3.2,6.3) rectangle (3.6,7.9);
\draw [ultra thick,fill=gray!20](5.25,12.1) rectangle (6.05,14);
\draw [ultra thick,fill=gray!20](6.15,6) rectangle (6.5,8.1);
\draw [ultra thick,fill=gray!20](6.55,11.9) rectangle (7.1,13.1);
\draw [ultra thick,fill=gray!20](3.65,11.85) rectangle (5.1,12.5);
\draw [ultra thick,fill=gray!20](3.8,12.55) rectangle (5,13.15);

\draw [ultra thick, fill=gray!60] (0,0) rectangle (0.85,6.45);
\draw [ultra thick, fill=gray!60] (0.9,0) rectangle (1.3,6.1);
\draw [ultra thick, fill=gray!60] (1.35,0) rectangle (1.75,5.9);
\draw [ultra thick, fill=gray!60] (,) rectangle (,);
\draw [ultra thick, fill=gray!60] (,) rectangle (,);
\draw [ultra thick, fill=gray!60] (1.85,0) rectangle (2.15,5.45);
\draw [ultra thick, fill=gray!60] (2.2,0) rectangle (2.55,5.1);
\draw [ultra thick, fill=gray!60] (5.2,0) rectangle (5.45,7.05);
\draw [ultra thick, fill=gray!60] (5.5,0) rectangle (6,6);
\draw [ultra thick, fill=gray!60] (7.15,0) rectangle (7.6,7.95);
\draw [ultra thick, fill=gray!60] (7.65,0) rectangle (7.95,7);

\filldraw [ultra thick,color=red!60,pattern=north west lines,pattern color =green] (0.2,6.5) rectangle (1.6,11.5) ;

\draw (-0.75,5) node {\Large $\boldsymbol{\frac{1}{2}\textbf{OPT}}$};
\draw (-0.75,10) node {\Large $\boldsymbol{\textbf{OPT}}$};
\draw (-0.75,6.5) node {\Large $\boldsymbol{h^{*}}$};
\draw (-0.75,15) node {\Large $\boldsymbol{\frac{3}{2}\textbf{OPT}}$};
\draw (0.9,9) node {\Large $\boldsymbol{B^{*}}$};


\draw [dashed] (0,5) -- (8,5);
\draw [dashed] (0,10) -- (8,10);
\draw [dashed,red] (0,6.5) -- (8,6.5);
\draw [red] (0,11.5) -- (8,11.5);
\end{tikzpicture}}
\caption{\hspace*{1.7em}}
\label{fig_3/2_approx_b}
\end{subfigure}
\begin{subfigure}[b]{0.3\linewidth}
\resizebox{3cm}{4.5cm}{
\begin{tikzpicture}
 \draw [ultra thick](0,0) rectangle (8,15);
%

\draw [ultra thick](0,0) rectangle (1.8,6.5);
\draw [ultra thick](1.85,0) rectangle (2.6,5.5);
\draw [ultra thick,fill=gray!20](2.65,0) rectangle (3.6,1.5);
\draw [ultra thick,fill=gray!20](2.65,1.5) rectangle (3.5,3.2);
\draw [ultra thick,fill=gray!20](3.65,0) rectangle (5.15,3.3);
\draw [ultra thick,fill=gray!20](4.1,3.4) rectangle (5.15,5.1);
\draw [ultra thick,fill=Cerulean](4.1,6.65) rectangle (4.7,7.55);
\draw [ultra thick,fill=Cerulean](4.7,6.65) rectangle (5.15,8.15);
\draw [ultra thick,fill=gray!20](2.65,3.3) rectangle (3.1,5.2);
\draw [ultra thick,fill=gray!20](3.1,3.3) rectangle (3.6,5.6);
\draw [ultra thick,fill=gray!20](3.65,3.3) rectangle (4.1,6.8);
\draw [ultra thick](5.2,0) rectangle (6.1,7.1);
\draw [ultra thick](7.15,0) rectangle (8,8);
\draw [ultra thick,fill=gray!20](6.1,0) rectangle (7.1,4);
\draw [ultra thick,fill=gray!20](6.1,4) rectangle (6.5,6);
\draw [ultra thick,fill=gray!20](6.5,4) rectangle (7.1,6.9);
\draw [ultra thick,fill=gray!20](0.1,13.2) rectangle (5,14);
\draw [ultra thick,fill=gray!20](0.1,14.05) rectangle (5.9,14.9);
\draw [ultra thick,fill=gray!20](6.1,13.3) rectangle (8,14.8);
\draw [ultra thick,fill=gray!20](0.1,11.65) rectangle (1.7,13.1);
\draw [ultra thick,fill=Cerulean](1.85,7.2) rectangle (3.6,7.85);
\draw [ultra thick,fill=Cerulean](1.85,7.85) rectangle (3.2,9.65);
\draw [ultra thick,fill=Cerulean](3.2,7.85) rectangle (3.6,9.45);
\draw [ultra thick,fill=gray!20](5.25,12.1) rectangle (6.05,14);
\draw [ultra thick,fill=Cerulean](6.15,7.55) rectangle (6.5,9.65);
\draw [ultra thick,fill=gray!20](6.55,11.9) rectangle (7.1,13.1);
\draw [ultra thick,fill=gray!20](3.65,11.85) rectangle (5.1,12.5);
\draw [ultra thick,fill=gray!20](3.8,12.55) rectangle (5,13.15);

\draw [ultra thick, fill=gray!60] (0,0) rectangle (0.85,6.45);
\draw [ultra thick, fill=gray!60] (0.9,0) rectangle (1.3,6.1);
\draw [ultra thick, fill=gray!60] (1.35,0) rectangle (1.75,5.9);
\draw [ultra thick, fill=gray!60] (,) rectangle (,);
\draw [ultra thick, fill=gray!60] (,) rectangle (,);
\draw [ultra thick, fill=gray!60] (1.85,0) rectangle (2.15,5.45);
\draw [ultra thick, fill=gray!60] (2.2,0) rectangle (2.55,5.1);
\draw [ultra thick, fill=gray!60] (5.2,0) rectangle (5.45,7.05);
\draw [ultra thick, fill=gray!60] (5.5,0) rectangle (6,6);
\draw [ultra thick, fill=gray!60] (7.15,0) rectangle (7.6,7.95);
\draw [ultra thick, fill=gray!60] (7.65,0) rectangle (7.95,7);

\filldraw [ultra thick,color=red!60,pattern=north west lines,pattern color =green] (0.2,6.5) rectangle (1.6,11.5) ;

\draw (-0.75,5) node {\Large $\boldsymbol{\frac{1}{2}\textbf{OPT}}$};
\draw (-0.75,10) node {\Large $\boldsymbol{\textbf{OPT}}$};
\draw (-0.75,6.5) node {\Large $\boldsymbol{h^{*}}$};
\draw (-0.75,15) node {\Large $\boldsymbol{\frac{3}{2}\textbf{OPT}}$};
\draw (0.9,9) node {\Large $\boldsymbol{B^{*}}$};


\draw [dashed] (0,5) -- (8,5);
\draw [dashed] (0,10) -- (8,10);
\draw [dashed,red] (0,6.5) -- (8,6.5);
\draw [red] (0,11.5) -- (8,11.5);
\end{tikzpicture}}
\caption{\hspace*{1.4em}}
\label{fig_3/2_approx_c}
\end{subfigure}
\newline
\begin{subfigure}[b]{0.3\linewidth}
\resizebox{3cm}{4.5cm}{
\begin{tikzpicture}
 \draw [ultra thick](0,0) rectangle (8,15);
%

\draw [ultra thick](0,0) rectangle (1.8,6.5);
\draw [ultra thick](1.85,0) rectangle (2.6,5.5);
\draw [ultra thick,fill=gray!20](2.65,0) rectangle (3.6,1.5);
\draw [ultra thick,fill=gray!20](2.65,1.5) rectangle (3.5,3.2);
\draw [ultra thick,fill=gray!20](3.65,0) rectangle (5.15,3.3);
\draw [ultra thick,fill=gray!20](4.1,3.4) rectangle (5.15,5.1);
\draw [ultra thick,fill=Cerulean](4.1,6.65) rectangle (4.7,7.55);
\draw [ultra thick,fill=Cerulean](4.7,6.65) rectangle (5.15,8.15);
\draw [ultra thick,fill=gray!20](2.65,3.3) rectangle (3.1,5.2);
\draw [ultra thick,fill=gray!20](3.1,3.3) rectangle (3.6,5.6);
\draw [ultra thick,fill=Dandelion](3.65,6.55) rectangle (4.1,10.05);
\draw [ultra thick](5.2,0) rectangle (6.1,7.1);
\draw [ultra thick](7.15,0) rectangle (8,8);
\draw [ultra thick,fill=gray!20](6.1,0) rectangle (7.1,4);
\draw [ultra thick,fill=gray!20](6.1,4) rectangle (6.5,6);
\draw [ultra thick,fill=Dandelion](6.5,6.55) rectangle (7.1,9.45);

\draw [ultra thick,fill=gray!20](0.1,13.2) rectangle (5,14);
\draw [ultra thick,fill=gray!20](0.1,14.05) rectangle (5.9,14.9);
\draw [ultra thick,fill=gray!20](6.1,13.3) rectangle (8,14.8);
\draw [ultra thick,fill=gray!20](0.1,11.65) rectangle (1.7,13.1);
\draw [ultra thick,fill=Cerulean](1.85,7.2) rectangle (3.6,7.85);
\draw [ultra thick,fill=Cerulean](1.85,7.85) rectangle (3.2,9.65);
\draw [ultra thick,fill=Cerulean](3.2,7.85) rectangle (3.6,9.45);
\draw [ultra thick,fill=gray!20](5.25,12.1) rectangle (6.05,14);
\draw [ultra thick,fill=Cerulean](6.15,7.55) rectangle (6.5,9.65);
\draw [ultra thick,fill=gray!20](6.55,11.9) rectangle (7.1,13.1);
\draw [ultra thick,fill=gray!20](3.65,11.85) rectangle (5.1,12.5);
\draw [ultra thick,fill=gray!20](3.8,12.55) rectangle (5,13.15);

\draw [ultra thick, fill=gray!60] (0,0) rectangle (0.85,6.45);
\draw [ultra thick, fill=gray!60] (0.9,0) rectangle (1.3,6.1);
\draw [ultra thick, fill=gray!60] (1.35,0) rectangle (1.75,5.9);
\draw [ultra thick, fill=gray!60] (,) rectangle (,);
\draw [ultra thick, fill=gray!60] (,) rectangle (,);
\draw [ultra thick, fill=gray!60] (1.85,0) rectangle (2.15,5.45);
\draw [ultra thick, fill=gray!60] (2.2,0) rectangle (2.55,5.1);
\draw [ultra thick, fill=gray!60] (5.2,0) rectangle (5.45,7.05);
\draw [ultra thick, fill=gray!60] (5.5,0) rectangle (6,6);
\draw [ultra thick, fill=gray!60] (7.15,0) rectangle (7.6,7.95);
\draw [ultra thick, fill=gray!60] (7.65,0) rectangle (7.95,7);

\filldraw [ultra thick,color=red!60,pattern=north west lines,pattern color =green] (0.2,6.5) rectangle (1.6,11.5) ;

\draw (-0.75,5) node {\Large $\boldsymbol{\frac{1}{2}\textbf{OPT}}$};
\draw (-0.75,10) node {\Large $\boldsymbol{\textbf{OPT}}$};
\draw (-0.75,6.5) node {\Large $\boldsymbol{h^{*}}$};
\draw (-0.75,15) node {\Large $\boldsymbol{\frac{3}{2}\textbf{OPT}}$};
\draw (0.9,9) node {\Large $\boldsymbol{B^{*}}$};
\draw (2,-0.75) node {\Large $\boldsymbol{}$};

\draw [dashed] (0,5) -- (8,5);
\draw [dashed] (0,10) -- (8,10);
\draw [dashed,red] (0,6.5) -- (8,6.5);
\draw [red] (0,11.5) -- (8,11.5);
\end{tikzpicture}}
\vspace{-0.5em}
\caption{\hspace*{1.7em}}
\label{fig_3/2_approx_d}
\end{subfigure}
\begin{subfigure}[b]{0.3\linewidth}
\resizebox{3cm}{4.5cm}{
\begin{tikzpicture}
 \draw [ultra thick](0,0) rectangle (8,15);
%

\draw [ultra thick](1.85,0) rectangle (3.65,6.5);
\draw [ultra thick](3.7,0) rectangle (4.45,5.5);
\draw [ultra thick](0.9,0) rectangle (1.8,7.1);
\draw [ultra thick](0,0) rectangle (0.85,8);

\draw [ultra thick,fill=gray!20](4.45,0) rectangle (5.4,1.5);
\draw [ultra thick,fill=gray!20](4.45,1.5) rectangle (5.3,3.2);
\draw [ultra thick,fill=gray!20](5.45,0) rectangle (6.95,3.3);
\draw [ultra thick,fill=gray!20](5.9,3.4) rectangle (6.95,5.1);
\draw [ultra thick,fill=Cerulean](5.9,6.65) rectangle (6.5,7.55);
\draw [ultra thick,fill=Cerulean](6.5,6.65) rectangle (6.95,8.15);
\draw [ultra thick,fill=gray!20](4.45,3.3) rectangle (4.9,5.2);
\draw [ultra thick,fill=gray!20](4.9,3.3) rectangle (5.4,5.6);
\draw [ultra thick,fill=Dandelion](5.45,6.55) rectangle (5.9,10.05);

\draw [ultra thick,fill=gray!20](7,0) rectangle (8,4);
\draw [ultra thick,fill=gray!20](7,4) rectangle (7.4,6);
\draw [ultra thick,fill=Dandelion](7.4,6.55) rectangle (8,9.45);

\draw [ultra thick,fill=gray!20](0.1,13.2) rectangle (5,14);
\draw [ultra thick,fill=gray!20](0.1,14.05) rectangle (5.9,14.9);
\draw [ultra thick,fill=gray!20](6.1,13.3) rectangle (8,14.8);
\draw [ultra thick,fill=gray!20](0.1,11.65) rectangle (1.7,13.1);
\draw [ultra thick,fill=Cerulean](3.65,7.2) rectangle (5.4,7.85);
\draw [ultra thick,fill=Cerulean](3.65,7.85) rectangle (5,9.65);
\draw [ultra thick,fill=Cerulean](5,7.85) rectangle (5.4,9.45);
\draw [ultra thick,fill=gray!20](5.25,12.1) rectangle (6.05,14);
\draw [ultra thick,fill=Cerulean](7.05,7.55) rectangle (7.4,9.65);
\draw [ultra thick,fill=gray!20](6.55,11.9) rectangle (7.1,13.1);
\draw [ultra thick,fill=gray!20](3.65,11.85) rectangle (5.1,12.5);
\draw [ultra thick,fill=gray!20](3.8,12.55) rectangle (5,13.15);

\draw [ultra thick, fill=gray!60] (1.9,0) rectangle (2.75,6.45);
\draw [ultra thick, fill=gray!60] (2.8,0) rectangle (3.2,6.1);
\draw [ultra thick, fill=gray!60] (3.25,0) rectangle (3.65,5.9);
\draw [ultra thick, fill=gray!60] (,) rectangle (,);
\draw [ultra thick, fill=gray!60] (,) rectangle (,);
\draw [ultra thick, fill=gray!60] (3.75,0) rectangle (4.05,5.45);
\draw [ultra thick, fill=gray!60] (4.1,0) rectangle (4.45,5.1);
\draw [ultra thick, fill=gray!60] (0.95,0) rectangle (1.2,7.05);
\draw [ultra thick, fill=gray!60] (1.25,0) rectangle (1.75,6);
\draw [ultra thick, fill=gray!60] (0,0) rectangle (0.45,7.95);
\draw [ultra thick, fill=gray!60] (0.5,0) rectangle (0.8,7);

\filldraw [ultra thick,color=red!60,pattern=north west lines,pattern color =green] (2,6.5) rectangle (3.4,11.5) ;

\draw (-0.75,5) node {\Large $\boldsymbol{\frac{1}{2}\textbf{OPT}}$};
\draw (-0.75,10) node {\Large $\boldsymbol{\textbf{OPT}}$};
\draw (-0.75,6.5) node {\Large $\boldsymbol{h^{*}}$};
\draw (-0.75,15) node {\Large $\boldsymbol{\frac{3}{2}\textbf{OPT}}$};
\draw (2.7,9) node {\Large $\boldsymbol{B^{*}}$};
\draw (1.85,-0.75) node {\fontsize{18}{24}\selectfont $\boldsymbol{x_0}$};

\draw [dashed] (0,5) -- (8,5);
\draw [dashed] (0,10) -- (8,10);
\draw [dashed,red] (0,6.5) -- (8,6.5);
\draw [red] (0,11.5) -- (8,11.5);
\draw [dashed,blue] (1.85,-0.5) -- (1.85,11.5);
\end{tikzpicture}}
\vspace{-0.6em}
\caption{\hspace*{1.4em}}
\label{fig_3/2_approx_e}
\end{subfigure}
\begin{subfigure}[b]{0.3\linewidth}
\resizebox{3cm}{4.5cm}{
\begin{tikzpicture}
 \draw [ultra thick](0,0) rectangle (8,15);
%

\draw [ultra thick,fill=gray!20](4.45,0) rectangle (5.4,1.5);
\draw [ultra thick,fill=gray!20](4.45,1.5) rectangle (5.3,3.2);
\draw [ultra thick,fill=gray!20](5.45,0) rectangle (6.95,3.3);
\draw [ultra thick,fill=gray!20](5.9,3.4) rectangle (6.95,5.1);
\draw [ultra thick,fill=Cerulean](5.9,6.65) rectangle (6.5,7.55);
\draw [ultra thick,fill=Cerulean](6.5,6.65) rectangle (6.95,8.15);
\draw [ultra thick,fill=gray!20](4.45,3.3) rectangle (4.9,5.2);
\draw [ultra thick,fill=gray!20](4.9,3.3) rectangle (5.4,5.6);
\draw [ultra thick,fill=Dandelion](5.45,6.55) rectangle (5.9,10.05);

\draw [ultra thick,fill=gray!20](7,0) rectangle (8,4);
\draw [ultra thick,fill=gray!20](7,4) rectangle (7.4,6);
\draw [ultra thick,fill=Dandelion](7.4,6.55) rectangle (8,9.45);

\draw [ultra thick,fill=gray!20](0.1,13.2) rectangle (5,14);
\draw [ultra thick,fill=gray!20](0.1,14.05) rectangle (5.9,14.9);
\draw [ultra thick,fill=gray!20](6.1,13.3) rectangle (8,14.8);
\draw [ultra thick,fill=gray!20](0.1,11.65) rectangle (1.7,13.1);
\draw [ultra thick,fill=Cerulean](3.65,7.2) rectangle (5.4,7.85);
\draw [ultra thick,fill=Cerulean](3.65,7.85) rectangle (5,9.65);
\draw [ultra thick,fill=Cerulean](5,7.85) rectangle (5.4,9.45);
\draw [ultra thick,fill=gray!20](5.25,12.1) rectangle (6.05,14);
\draw [ultra thick,fill=Cerulean](7.05,7.55) rectangle (7.4,9.65);
\draw [ultra thick,fill=gray!20](6.55,11.9) rectangle (7.1,13.1);
\draw [ultra thick,fill=gray!20](3.65,11.85) rectangle (5.1,12.5);
\draw [ultra thick,fill=gray!20](3.8,12.55) rectangle (5,13.15);

\draw [ultra thick, fill=gray!60] (0,0) rectangle (0.45,7.95);
\draw [ultra thick, fill=gray!60] (0.5,0) rectangle (0.75,7.05);
\draw [ultra thick, fill=gray!60] (0.8,0) rectangle (1.1,7);
\draw [ultra thick, fill=gray!60] (1.15,0) rectangle (2,6.45);
\draw [ultra thick, fill=gray!60] (2.05,0) rectangle (2.45,6.1);
\draw [ultra thick, fill=gray!60] (2.5,0) rectangle (3,6);
\draw [ultra thick, fill=gray!60] (3.05,0) rectangle (3.45,5.9);
\draw [ultra thick, fill=gray!60] (3.5,0) rectangle (3.8,5.45);
\draw [ultra thick, fill=gray!60] (3.85,0) rectangle (4.2,5.1);

\filldraw [ultra thick,color=red!60,pattern=north west lines,pattern color =green] (2,6.5) rectangle (3.4,11.5) ;

\draw (-0.75,5) node {\Large $\boldsymbol{\frac{1}{2}\textbf{OPT}}$};
\draw (-0.75,10) node {\Large $\boldsymbol{\textbf{OPT}}$};
\draw (-0.75,6.5) node {\Large $\boldsymbol{h^{*}}$};
\draw (-0.75,15) node {\Large $\boldsymbol{\frac{3}{2}\textbf{OPT}}$};
\draw (2.7,9) node {\Large $\boldsymbol{B^{*}}$};
\draw (1.15,-0.75) node {\fontsize{18}{24} $\boldsymbol{x_1}$};

\draw [dashed] (0,5) -- (8,5);
\draw [dashed] (0,10) -- (8,10);
\draw [dashed,red] (0,6.5) -- (8,6.5);
\draw [red] (0,11.5) -- (8,11.5);
\draw [dashed,blue] (1.15,-0.5) -- (1.15,11.5);
\end{tikzpicture}}
\vspace{-0.6em}
\caption{\hspace*{1.6em}}
\label{fig_3/2_approx_f}
\end{subfigure}
\caption{(a) A guillotine separable packing with items nicely packed in containers. The gray colored rectangles are the tall items and the light-gray rectangles are containers with items nicely packed inside. The blue line segment indicates ${l^{*}}$ at height $h^{\star}$.\:(b) Items completely packed in $[h^{*},\opt]$ are shifted by $\frac{1}{2}\opt$ vertically upward. The thick red line indicates $y=h^{*}+\frac{1}{2}\opt$ which separates the items shifted up from the items below.  The dashed red line indicates the height $h^{*}$ and $B^{*}$ is packed in the strip of sufficient width and lowest height $h^{*}$. \:(c) The containers of type $1$ (colored blue) are moved accordingly so they do not intersect $y=h^{\star}$. \:(d) The containers of type $2$ (colored yellow) are moved accordingly so they do not intersect $y=h^{\star}$. \:(e) The containers in $\B_{tall}$ are \textit{bottom-left-flushed} while other non-tall containers are moved accordingly to the right. The blue vertical dashed line $x=x_0$ separates containers in $\B_{tall}^{+}$ to its left hand side from other containers to the right.  \:(f) Final packing where tall items are \textit{bottom-left-flushed} and the blue vertical dashed line $x=x_1$ separates items $i\in I_{tall}$ with $h_i>h^{*}$ to the left from other items and containers to the right.}
\label{fig_3/2_approx}
\end{figure}
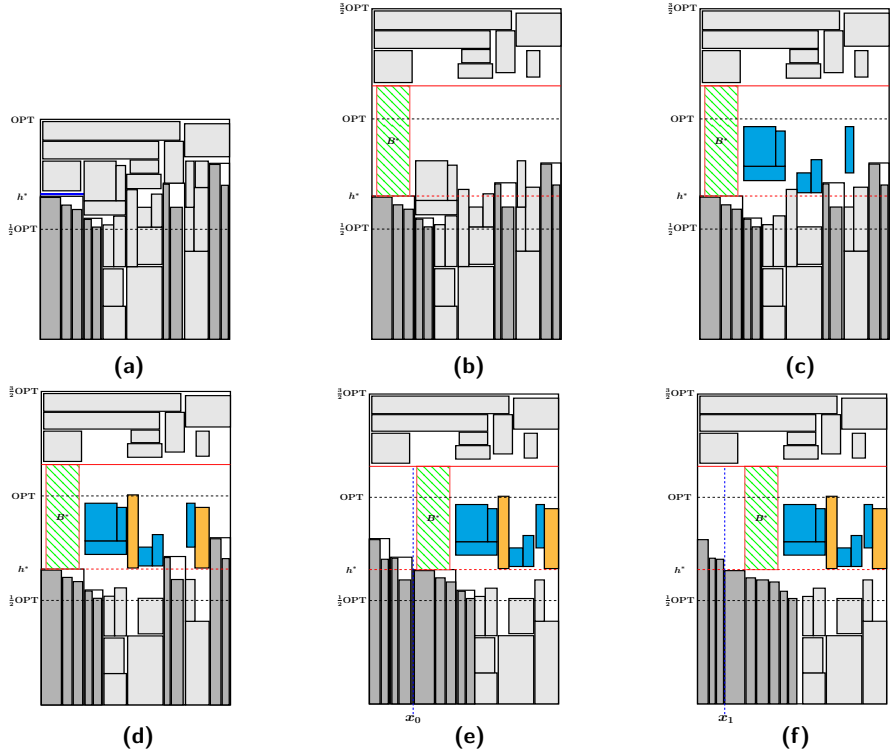

Let $\B$ be the set of containers due to Lemma~\ref{lem:structural-1}.
We want to move some of them up in order to make space for the additional
box $B^{*}$. To this end, we identify a horizontal line segment $\ell^{*}$
in the following lemma. 
\begin{lemma}
\label{lem_l^{*}}
There is a horizontal line segment $\ell^{*}$ of width at least $\eps_1\W$
that does not intersect any container in $\B$, and such that the $y$-coordinate
of $\ell^{*}$ is at least $\opt/2$.
\end{lemma}
\begin{proof}
Consider the containers in $\B$ that intersect with the horizontal line segment 
$\ell:=[0,W]\times\{\opt/2\}$ and let $p_1,...p_k$ be the maximally long line segments on $\ell$ that do not intersect any container.
 Since the line segments $\{p_1,...p_k\}$ are between containers in $\B$, we have that $k\leq |\B|+1$. Therefore by an averaging argument we can find a horizontal
line segment $\ell^{*}$ of width at least $\frac{W}{2(g(\delta,\eps))+1}\geq \frac{W}{3g(\delta,\eps)}\geq\eps_1 W$ that 
 either  contains
the top edge of one of these containers such that $\ell^{*}$ does
not intersect any other container in $\B$ or $\ell^{*}$ is one of the line segments in the set $\{p_1,...,p_k\}$. Hence, the $y$-coordinate of $\ell^{*}$ is at least $\opt/2$.
\end{proof}
Let $h^{*}$ be the $y$-coordinate of $\ell^{*}$. We take all containers
in $\B$ that lie ``above $h^{*}$'', i.e., that lie inside $[0,\W]\times[h^{*},\infty)$.
We translate them up by $\opt/2$. We define a container $B^{*}$ which has height $\opt/2$ and width $\eps_1 W$
to be packed such that $\ell^{*}$
is the bottom edge of $B^{*}$ (see Figure \ref{fig_3/2_approx}(\subref{fig_3/2_approx_b})). We then make the following claim about the resulting packing of $\B\cup\{B^{*}\}$ (we call this packing $P_1$). 
\begin{lemma}
\label{lem:P_1}
 The packing $P_1$ is feasible, guillotine separable and has height $(3/2+O(\eps))\opt$. 
\end{lemma}
\begin{proof}
Since $h^{*}>\opt/2$, observe that no containers are intersecting the line $[0,W]\times \{h^{*}+\opt/2\}$. This is because any containers which were lying above the line $[0,W]\times\{h^*\}$ before were pushed up by  $\opt/2$ and the height of such containers is at most $\opt/2$. Thus, the first guillotine cut is applied at $y=h^{*}+\opt/2$ so that we get two pieces $R$ and $R_{top}$. For the guillotine separability of the top piece $R_{top}$, we use the fact that the packing to begin with was guillotine separable and we have moved a subset of the items in the initial packing vertically upwards by the same height. For the bottom piece $R$, which has a subset of the initial packing, we have packed $B^{*}$ on the top edge (which is part of  the line $[0,W]\times\{h^*\}$) of another container (say $B$) whose width is more than the width of $B^{*}$. In the guillotine cutting sequence of this piece without the addition of $B^{*}$, consider the horizontal cuts at height at least $h^*$. Note that there is no container lying completely above the line $[0,W]\times\{h^*\}$ in $R$. Hence, we can remove such horizontal cuts and extend the vertical cuts that were intercepted by these horizontal cuts until they hit the topmost horizontal edge of $R$. Now, if we follow this new guillotine cutting sequence, we would finally have a rectangular region  with only  the container $B$. As there is no container in the region $[left(B),right(B)]\times [h^{*},h^{*}+\opt/2]$, we can pack $B^{*}$ in this region without violating the guillotine separability condition. Now, observe that the height of the piece $R_{top}$ is $(1+O(\eps))\opt-h^*$ and height of the piece $R$ is $h^*+\opt/2$. Hence the height of the packing $P_1$ is $(3/2+O(\eps)\opt$.
\end{proof}
 
Our next goal is to rearrange the tall items and their containers such
that the tall items are bottom-left flushed. Let $\B_{tall}\subseteq\B$
denote the containers in $\B$ that contain at least one tall item. Consider
the line segment $\ell:=[0,\W]\times\{h^{*}\}$ and observe that it
might be intersected by containers in $\B_{tall}$. Let $\ell_{1},\ell_{2},...,\ell_{t}$
be the connected components of $\ell\setminus\bigcup_{B\in\B_{tall}}B$.
For each $j\in\{1,...,t\}$ we do the following. Consider the containers
in $\B\setminus\B_{tall}$ whose bottoms are contained in $\ell_{j}\times[\opt/2,h^{*}]$ (we call them \emph{type $1$ containers}).
We move them up by $h^{*}-\opt/2$ units. There is enough space for them
since the top edge of any of these containers lies below the line segment $[0,W]\times \{h^{*}+\opt/2\}$ after shifting.

Then we take all containers in $\B\setminus\B_{tall}$ that intersect $\ell_{j}$
and also the line segment $[0,\W]\times\{\opt/2\}$ (\emph{type $2$ containers}). We move them
up such that their respective bottom edges are contained in $\ell_{j}$.
Again there
is enough space for this since the containers have height at most $\opt/2$ and hence, their top edges cannot cross the line segment $[0,W]\times \{h^{*}+\opt/2\}$. Note that in this step we do not necessarily
move the affected containers uniformly. See Figure~\ref{fig_3/2_approx}(\subref{fig_3/2_approx_c}) and Figure ~\ref{fig_3/2_approx}(\subref{fig_3/2_approx_d}) for a sketch. Note that due to the way $\ell^*$ is defined, no type 1 or type 2 container after being shifted overlaps with the region occupied by $B^*$.

One can show that the resulting packing is still guillotine separable.
In particular, there is such a sequence that starts as follows: the
first cut of this sequence is a horizontal cut with $y$-coordinate
$h^{*}+\opt/2$. For the  resulting bottom piece $R$, there are
vertical cuts that cut through the vertical edges of the containers in
$\B_{tall}$ whose height is strictly greater than $h^{*}$, denote these containers by
$\B_{tall}^{+}$. Let $R_{1},...,R_{t'}$ denote the resulting partition
of $R$. We can rearrange our packing by reordering
the pieces $R_{1},...,R_{t'}$. We reorder them such that on the left
we place the pieces containing one container from $\B_{tall}^{+}$ each,
sorted non-increasingly by their heights. Then we place
the remaining pieces from $\{R_{1},...,R_{t'}\}$ (which hence, do
not contain any containers in $\B_{tall}^{+}$), denote their union by $R'$. Let the left end of $R'$ be $x=x_0$.
We can assume that the guillotine cutting sequence places a vertical
cut that separates $R'$ from the other pieces in $\{R_{1},...,R_{t'}\}$ at $x=x_0$. From Lemma~\ref{lem_l^{*}}, we know that there is a container $B$ (or possibly the case when $h^{*}=\opt/2$ and we have a line segment $\ell'$ of width at least $\eps_1 W$ on top of which we can pack $B^{*}$) whose top is at height $h^{*}$, has width at least $\ell^{*}$ which now lies to the right of $x_0$ in $R'$. Thus, the region $[left(B),left(B)+\eps_1 W]\times[h^{*},h^{*}+\opt/2]$ is empty and can be used to place $B^{*}$.

We change now the placement of the containers within $R'$. Due to our
rearrangements, no container inside $R'$ intersects the line segment $[0,\W]\times\{h^{*}\}$,
so we can assume that $R'$ is cut by the horizontal cut $[0,\W]\times\{h^{*}\}$,
let $R''$ be the resulting bottom piece and $R'''$ be the piece above. We first show why $R'''$ is guillotine separable. First, we separate $B^{*}$ using vertical guillotine cuts at its left and right edges. Then we prove that the shifting operation for type $2$ and type $1$ containers does not violate guillotine separability of the packing for any region defined by  some horizontal segment $l_j$ for $j\in[t]$. 
Consider any type $2$ container $B'$. Its top edge was initially lying above $y=h^{*}$ and its bottom below $\opt/2$. Hence, before shifting this container no item could have been packed such that it was in the region $[0,W]\times [h^{*},h^{*}+\opt/2]$ and was intersecting the vertically extended line segments from the left and right edges of $B'$ because any item packed in $[0,W]\times [h^{*},\infty)$ initially was shifted upward by $\opt/2$. Hence, after shifting $B'$ such that its bottom touches $y=h^{*}$, after considering the cut $y=h^{*}$ in $l_j$, extend its left and right edges vertically upward to separate $B'$ using guillotine cuts. For the type $1$ containers, after the aforementioned cuts observe that all such containers have been shifted by an equal amount vertically upward and using the fact that they were guillotine separable initially, we claim that they are guillotine separable afterward. This is proved by considering the initial guillotine cuts that were separating such items and shifting the horizontal cuts upward by $h^{*}-\opt/2$ (equal to the distance the type $1$ containers were shifted upward by).

 To show that $R''$ is guillotine separable, observe that due to our rearrangements
there are no containers that are completely contained in $R''\cap\left([0,\W]\times[\opt/2,h^{*}]\right)$.
Therefore, we can assume that the next cuts for $R''$ are vertical
cuts that contain all vertical edges of the boxes in $\B_{tall}$
that are contained in $R''$. Let $R''_{1},...,R''_{t''}$ denote
the resulting pieces. Like above, we change our packing such that
we reorder the pieces in $R''_{1},...,R''_{t''}$ non-increasingly
by the height of the respective box in $\B_{tall}$ contained in them,
and at the very right we place the pieces from $R''_{1},...,R''_{t''}$
that do not contain any container from $\B_{tall}$ (see Figure~\ref{fig_3/2_approx}(\subref{fig_3/2_approx_e}))

Finally, we sort the tall items inside the area $\bigcup_{B\in\B_{tall}}B$
non-increasingly by height so that they are bottom-left-flushed, and
we remove the containers $\B_{tall}$ from $\B$ (see Figure \ref{fig_3/2_approx}(\subref{fig_3/2_approx_f})).
 We now prove that the tall items can be sorted inside the area $\bigcup_{B\in\B_{tall}}B$
non-increasingly by height without violating guillotine separability and feasibility. Note that the area $\bigcup_{B\in\B_{tall}}B$ can possibly contain some vertical items.
 Now, we reorder the tall items within $R'$ such that they are sorted in non-increasing order of their heights. We do the same for all the tall items on the left of $R'$. There may be tall items (or vertical items) on the left hand side of $R'$ such that for any such item, its height is less than the the tallest tall item in $R'$. Note that such tall items have to have a height of at most $h^{*}$. Such items can be repeatedly swapped with their neighboring tall item till they are in the correct position  according to the \textit{bottom-left-flushed} packing of the tall items, while maintaining guillotine separability. Such a swap operation between consecutive tall items ensures that all of the tall items and possibly some vertical items which were initially packed in tall containers remain inside the area $\bigcup_{B\in\B_{tall}}B$. We ensure that the vertical items which were packed to the left of $R'$ get swapped so that they are packed on the right of all the tall items in a container. Now, to prove that guillotine separability of the packing is maintained after all such swapping operations, that is, after all tall items are sorted according to their heights in a non-increasing order consider the $x$-coordinate (say $x_1$) of the right edge of the shortest tall item which has height strictly greater than $h^{*}$. Observe that there were no tall containers of height strictly greater than $h^{*}$ beyond $x=x_0$, which implies $x_1\leq x_0$ and hence, now, for the guillotine cutting sequence, we can have a vertical guillotine cut at $x=x_1$ instead of at $x=x_0$, the rest being the same as mentioned before.  
This yields the packing
claimed by Lemma~\ref{lem:structural-3/2}.


\subsection{Algorithm for polynomial time $\boldsymbol{(\frac{3}{2}+\eps)}$-approximation}
\label{subsec_alg_poly}
 First we guess a value $\opt\textsuperscript{$\prime$}$ such that $\opt\leq \opt\textsuperscript{$\prime$}\leq (1+\varepsilon)\opt$ in $n^{O_\eps(1)}$ time (see Appendix~\ref{sec:omitted_proofs2} for details).
In order to keep the notation light we denote $\opt\textsuperscript{$\prime$}$ by $\opt$. 
We want to compute a packing of height at most $(\frac{3}{2}+O(\eps))\opt$
using Lemma~\ref{lem:structural-3/2}. 

Intuitively, we first place the tall items in a bottom-left-flushed way. Then we guess approximately the sizes of the boxes, place them in the free area, and place the items
inside them via guessing the relatively large items, solving an instance of the generalized assignment problem (GAP), using NFDH for the small items, and invoking again Lemma~\ref{lem:pack-med} for the medium items (see Appendix~\ref{subsec_GAP} for the definition of GAP). This is similar as in, e.g., \cite{GalvezGIHKW21,Jansen2004}.

Formally, first we place all items
in $\Rit$ inside $[0,\W]\times[0,(3/2+\eps)\opt)$ such that they
are bottom-left-flushed. Then, we guess approximately the sizes of
the containers in $\B$. Note that 
in polynomial time we cannot guess the sizes of the containers exactly. Let $B\in\B$. Depending
on the items packed inside $B$, we guess different quantities for
$B$. 
\begin{itemize}
\item If there is only one single large item $i\in\R$ packed inside $B$ then
we guess $i$.
\item If $B$ contains only items from $\Rho$ then we guess the widest
item packed inside $B$. This defines our guessed width of $B$. Also,
we guess all items packed inside $B$ whose height is at least $\eps_2\opt$
(at most $O(1/\eps_2)$ many), denote them by $\R_{B}'$.
We guess the total height of the remaining items $\R_{B}\setminus\R'_{B}$
approximately by guessing the quantity $\hat{\height}(B):=\left\lfloor \frac{\height(\R_{B}\setminus\R'_{B})}{\eps_2\opt}\right\rfloor \eps_2\opt$.
Our guessed height for $B$ is then $\sum_{i\in\R'_{B}}\height(i)+\hat{\height}(B)$.
\item Similarly, if $B$ contains only items from $\Rve$ then we guess
the highest item packed inside $B$, which defines our guessed height
of $B$. Also, we guess all items packed inside $B$ whose width is
at least $\eps_3\W$ (at most $O(1/\eps_3)$ many),
denote them by $\R_{B}'$. We guess the total width of the remaining
items $\R_{B}\setminus\R'_{B}$ approximately by guessing the quantity
$\hat{\width}(B):=\left\lfloor \frac{\width(\R_{B}\setminus\R'_{B})}{\eps_3\opt}\right\rfloor \eps_3\opt$
and our guessed width of $B$ is then $\sum_{i\in\R'_{B}}\width(i)+\hat{\width}(B)$.
\item If $B$ contains only small items, then our guessed heights and widths
of $B$ are $\left\lfloor \frac{\height(B)}{\eps_4\opt}\right\rfloor \eps_4\opt$
and $\left\lfloor \frac{w(B)}{\eps_4\W}\right\rfloor \eps_4\W$,
respectively.
\end{itemize}

Note that here $\eps_2=\frac{\eps}{4\left|\B_{hor}\right|}$, $\eps_3=\frac{\eps_1}{4\left|\B_{ver}\right|}$ and $\eps_4=\mu$ are chosen so that the unpacked horizontal items, unpacked vertical items and unpacked small items due to container rounding can be packed in  containers $B_{hor}$ (defined below), $B^{*}$ and $B_{small}$, respectively (see Lemma~\ref{lem_item_placement_hor_vert_1} and Lemma~\ref{lem_item_placement_small_1}).

We have at most $O_{\eps}(1)$ containers and for each container $B\in\B$ we guess the type of container $B$ and its respective width and height (depending on the type) in  $n^{O_{\eps}(1)}$ time.

Additionally, we guess three containers $B_{med}$ of height $2\eps \opt$, $B_{hor}$ of height $\eps \opt$, and $B_{small}$
of height $27\eps\opt$ and width $\W$ each that we will use to
place all medium items, and to compensate errors due to inaccuracies
of our guesses for the sizes of the containers for horizontal and small
items, respectively. Let $\B'$ denote the guessed containers (including $B_{med}$,
$B_{hor}$, and $B_{small}$). Since $|\B'|=O_{\eps}(1)$ and
the containers in $\B'$ are not larger than the containers in $\B$, we can
guess a placement for the containers $\B'$ such that together with $\Rit$
they are guillotine separable. We place the containers $B_{med}$, $B_{hor}$,
and $B_{small}$ on top of the packing of rest of the containers in $\B'$, and $\Rit$.
\begin{lemma}\label{lem_guess_containers_poly}
In time $n^{O_{\eps}(1)}$ we can compute a placement for the
containers in $\B'$ such that together with the items $\Rit$, they are
guillotine separable.
\end{lemma}

Next, we place the vertical items. Recall that for each container $B\in\B$
containing items from $\Rve$ we guessed the items packed inside $B$
whose width is at least $\eps_3\W$. For each such container $B$
we pack these items into the container $B'\in\B'$ that corresponds to $B$.
With a similar technique as used for the generalized assignment problem
(GAP)~\cite{GalvezGIHKW21}, we place all but items with width at most $\eps_3 W$ for each container in $\Rve$. Further using the PTAS for this variant of GAP (see Lemma~\ref{gap_approx}), we can ensure that  items of total area at most $3\eps_5\cdot \opt\cdot W$ are not packed. Hence, items of total width at most $(3\eps_5/\delta)  W$ remain unpacked as each such item has height at least $\delta\opt$.
 We pack these remaining items into
$B^{*}$, using that each of them has a height of at most $\opt/2$
and that their total width is at most $|\B'|\cdot 2\eps_3\W+(3\eps_5/\delta)W\le\eps_1\W=\width(B^{*})$ (see Lemma~\ref{lem_item_placement_hor_vert_1} and Appendix \ref{section_constants}). In other words, we fail to pack some of the vertical items since we guessed the widths of the containers only approximately
and since our polynomial time approximation algorithm for GAP might not find the optimal packing.
%
We use a similar procedure for the items in $\Rho$ where instead
of $B^{*}$ we use $B_{hor}$ in order to place the unassigned items.
\begin{lemma}{\label{lem_item_placement_hor_vert}}
In time $n^{O_{\eps}(1)}$ we can compute a placement for all
items in $\Rve\cup\Rho$ in $B^{*}$, $B_{hor}$, and their corresponding
boxes in $\B'$. 
\end{lemma}

For the medium items we invoke again Lemma~\ref{lem:pack-med} and we place 
$B_{med}$ on top of the containers in $\B$ which
increases the height of the packing only by $2\eps\opt$. 

Finally, we use  NFDH again to pack the small items 
into their corresponding containers in $\B'$, which we denote by $\B'_{small}$, and $B_{small}$. We need $B_{small}$ due to inaccuracies of
NFDH and of our guesses of the container sizes. 
\begin{lemma}{\label{lem_item_placement_small}}
\label{lem:pack-small-1}In time $n^{O(1)}$ we can compute a placement
for all items in $\Rsm$ in $\B'_{small}$ and $B_{small}$. 
\end{lemma}
\begin{theorem}{\label{thm_{main_poly}}}
There is a $(3/2+\eps)$-approximation algorithm for the guillotine
strip packing problem with a running time of $n^{O_{\eps}(1)}$.
\end{theorem}

\section{Conclusion and Open problems}
We were able to show essentially tight approximation algorithms for \tsg~in both the polynomial and the pseudo-polynomial settings. This was possible due to the structure of the respective optimal packings since they are guillotine separable. However, it is unclear how to obtain such a structured packing in the general case of \ts, and the question remains to close the gap between the best approximation guarantee of $({5}/{3}+\eps)$ and the lower bound of ${3}/{2}$.
 Another interesting open problem related to guillotine cuts is to find out whether there exists a PTAS for the 2D guillotine geometric knapsack (2GGK) problem.

	\bibliographystyle{plain}
	\bibliography{bibliography}
	
	\appendix

\newpage
\section{PPTAS omitted details}
\label{sec_omitted_details}

\subsection{ A pseudo-polynomial time approximation scheme}
\label{subsec:algo-PPTAS} 
We describe the rest of the details of the PPTAS in this subsection.
In \Cref{subsec:pptasmed} we describe the packing of medium items.
Then in \Cref{subsec:pptasskew} we describe the packing of horizontal, tall, and vertical items.
Thereafter, in \Cref{subsec:pptassmall} we describe the packing of small items.
Finally, in \Cref{subsec:pptasdepen} we remove the dependency on $\H$ in the runtime.

\subsubsection{Packing of medium items}
\label{subsec:pptasmed}
\begin{lemma}[Restatement of Lemma~\ref{lem:pack-med}]
\label{lem:pack-med_1}In time $n^{O(1)}$ we can find a nice placement
of all items in $\Rme$ inside one container $B_{med}$ of height $2\eps\opt$
and width $\W$.
\end{lemma}
\begin{proof}
By choosing the function $f$  in  Lemma~\ref{class_1} appropriately (see Appendix~\ref{section_constants} for details on how to choose $f$), we ensure that the total area of medium items is at most $\eps\opt\cdot W$. Now, since the height of the medium items is at most $\delta\opt\leq \eps\opt$, we use Steinberg's algorithm (see Lemma~\ref{steinberg_alg}) to pack them in a container $B_{med}$ of height $2\eps\opt$ and width $W$, i.e., according to Lemma~\ref{steinberg_alg}, we choose $w=W$ and $h=2\eps\opt$ and observe that $2h_{\max}\leq 2\eps\opt= h$. 
\end{proof}

\subsubsection{Packing of tall, vertical, and horizontal items}
\label{subsec:pptasskew}
\begin{lemma}
\label{lem:DP}There is an algorithm with a running time of $(nh_{\max})^{|\B_{hor}|}$
that computes a packing of the items in $\Rho$ into the containers $\B_{hor}$.
Similarly, there is an algorithm with a running time of $(nW)^{|\B_{tall+ver}|}$
that computes a packing of the items in $\Rit\cup\Rve$ into the containers
$\B_{tall+ver}$.
\end{lemma}
\begin{proof}


 We need to pack the  items nicely in the containers for which we convert the instance to an instance of the Maximum Generalized Assignment problem with one bin per container and the size of $j$-th container $B_j$ is $a(B_j)=w(B_j)\times h(B_j)$. 
We build an instance of GAP as follows. There is one item $R$ per rectangle $R \in I$, with profit $a(R)$. For each
horizontal container $B_j$, we create a knapsack $j$ of size $S_j := h(B_j )$. Furthermore, we define
the size $s(R, j)$ of a horizontal rectangle $R$ w.r.t. knapsack $j$ as $h(R)$ if $h(R) \leq h(B_j )$ and $w(R) \leq w(B_j )$.
Otherwise $s(R, j) = \infty$ (meaning that $R$ does not fit in $B_j$). The profit  for rectangle $R$, $p_R=a(R)$ throughout all the bins. The construction for vertical
containers is symmetric.
  ~We now use the exact algorithm for GAP mentioned in Lemma~\ref{gap_alg_1} and finish the packing of all the vertical items in the respective containers in $(nW)^{|\B_{tall+ver}|}$ time and pack the horizontal items in the respective containers in $(nh_{\max})^{|\B_{hor}|}$ time. Note that the dependence on $h_{\max}$ in the running time of packing the horizontal items in their respective containers can be significantly reduced via Lemma~\ref{lem:discretize}.
\end{proof}

\subsubsection{Packing of small items}
\label{subsec:pptassmall}
Finally, we need to pack the small items. Let $\B_{small}\subseteq\B$
denote the boxes in $\B$ that contain small items in the packing
due to Lemma~\ref{lem:structural}. Note that for each small item
$i\in\Rsm$ only some of the boxes in $\B_{small}$ are allowed since
we require that $w_{i}\le\eps\cdot w(B)$ and $h_{i}\le\eps\cdot h(B)$
if $i$ is nicely packed inside a box $B$. 

On a high level, we pack the small items in two steps, similarly as
in \cite{nadiradze2016approximating, GGIK16, JansenR19}.
 First, we assign at least $I'_{small}\subset I_{small}$ items
 to the containers $\B_{small}$ such that for each box $B\in\B_{small}$,
each item $i\in\Rsm$ assigned to $B$ satisfies that $w_{i}\le\eps\cdot w(B)$
and $h_{i}\le\eps\cdot h(B)$ and the total area of the items assigned
to $B$ is at most $\height(B)\cdot\width(B)$, and the items not assigned have area at most $\eps \opt\cdot W$. Then, for each box
$B\in\B_{small}$ we try to pack its assigned items greedily using
NFDH~\cite{coffman1980performance, Galvez0AJ0R20} (see Lemma~\ref{nfdh_alg1} and Lemma~\ref{nfdh_alg2}). For the items in $\Rsm$ which remain unassigned (by our
first step or by NFDH), we will show that their total area
is at most $3\eps\opt\cdot\W$. Therefore, we can pack them
using NFDH into an additional container $B_{small}$ of height $9\eps\opt$
and width $\W$. 
\begin{lemma}
\label{lem:pack-small}There is an algorithm with a running time of
$(n|\B_{small}|)^{O(1)}$ that packs all items in $\Rsm$ into $\B_{small}$
and an additional container $B_{small}$ of height $9\eps\opt$ and width
$\W$. 
\end{lemma}
\begin{proof}
By Lemma \ref{lem_small_rearrange} we know that all items in $I_{small}$ can be nicely packed in $\B_{small}$. Hence, we convert this problem of nicely packing items in $I_{small}$ into containers in $\B_{small}$ to an instance of GAP  with one bin per container and the size of $j$-th container $B_j$ is $a(B_j)=w(B_j)\times h(B_j)$. 
We build an instance of GAP as follows. There is one item $R$ per rectangle $R \in I_{small}$, with profit $a(R)$. For each container $B_j$, we create a knapsack $j$ of size $S_j := a(B_j)$. Furthermore, we define
the size $s(R, j)$ of a horizontal rectangle $R$ w.r.t. knapsack $j$ as $a(R)$ \textit{iff} $h(R) \leq \eps h(B_j )$ and $w(R) \leq \eps w(B_j )$.
Otherwise $s(R, j) = \infty$ (meaning that $R$ is not small enough to nicely fit in $B_j$). The profit  for rectangle $R$, $p_R=a(R)$ throughout all the bins. Similarly as in Lemma \ref{lem:DP}, we now use the exact algorithm for GAP which runs in pseudo-polynomial time. Hence, the total area of nicely assigned small items to the containers in $\B_{small}$ is at least  the area of such items packed according to the packing in Lemma \ref{lem_small_rearrange}. Then we use NFDH to pack all items in their respective assigned containers according to GAP, in which case we might miss out on at most $3\eps$ fraction of the total area. All such items are packed in another box $B_{small}$ of height at most $9\eps \opt$ and width $W$ by NFDH.
\end{proof}

We pack $B_{small}$ on top of $B_{med}$ and obtain a packing of
height $(1+O(\eps))\opt$. This yields an algorithm with a running
time of $(n\H\W)^{O_{\eps}(1)}$. 

\subsubsection{Removal of dependency on $\H$}
\label{subsec:pptasdepen}
With a minor modification we
can remove the dependence on $\H$ and obtain a running time of $(n\W)^{O_{\eps}(1)}$,
see Lemma~\ref{lem:discretize} for details.

\begin{lemma}
\label{lem:discretize}By increasing the height of the optimal solution
by at most a factor $1+\eps$, we can assume that $\height_{i}\le {\lceil n/\eps \rceil}$
for each item $i\in I$ and that the corners of each item are placed
on integral coordinates in the optimal solution.
\end{lemma}
\begin{proof}

Assume that $h_{\max}>n/\eps$ as otherwise the lemma statement holds trivially. Let the height of the optimum packing be OPT. For each item $i\in I$, we normalize its height $h_i$ to $h_{i}'={\left\lceil\frac{h_i n}{\varepsilon h_{\max}}\right\rceil}$. Let the height of the optimum packing for this new instance be $\opt'$ and call this optimum packing $P'$. 

Consider another instance where for each item $i\in I$, we normalize its height $h_i$ to $h_{i}''=\frac{h_i n}{\varepsilon h_{\max}}$. Note that this new instance can have heights for the rectangular items which are not necessarily integers. Let its optimum packing have height $\opt''$ and call this optimum packing $P$. Now, consider the packing $P$ and for each item $i\in I$ round its height $h_i$ to $h_i'$. Call this packing $P_1$. Then we claim that 
\[\opt'\leq h(P_1)\leq (1+\eps)\opt''\]
where $h(P_1)$ denotes the height of $P_1$.

The first inequality holds since heights of items in $P_1$ are at least as much as heights of items in the packing $P'$. For the second inequality,
 \[h(P_1)-\opt''\leq \sum_{i}(h_{i}'-h_{i}'')\leq n(1)\leq \eps(n/\eps)\leq \eps h_{\max}''\leq \eps \opt'' \]
The last inequality holds since $h_{\max}''$ is a trivial lower bound on the height of the optimal packing $\opt''$ for the instance with heights $h_i''$ for each item $i\in I$. 
Now, consider the packing $P'$ but with the heights of each item $i\in I$ to be their original height $h_i$. Call this packing $P'''$.
 \[h(P''')\leq  \opt' \left(\frac{\eps h_{\max}}{n}\right)\leq (1+\eps) \opt'' \left(\frac{\varepsilon h_{\max}}{n}\right)\leq (1+\eps) \opt \]
Thus, the required conditions of the lemma are satisfied for the respective heights of each item $i\in I$ normalized to $h_i'$.
\end{proof}

We run our pseudo-polynomial time approximation algorithm on the resulting instance for
which $\H\le n/\eps+1$ holds. Thus, we obtain a running time of
$(n\W)^{O_{\eps}(1)}$.

\begin{theorem}[Restatement of Theorem~\ref{thm:pptas}]
\label{thm:pptas_1}
There is a $(1+\eps)$-approximation algorithm for the guillotine
strip packing problem with a running time of $(n\W)^{O_{\eps}(1)}$.
\end{theorem}
\begin{proof}
Follows from Lemma \ref{lem:structural}, the aforementioned algorithm and by choosing the parameter to be $\eps/33$ to finally achieve a $(1+\eps)$-approximation.
\end{proof}

\subsection{ Structural Lemma 1 omitted proofs}
\label{sec:algo_new}

\begin{lemma}[Restatement of Lemma \ref{lem:rearrange-2}]
\label{lem:rearrange-3}There exists a partition of $[0,\W]\times[0,(1+2\eps)\opt]$
into a set $\B_2$ of $O_{\eps}(1)$ boxes such that 
\begin{itemize}
\item the boxes in $\B_2$ are pairwise non-overlapping and admit a guillotine
cutting sequence,
\item the items in $\Rhard$ can be packed into $\B_2$ such that they are
guillotine separable and each box $B\in\B_2$ either 
\begin{itemize}
\item contains only items from $\Rit\cup\Rbg\cup\Rve$, or
\item contains only items from $\Rho$.
\end{itemize}
\item Any item $i\in \Rit\cup\Rbg\cup\Rve$ has height $h_i'=k_i \delta^2 \opt$ for some $k_i\in \mathbb{Z}$ such that $0<k_i\leq 1/\delta^2+1$.
\end{itemize}
\end{lemma}
\begin{proof}
The proof consists of two steps. 
The first step is to process the individual $\Lc$-compartments (from Lemma~\ref{lem:rearrange}) to obtain $O_{\eps}(1)$ $\Vc$-compartments. Intuitively we do this by shifting the vertical leg of the $\Lc$ by $O(\eps) \opt$ amount vertically and allowing division into  a constant number of new $\Vc$-compartments. %
In the second step, we obtain packing of items in $O_{\eps}(1)$ $\Hc$-compartments. Here, $\Hc$-compartments are boxes with height at most $\delta\frac{\text{OPT}}{2}$ and $\Vc$-compartments are boxes with width at most $\delta\frac{W}{2}$. Hence, a $\Hc$-compartment cannot contain items in $\Rit\cup\Rve$ and a $\Vc$-compartment cannot contain items in $\Rho$.

We now explain the process in detail by considering an $\Lc$-compartment $L$ defined by a simple rectilinear polygon with six edges $e_{0},e_{1},\dots,e_{5}$ as given in Definition \ref{def:Lc}. Let the vertical and horizontal legs be $L_V$ and $L_H$, respectively. W.l.o.g we can assume that the horizontal leg lies on the bottom right of the vertical (other cases can be handled analogously). Let $p_1$ be the bottommost point where $L_V$ and $L_H$ intersect. Let $e_5$ (resp. $e_4$) be the shorter vertical (resp. horizontal) edge of $L_V$ (resp. $L_H$). Let the height and width of $L$  be $h(L),\ w(L)$, respectively. 

We first shift each item in $L_V$ vertically by an $\eps h(L)$ amount, see Figure \ref{fig_rayshoot}. We then create the boundary curve with $\frac{2}{\eps}$ bends. We begin at $p_1$, i.e., the leftmost bottommost point of  $L$. Using a ray shooting argument, continue drawing in the horizontal direction till the ray hits an horizontal item. Then we bend the ray and move it upwards till it hits one of the vertical items. Recursively repeating this process till the ray hits one of the bounding edges $e_5$ or $ e_4$. Let this point be $p_2$. The trace of the ray defines the boundary curve $C_{H,V}$ between $L_H$ and $L_V$, starting at $p_1$ and ending at $p_2$. 
Now we argue that the number of bends in the given boundary curve $C_{H,V}$ is $\frac{4}{\eps}$. It is clear that the number of bends in $C_{H,V}$ will be twice as much as the number of distinct vertical paths in the curve $C_{H,V}$. Since each vertical item is shifted vertically by an $\eps h(L)$ height initially, each such vertical path will be at least $\eps h(L)$ in length. This establishes that the number of such distinct vertical paths can be at most $\frac{1}{\eps}$ which bounds the number of bends by at most $2\frac{1}{\eps}+2 \leq \frac{4}{\eps}$ bends. 

We then further create the vertical compartments by first extending the projections from bends of boundary curve $C_{H,V}$ in vertical direction. If the vertical projection does not intersect with any of the vertical items, this can be considered as a guillotine cut and used to separate the items. If any such vertical projection intersects any item in $L_V$, we cannot divide the vertical leg using this line. Instead, we exploit the fact that vertical leg has first stage of guillotine cuts separating them since the $\Lc$ considered is pseudo-guillotine separable. Now, we instead consider the two nearest consecutive guillotine cuts, one on the immediate left and other on the immediate right of the vertical projection. The arrangements in previous steps and pseudo-guillotine separability ensure that any two consecutive guillotine cuts in the vertical leg are at the same level and form a box. This divides the vertical leg into one additional vertical box.
Thus at the end of this procedure, we will have at most $\frac{24}{\eps}$.  

Similarly we extend the horizontal projections from bends of boundary curve $C_{H,V}$ dividing the horizontal leg $L_H$ into $O_{\eps}(1)$ $\Hc$-compartments.

We do the procedure for all the $\Lc$-compartments which  are at most $4\delta$ many. The total height added considering $\eps h(L)$ vertical shifts corresponding to each $L$ which is a $\Lc$-compartment, can be at most $\eps\opt$ since the cumulative height of $\Lc$-compartments on the top of each other can at most be $\opt$. This gives us $O_\eps(1)$ new $\Hc$ and $\Vc$-compartments. Since the items in  $ \Rit\cup\Rbg\cup\Rve$ have height at least $\delta \text{OPT}$, there can be at most $1/\delta$ items  in  $ \Rit\cup\Rbg\cup\Rve$ stacked on top of each other. Hence, rounding the heights of these items to multiples of $\delta^2 \opt$ can only increase the height of the packing by at most $\delta^2 \opt$ per item. Since there are at most $1/\delta$ items in $ \Rit\cup\Rbg\cup\Rve$ stacked on top of each other, the overall increase in the height of the packing can be at most $\delta^2\opt\cdot(1/\delta)\leq \delta\opt\leq \eps\opt$. Hence, the total  increase in the height of the packing is at most $2\eps\opt$.
\end{proof}

\begin{lemma}[Restatement of Lemma \ref{lem_small_rearrange}]
\label{lem_small_rearrange_1}
There exists a set of $O_{\eps}(1)$ boxes $\B_{small}$, all
contained in $[0,\W]\times[0,(1+14\eps)\opt]$, such that the
boxes in $\B_{hard}\cup\B_{small}$ are non-overlapping and guillotine
separable and the items in $\Rsm$ can be placed nicely into the boxes
$\B_{small}$ .
\end{lemma}
\begin{proof}
Consider a vertical container $B$ in which the items are nicely packed. Now we rearrange items such that items are placed side by side in a non-increasing order of their heights from left to right with no gap between the adjacent items and the items touch the bottom edge of $B$. Now group the items in $B$ based on their heights. For all $i\in[(1-{\delta})/\delta^2]$, group $i$ contains the items with heights in the range $({\delta}{\opt}+(i-1)\cdot\opt\cdot\delta^2, {\delta}{\opt}+i\cdot\opt\cdot\delta^2]$. Now for each group $i$, consider the minimal rectangular region containing all the items in group $i$ and make it a container. Repeat a similar process for the horizontal containers too. Now the ratio of the area of the rectangles in a container to the area of the container is at least $1-\delta$. 

Now, we form a non-uniform grid by extending the edges of the containers until it hits a container or the edges of the half-strip $[0,W]\times[0,\infty)$. Note that the empty grid cells are guillotine separable as the containers when considered as pseudo-items are guillotine separable and the guillotine cuts coincide with one of the edges of the containers. 
 Now we can choose $\mu,\delta$ appropriately such that the total area of empty grid cells with height less than $\eps'\opt$ or width less than $\eps'\W$ is $\delta^2\cdot\opt\cdot \W$ where $\eps'=\mu/\eps$.  Note that the total number of containers which have height more than $\eps'\opt$ and width more than $\eps'\W$ is $O_{\eps}(1)$.
 Now we can pack small rectangles using NFDH in the small containers. Then the total area of the small items not packed nicely is at most $3\eps\opt\cdot W$. Hence, using NFDH we can pack these remaining small rectangles of area in another container $B_{small}$ of height at most $9\eps \opt$ and width $W$. Thus we have at most $O_{\eps}(1)$ small containers which have a nice packing. Thus, to calculate the extra height on top of OPT,
\begin{enumerate}
 \item We have a box $B_{small}$ of height $9\eps\opt$ and width $W$,
\item Additional containers in $\B_{hor}$ due to resource augmentation, which can be packed in a box of height $3\eps\opt$ and width $W$,
\item An increase of $2\eps\opt$ to account for rounding the items in $\Rit\cup\Rbg\cup\Rve$ and shifting items in each $\Lc$-compartment vertically upward according to Lemma~\ref{lem:rearrange-3}.
\end{enumerate}
\end{proof}

\section{$\boldsymbol{(\frac{3}{2}+\eps)}$-approximation omitted proofs}
\label{sec:omitted_proofs2}

For the polynomial $(3/2+\eps)$-approximation algorithm, we show first how to guess  value $\opt\textsuperscript{$\prime$}$ such that $\opt\leq \opt\textsuperscript{$\prime$}\leq (1+\varepsilon)\opt$ in polynomial time. We do this  by computing a $2$-approximation APX by Steinberg's algorithm~\cite{steinberg1997strip} (which is a guillotine separable packing~\cite{khan2021guillotine}) and then run our algorithm for all values of $\opt\textsuperscript{$\prime$}=(1+\varepsilon)^j\frac{\text{APX}}{2}$ which fit in the range $[\frac{\text{APX}}{2},\text{APX}(1+\varepsilon)]$, i.e., for $j\in\mathbb{Z}$ such that $1\leq j\leq 1+{\lfloor \log_{(1+\eps)}2  \rfloor}$. One of these values will satisfy the claim.

\begin{lemma}[Restatement of Lemma \ref{lem:structural-1}]
\label{lem:structural-1_app}There exists a set $\B$ of $O_{\eps}(1)$
pairwise non-overlapping and guillotine separable boxes that are all
placed inside $[0,\W]\times[0,(1+16\eps)\opt)$ and a partition
$\R=\bigcup_{B\in\B}\R_{B}$ such that for each $B\in\B$ the items
in $\R_{B}$ can be placed nicely into $B$. Also, for each box $B\in\B$
containing an item from $\Rit$ we have that the bottom edge of $B$
intersects the line segment $[0,\W]\times\{0\}$. 
\end{lemma}
\begin{proof}
The proof follows from the proof of Lemma~\ref{lem:structural} and the algorithm described at the start of Section~\ref{subsec:structural_2} for shifting the tall items using guillotine cuts until their bottom edges intersect the bottom of the half-strip. First of all, for a packing, in a stage $k$ (be it vertical or horizontal) in a guillotine cutting sequence  if a box $B$ is subdivided into $j$ boxes $B_1,B_2,...,B_j$ by guillotine cuts, then swapping any of these boxes does not affect the guillotine separability of the packing. Now, for the sake of contradiction assume that at the end of a valid guillotine cutting sequence, after applying the algorithm mentioned before for swapping tall items, a tall item $i$ is such that $bottom(i)>0$. Then, since $i$ has been separated from all other items by guillotine cuts at the end of the cutting sequence, assume that the box $B$ which contains it (on account of the guillotine cuts) at the end is exactly the size of $i$ itself. Now, to have such a box $B$ containing $i$, at some stage in the guillotine cutting sequence we would have had a horizontal cut which intersects the bottom edge of $i$. But then $i$ would have been swapped across that cut with the resulting guillotine box below $B$ in that stage (by virtue of the algorithm), which is a contradiction. Hence, the lemma statements follows.
\end{proof}

\begin{lemma}[Restatement of Lemma \ref{lem_guess_containers_poly}]
\label{lem_guess_containers_poly_1}
In time $n^{O_{\eps}(1)}$ we can compute a placement for the
containers in $\B'$ such that together with the items $\Rit$, they are
guillotine separable.
\end{lemma}
\begin{proof}
We guess the structure as guaranteed by Lemma \ref{lem:structural-3/2}. We first sort all the tall items according to non-increasing order of their heights and place them in that order starting from the left end of the strip $[0,W]\times [0,\infty)$ such that their bottom edge touches the line segment $[0,W]\times \{0\}$. Let the tall items in this order be $I_{t_1}, I_{t_2},..., I_{t_k}$. Then as mentioned in Section \ref{subsec_alg_poly}, we have that sizes of the containers belong to a set (let's say $S$) that can be computed in $n^{O_{\eps}(1)}$ time. The height $h^{*}$ is one among the heights of the tall items and we have at most $O_{\eps}(1)$ containers in $\B'$. Hence, the position  of the bottom of a container can be either of the following:
\begin{enumerate}
\item A linear combination of heights from the set $S$.
\item Sum of $h^{*}$ with a linear combination of heights from the set $S$.
\item Sum of $h^{*}+\frac{1}{2}\text{OPT}$ with a linear combination of heights from the set $S$.
\end{enumerate}
Hence, the possible positions for the bottoms of the containers can be at most $n^{O_{\eps}(1)}$ many and which can be computed in $n^{O_{\eps}(1)}$ time. Similarly, the positions for the left end of the containers can be either of the following:
 \begin{enumerate}
\item A linear combination of widths from the set $S$.
\item Sum of $\sum_{i=1}^{i=j}w(I_{t_i})$ with a linear combination of widths from the set $S$ for some $j\leq t$.
\item Sum of $\sum_{i=1}^{i=j}w(I_{t_i})$ with a linear combination of widths from the set $S$ for some $j\leq t$ and $\eps_1 W$.
\end{enumerate}
Hence, the possible positions for the left end of the containers can be at most $n^{O_{\eps}(1)}$ many and which can be computed in $n^{O_{\eps}(1)}$ time. The three containers $B_{med}$, $B_{hor}$, and $B_{small}$
of height $O(\eps\opt)$ and width $\W$ are placed on top of the packing. The guillotine separability of this placement of containers in $\B'$ along with the tall items is checked with the help of the algorithm in Lemma \ref{lem_guillo1} in polynomial time.
\end{proof}

\begin{lemma}[Restatement of Lemma \ref{lem_item_placement_hor_vert}]
\label{lem_item_placement_hor_vert_1}
In time $n^{O_{\eps}(1)}$ we can compute a placement for all
items in $\Rve\cup\Rho$ in $B^{*}$, $B_{hor}$, and their corresponding
boxes in $\B'$. 
\end{lemma}
\begin{proof}
We reduce the given instance to an instance of GAP exactly as is done in Lemma \ref{lem:DP}. The only difference here is that for the vertical items, due to approximation in the widths of the containers as done in Section \ref{subsec_alg_poly}, per container we are not able to pack at most $\eps_3 W$ width of items which are of height at most $\opt/2$. The combined width of such items which are unable to the packed in the containers is at most $2\eps_3\cdot \left|\B_{ver}\right| W$. For each vertical container, we guess all items packed inside $B$ whose width is at least $\eps_3 W$
(at most $O(1/\eps_3)$ many). We do a symmetric procedure for the horizontal containers. 
 Since we have at most $O_{\eps}(1)$ containers in total, we then make use of the PTAS for this variant of GAP (see Lemma \ref{gap_approx}) which ensures that for the vertical items, we are unable to pack at most $3\eps_5$ fraction of the total area. This means that the total width of such items is at most $(3\eps_5/\delta) W$. Hence, for the appropriate choice of $\eps_3,\eps_5$, after computing the packing of vertical items we might have unpacked vertical items with width at most 
\[2\eps_3\cdot \left|\B_{ver}\right|W+(3\eps_5/\delta) W\leq \left(\frac{2\eps_1|\B_{ver}|}{4|\B_{ver}|}\right)W+\left(\frac{3\eps_1\delta}{6\delta}\right)W\leq \eps_1 W\] 
which is due to inefficiency of the algorithm for GAP and due to the approximation of container widths (see Appendix~\ref{section_constants}). All of such items can now be packed in $B^{*}$. Similarly for the horizontal items, the total area of the items which are unable to be packed because of the inefficiency of GAP is at most $3\eps_6 \opt\cdot W$. Also due to container rounding for the horizontal items, we are unable to pack items with height at most $2\eps_2\cdot \left|\B_{hor}\right|\opt$. These items have height at most $\eps\opt$ by choosing $\eps_2,\eps_6=O_{\eps}(1)$ appropriately (see Appendix~\ref{section_constants}). We pack all such items simply by stacking them on top of each other in $B_{hor}$ which adds an additional height of $\eps\opt$ apart from the $3\eps\opt$ height from resource augmentation and width $W$.
\end{proof}

\begin{lemma}[Restatement of Lemma \ref{lem_item_placement_small}]
\label{lem_item_placement_small_1}
\label{lem:pack-small-1}In time $n^{O(1)}$ we can compute a placement
for all items in $\Rsm$ in $\B'_{small}$ and $B_{small}$. 
\end{lemma}
\begin{proof}
The proof follows in the same vein as the one for Lemma \ref{lem_small_rearrange}, only that now we have to additionally account for inaccuracies in the container sizes. There are at most $\eps^2/\mu^2$ containers with width at least $(\mu/\eps) W$ and height at least $(\mu/\eps) \opt$ and one container of height at most $O(\eps \opt)$ and width $W$ on top of the packing according to Lemma \ref{lem_small_rearrange}. By the container rounding mentioned in the algorithm in Section \ref{subsec_alg_poly}, for each such container $B_i$ we are unable to pack at most $(\eps_4 +\mu)(h(B)\cdot W+w(B)\cdot \opt)$ area of small items. Thus, the ratio of this area of area of $B$ is given by $(\eps_4+\mu)W/w(B)+2(\eps_4+\mu)\opt/h(B)$ and this quantity is maximized when we consider the minimum values of $h(B)$ and $w(B)$ which are $h(B)=(\mu/\eps) \opt$ and $w(B)=(\mu/\eps) W$ and hence, we get that this ratio of area of small items which we are unable to pack due to container rounding is at most $4\eps$ for $\eps_4= \mu$. All such items can be packed on top of the packing in another container $B_{extra}$ of height $12\eps \opt$ and width $W$ using NFDH (two other containers each of height $6\eps\opt$ and width $W$ are used to account for inaccuracies due to NFDH packing and some unassigned area as in the PPTAS for small items).

Now, that we shown the existence of such  a packing, for the algorithmic part we follow the same procedure of converting to an instance of GAP and assigning items according to the polynomial time approximation algorithm for GAP (Lemma \ref{gap_approx}). We might lose out on at most $2\eps\opt\cdot W$ area due to inefficiency in the algorithm. Such small items can be packed in a container of height $6\eps\opt$ and width $W$ using NFDH along with the items we are unable to pack due to container rounding. Note that there are inaccuracies due to packing using NFDH nicely as well which accounts for at most $2\eps\opt\cdot W$ area of small items not being packed. Such items can be packed using NFDH again using another container of height $6\eps \opt$ and width $W$ on top of the packing. Thus, we just need one container of height at most $27\eps\opt$ and width $W$.
\end{proof}

\begin{theorem}[Restatement of Theorem \ref{thm_{main_poly}}]
{\label{thm_{main_poly_1}}}
There is a $(3/2+\eps)$-approximation algorithm for the guillotine
strip packing problem with a running time of $n^{O_{\eps}(1)}$.
\end{theorem}
\begin{proof}
Follows from Lemma \ref{lem:structural-3/2} and the algorithm from Section \ref{subsec_alg_poly} and guessing $\opt$ to within a $1+\eps$ factor  as mentioned at the beginning of Section~\ref{subsec_alg_poly}. Here we take the parameter to be $\eps/79$ to finally achieve a $(1+\eps)$-approximation.
\end{proof}



\section{Tools}{\label{tools}}
\subsection{Maximum Generalized Assignment Problem}
\label{subsec_GAP}
In this section we show that there is a pseudo-polynomial time algorithm for the Maximum Generalized Assignment Problem (GAP), if
the number of bins is constant. In GAP, we are given a set of $k$ bins with capacity constraints and a set of
$n$ items that have a possibly different size and profit for each bin and the goal is to pack a maximum-profit
subset of items into the bins. Let us assume that if item $i$ is packed in bin $j$, then it requires size $s_{ij} \in \mathbb{Z}$ and
profit $p_{ij} \in \mathbb{Z}$.

Let $C_j$ be the capacity of bin $j$ for $j\in[k]$. Let $p(\textup{OPT})$ be the cost of the optimal assignment.
\begin{lemma}[{\cite{GalvezGIHKW21}}]
\label{gap_alg_1}
There is a $O(n\prod\limits_{j=1}^{k}C_j)$ time algorithm for the maximum generalized assignment problem with $k$ bins and returns a solution with maximum profit $p(\textup{OPT})$.
\end{lemma}
\begin{proof}
For each $i\in[n]$ and $c_j\in[C_j]$ and $j\in[k]$, let $S_{i,c_1,...,c_k}$ denote a subset of the set of items $\{1,2,...,i\}$ packed into the bins such that the profit is maximized and the capacity of bin $j$ is at most $c_j$. Let $P[i,c_1,c_2,...,c_k]$ denote the profit of $S_{i,c_1,...,c_k}$. Clearly $P[i,c_1,c_2,...,c_k]$ is known for all $c_j\in[C_j]$ for $j\in[k]$. Moreover we define $P[i,c_1,c_2,...,c_k]=0$ if $c_j<0$ for any $j\in[k]$. We can compute the value of $P[i,c_1,c_2,...,c_k]$ by a dynamic program that exploits the following recurrence:
\begin{multline}
\mspace{150mu} P[i,c_1,c_2,...,c_k]=\max \{P[i-1,c_1,c_2,...,c_k], \\ \mspace{290mu}\max\limits_{j}\{p_{ij}+P[i-1,c_1,...,c_j-s_{ij},...,c_k]\}\}  \notag\\
\end{multline}
This dynamic program clearly runs in $O(n\prod\limits_{j=1}^{k}C_j)$ corresponding to the entries in the DP table.
\end{proof}

We also have the following algorithm for GAP running in time $n^{O_{\eps}(1)}$ which ensures a profit of at least $(1-O(\eps))p(\opt)$ if the number of bins is constant.

\begin{lemma}[\cite{GalvezGIHKW21}]
\label{gap_approx}
There is a $O((\frac{1+\varepsilon}{\varepsilon})^k n^{k/\eps^2 +k+1})$ time algorithm for the maximum generalized assignment problem with $k$ bins which returns a solution with profit at least $(1-3\eps)p(\textup{OPT})$ for any fixed $\varepsilon>0$.
\end{lemma}

\subsection{Next Fit Decreasing Height}
\label{subsec_nfdh}
One of the most recurring tools used as a subroutine in countless results on geometric problems is the Next Fit Decreasing Height (NFDH) algorithm which was originally analyzed in \cite{coffman1980performance} in the context of strip packing. We will use two standard results related to NFDH for our requirements. We will provide the proofs of the both these results for sake of completeness.

Suppose we have a set of rectangles $I^{\prime}$.  NFDH computes in polynomial time a packing (without rotations) of  $I^{\prime}$ as follows. It sorts the items $i\in I^{\prime}$ in non-increasing order of their heights $h_i$ ( corresponding widths $w_i$) and considers items in that order $i_1,...,i_n$ (let's call this list $L$). Let width of the strip be $W$ and define $A(L)=\sum_i h_i\cdot w_i$. Then the algorithm works in rounds $j\geq 1$. At the beginning of round $j$, it is given an index $n(j)$ and a horizontal segment $L(j)$(level $j$) going from the left to the right of $C$. Initially $n(1)=1$ and $L(1)$ is the bottom side of $C$ which is the first level. In round $j$, the algorithm packs a maximal set of items $i_{n(j)},...,i_{n(j+1)-1}$, with the bottom side touching $L(j)$ one next to the other from left to right. The segment $L(j+1)$ is the horizontal segment containing the top of $i_{n(j)}$ and extending from the left to the right of $C$. The space between two consecutive levels will be called a block. The algorithm continues in this manner till all items in $L$ are packed. Hence, we have a sequence of blocks $B_1,...,B_k$ where the index increases from the bottom to the top of the packing and $B_k$ is the last block in the packing of rectangles in $L$. Let $A_i$ denote the total area of rectangles in block $B_i$ and let $H_i$ denote the height of block $B_i$. By way of our algorithm we have that $H_1\geq H_2\geq ...\geq H_k$. We state the result regarding the height of this packing through the following lemma. 

\begin{lemma}[\cite{coffman1980performance}]
\label{nfdh_alg1}
For a list $L$ of rectangles ordered by nonincreasing height,
\[\text{\textnormal{NFDH}}(L)\leq 2A(L)/W+H_1\leq 3\textup{OPT},\]
where \textup{OPT} denotes the height of the  optimal packing.
\end{lemma}
\begin{proof}
For each $i$, let $x_i$ denote the width of the first rectangle in $B_i$ and $y_i$ be the total width of the rectangles in $B_i$. For each $i<k$, the first rectangle in $B_{i+1}$ does not fit in $B_i$. Therefore $y_i+x_{i+1}>W$, $1\leq i<k$. Since each rectangle in $B_i$ has height at least $H_{i+1}$, and the first rectangle in $B_{i+1}$ has height $H_{i+1}$, $A_i+A_{i+1}\geq H_{i+1}(y_i+x_{i+1})> H_{i+1}W$. Therefore,
\begin{align*}
\text{NFDH}(L)&=\sum\limits_{i=1}^{i=k} H_i\leq H_1+ \sum\limits_{i=1}^{i=k-1}A_i/W+\sum\limits_{i=2}^{i=k}A_i/W\\
           &\leq H_1+2A(L)/W \\
           &\leq 3\text{OPT}
\end{align*}
\end{proof}

Since $\text{OPT}\geq H_1$ and $A(L)/W\leq \text{OPT}$, NFDH is a polynomial time $3$-approximation. The second result is in the context of using NFDH to pack items inside a box. Suppose you  are given a box $C$ of size $w\times h$ and a set of items $I^{\prime}$ each one fitting in the box ( without rotations).  NFDH computes in polynomial time a packing ( without rotations) of $I^{\prime \prime }\subset I^{\prime}$ as mentioned before. But unlike the previous lemma the process halts at round $r$ when either all items are packed or $i_{n(r+1)}$ cannot be packed in the box. The following lemma describes the result regarding this packing.

\begin{lemma}[\cite{Galvez0AJ0R20}]
\label{nfdh_alg2}
Assume that for some parameter $\eps\in(0,1)$, for each $i\in I^{\prime}$ one has $w_i\leq \varepsilon w$ and $h_i\leq \varepsilon h$. Then \textup{NFDH} is able to pack in $C$ a subset $I^{\prime \prime }\subset I^{\prime}$ of area at least $a(I^{\prime \prime })\geq \min\{a(I^{\prime}),(1-2\varepsilon)w\cdot h\}$. In particular, if $a(I^{\prime})\leq (1-2\varepsilon)w\cdot h$, all items are packed.
\end{lemma}
\begin{proof}
The claim trivially holds if all items are packed. Thus suppose that is not the case. Observe that $\sum_{j=1}^{r+1}h(i_{n(j)})>h$, otherwise item $i_{n(r+1)}$ would fit in the next shelf above $i_{n(r)}$: hence, $\sum_{j=2}^{r+1}h(i_{n(j)})>h-h(i_{n(1)})\geq (1-\varepsilon)h$. Observe that the total width of items packed in each round $j$ is at least $w-\varepsilon w$, since $i_{n(j+1)}$, of width at least $\varepsilon w$ does not fit  to the right of $i_{n(j+1)-1}$. It follows that the total area of items packed in round $j$ is at least $(w-\varepsilon w)h_{n(j+1)-1}$, and thus
\[a(I^{\prime \prime })\geq \sum_{j=1}^{r}(w-\varepsilon w)h_{n(j+1)-1}\geq w\sum_{j=2}^{r+1}(1-\varepsilon)h_{n(j)}\geq (1-\varepsilon)^2w\cdot h\geq (1-2\varepsilon)w\cdot h\]
\end{proof}

\subsection{Steinberg's Algorithm}
\label{subsec:steinberg}
We will make use of Steinberg's algorithm \cite{steinberg1997strip} as a subroutine in our algorithms.

\begin{theorem}[Steinberg \cite{steinberg1997strip}]
\label{steinberg_alg}
We are given a set of rectangles $I^{\prime}$ and box $Q$ of size $w\times h$. Let $w_{\max}\leq w$ and $h_{\max}\leq h$ be the maximum width and the maximum height among the items in $I^{\prime}$, respectively. Also we denote $x_{+}:=\max\{x,0\}$. If,
\[2a(I^{\prime})\leq wh-(2w_{\max}-w)_{+}(2h_{\max}-h)_{+}\]
then $I^{\prime}$ can be packed into $Q$.
\end{theorem}

\subsection{Algorithm for checking guillotine separability}
We present an algorithm which checks whether a set of axis-aligned packed rectangles are guillotine separable.
\begin{lemma}
\label{lem_guillo1}
Given a set of packed rectangles $\Ip$ specified by their positions $(left(i), right(i))\times (bottom(i) ,top(i))$ for each $i\in \Ip$ with $|\Ip|=n$, we can  check in $O(n^3)$ time whether they are guillotine separable.
\end{lemma}
\begin{proof}
Using standard shifting arguments, we can pack all the rectangles in a box of $[0, 2n-1]\times [0, 2n-1]$ where all the rectangles have integer coordinates for all of their four corners. To demonstrate this, we show how to do this in the $x$-direction, i.e., we consider the projections of the rectangles on the $x$-axis and consider both of the endpoints for each rectangle and then we assign them integer coordinates from $[0,2n-1]$ in the same order. We do this similarly for the $y$-coordinates. Now, we specify a recursive procedure where for a box $\mathcal{C}$ in which we have a subset of rectangles $I^{\prime \prime}\subseteq \Ip$ packed, we check all of the horizontal cuts and the vertical cuts at integral points which are feasible, incorporate such feasible cuts and recurse on the resulting smaller boxes. That is, we check all cuts which are line segments joining $(i,0)$ and $(i,2n-1)$ for $i\in[2n-1]$ and line segments joining $(0,j)$ and $(2n-1,j)$ for $j\in[2n-1]$ and incorporate either all feasible horizontal cuts or all feasible vertical cuts and recurse further (with alternating cuts). We only check cuts at integral coordinates since all rectangles are packed at locations which have integer coordinates by our preprocessing.

It is easy to see that in each level of the guillotine cutting sequence, at least $1$ rectangle is separated from a box or else we can declare that they are not guillotine separable. And at each level of the guillotine cutting sequence we spend at most $O(n^2)$ time overall checking all the possible feasible cuts for each respective box. Since previous argument implies we can have at most $O(n)$ levels for the guillotine cutting sequence, the algorithm runs in $O(n^3)$ time.
\end{proof}

\subsection{Resource Augmentation}\label{subsec_resource_aug}
	In this section we state without their respective proofs the necessary resource augmentation lemmas.
	\begin{lemma}
	\label{lemma1_old}
	(Resource Augmentation Packing Lemma \cite{GalvezGIHKW21}) Let $\R'$ be a collection of rectangles that can be packed into a box of size $a \times b$, and $\epsau > 0$ be a given constant. Here $a$ denotes the height of the box and $b$ denotes the width. Then there exists a container packing of $\R''\subseteq \R'$ inside a box of size $a\times(1+\epsau)b$ (resp.~$(1+\epsau)a\times b$) such that:
	\begin{enumerate}
	\item $\profit(\R'')\geq (1-O(\epsau))\profit(\R')$;
	\item the number of containers is $C_{ra}=O_{\epsau}(1)$ and their sizes belong to a set of cardinality $n^{O_{\epsau}(1)}$ that can be computed in polynomial time;
	\item the total area of the the containers is at most $\area(\R')+\epsau ab$;  
	\end{enumerate}
	\end{lemma}

\begin{lemma}[\cite{khan2021guillotine}]
\label{lemma_res_aug_guillotine_separability}
If we have a guillotine separable packing of items $I$ in a rectangular box $B$, the container packing of Lemma~\ref{lemma1_old} is also a guillotine separable nice packing.
\end{lemma}

\section{Relationship between different constants}{\label{section_constants}}
Now let us define a function $g(\delta,\eps)$ to denote an upper bound on the number of containers in the packing obtained using Lemma~\ref{lem:structural}. Then $g(\delta,\eps) \ge\left|\B_{hor}\right|+\left|\B_{tall+ver}\right|+\left|\B_{large}\right|+\left|\B_{small}\right|$, i.e., the upper bound on the  total number of containers for items in $I_{hor}\cup I_{tall}\cup I_{vert}\cup I_{large}\cup I_{small}$. This function is used to get an upper bound on $\mu$ (to be chosen sufficiently small compared to $\delta$ as defined in the paragraph below) and to define the function $f$ from Lemma~\ref{class_1}.

For the pseudo-polynomial time algorithm, $\eps_{ra}=\eps$, $\left|\B_{hor}\right|\le \frac{96C_{ra}}{\eps\delta^2}+C_{ra}$
and $\left|\B_{tall+ver}\right|\leq \frac{96}{\eps\delta}(\frac{1}{\delta} 2^{(1/\delta^{2})})^{\frac{1}{\delta}}$, $\left|\B_{large}\right|\leq \frac{1}{\delta^2}$ and $\left|\B_{small}\right|\leq 4(\left|\B_{hor}\right|+\left|\B_{tall+ver}\right|+\left|\B_{large}\right|+1)^2$. Note that $C_{ra}$ is the number of containers we get from resource augmentation as in Lemma \ref{lemma1_old}. Let $g(\delta,\eps)=\left|\B_{hor}\right|+\left|\B_{tall+ver}\right|+\left|\B_{large}\right|+\left|\B_{small}\right|$. From the condition for small containers in Lemma \ref{lem_small_rearrange_1}, we get that $\mu\leq \frac{\delta \eps}{(g(\delta,\eps))^2}$. Since, $g(\delta,\eps)$ is a decreasing function in $\eps$, we choose the function $f$ from Lemma \ref{class_1} as $f(x)=\frac{x\eps}{g(x,\eps)^2}$.

For the polynomial time algorithm,
$\eps_1=\frac{1}{3g(\delta,\eps)}$ which is the constant associated with the length of line segment $l^{*}$. $\eps_2=\frac{\eps}{4\left|\B_{hor}\right|}$,  $\eps_3=\frac{\eps_1}{4\left|\B_{ver}\right|}$, $\eps_4=\mu$, $\eps_5=\frac{\eps_1\delta}{6}$, $\eps_6=\frac{\eps\delta}{6}$, $\mu\leq \frac{\delta \eps}{(g(\delta,\eps))^2}$. The function $g$ is as defined before.

\section{Hardness}
\label{section_hardness}

\begin{theorem}
\label{thm_pseudo_hardness}
There exists no exact pseudo-polynomial time algorithm for the $2$-dimensional guillotine strip packing problem unless $\mathsf{P=NP}$.
\end{theorem}
\begin{proof}
\textsc{Bin packing} is a strongly $\mathsf{NP}$-Hard problem~\cite{garey1978strong} and $2$-dimensional guillotine strip packing is a generalization of the same. To see why, consider the reduction: Reduce an instance of \textsc{Bin packing} where given, items of sizes $i_1,...,i_n$ and the problem is to find whether it is possible to pack said items in $k$ bins ($k\in\mathbb{Z}^{+}$, $i_1,...i_n\in[0,1]$ and $i_1,...i_n\in \mathbb{Q}$) to an instance of GSP where width $W$ of the half-strip is $1$ and for each $j\in[n]$ we have a rectangle $r_j$ such that $h(r_j)=1$ and $w(r_j)=i_j$. Note that here the widths of the rectangles may not be integers but they can be appropriately scaled (along with the width of the half-strip) to ensure that. The objective in the GSP instance is to find whether there exists a guillotine separable packing of height at most $k$. The proof of the equivalence of this reduction follows in the same vein as the proof of Theorem~\ref{thm_hardness_approx}.
\end{proof}

\begin{theorem}
\label{thm_hardness_approx}
There exists no polynomial time algorithm for the $2$-dimensional guillotine strip packing problem with an approximation ratio $(\frac{3}{2}-\eps )$ for any $\eps >0$ unless $\mathsf{P=NP}$.
\end{theorem}
\begin{proof}
Consider the following reduction from the \textsc{Partition} problem. For an instance of the \textsc{Partition} problem $P$ where we are given positive integers $i_1,...,i_n$ such that $T=\sum_{j=1}^{j=n} i_j$ and where we have to check if we can partition the given numbers into two sets $S_1$ and $S_2$ such that $\sum_{i_j\in S_1}i_j=\sum_{i_k\in S_2}i_k=T/2$ , we construct the following instance $I$ of the $2$-dimensional guillotine strip packing problem: Rectangles $\mathcal{R}=\{R_1,...,R_n\}$ such that $h(R_k)=1$ for any $k\in[n]$ and $w(R_k)=i_k$ and we want to check if there exists a guillotine separable packing of the rectangles in $\mathcal{R}$ in a half-strip of width $T/2$ such that height of this packing is at most $2$. 

We now show that for the above reduction, If the answer to the \textsc{Partition} instance $P$ is Yes, there exists a guillotine separable packing of height exactly $2$ for instance $I$. And if the answer to the instance $P$ is No, the optimal guillotine separable packing has height at least $3$. Note that any optimal packing can have only an integral height as all rectangles have a height of exactly $1$. Now, if the answer to $P$ is Yes, we have $2$ sets $S_1$ and $S_2$ such that $S_1\cup S_2=\{i_1,...,i_n\}$ and $\sum_{i_j\in S_1}i_j=\sum_{i_k\in S_2}i_k=T/2$ where $T=\sum_{j=1}^{j=n} i_j$. Thus, we first pack all rectangles corresponding to numbers in $S_1$ from left to right at the bottom of the half-strip starting from $x=0$ and without leaving any gap. Since the height of each rectangle is $1$, we pack all rectangles similarly as before corresponding to numbers in $S_2$ from left to right on top of this packing. This results in a packing of height $2$. It is a $2$-stage guillotine separable packing because we first consider the horizontal cut $y=1$ and then we separate all the rectangles in the resulting $2$ boxes by way of vertical cuts. 

We show that if there exists a guillotine separable packing of rectangles in $\mathcal{R}$ of height at most $2$, then the answer to the instance $P$ would be Yes. Observe that any packing of the rectangles has to have a height of at least $2$ since the area of the rectangles in $\mathcal{R}$ is $T$ and the width of the half-strip is $T$. If we have a guillotine separable packing of height $2$, then by the area lower bound and the fact that all the rectangles have height $1$, we have a guillotine cut at $y=1$ and both the resulting boxes are completely filled by rectangles packed side by side without any space in between. Hence, we consider all items corresponding to rectangles packed in $1$ box as $S_1$ and the others as $S_2$. Hence, we have a positive \textsc{Partition} instance.  Taking the contrapositive of this statement proves our first claim of equivalence of the reduction.

If we have a polynomial time algorithm $A$ for the $2$-dimensional guillotine strip packing problem with an approximation ratio $({3}/{2}-\eps )$ for $\eps >0$, then 
\begin{enumerate}
\item If the instance $I$ is a Yes instance, we have a guillotine separable packing of height $2$ and by applying the algorithm, we get a guillotine separable packing of height at most $2({3}/{2}-\eps )<3$. And since only integral height packings are possible, we get a packing of height $2$.
\item If the instance $I$ is a No instance, from our reduction we have a guillotine packing of height at least $3$.
\end{enumerate}
Consider the following polynomial time algorithm for the \textsc{Partition} problem. For an instance $P$ of \textsc{Partition}, we reduce the problem to an instance $I$ of $2$-dimensional guillotine strip packing as described. Then we apply the approximation algorithm $A$ on this instance. If we get a packing of height $2$, then by our previous claim for the reduction, $P$ is a Yes instance. Else if we get a packing of height at least $3$ we have a No instance.

This proves the theorem.
\end{proof}

\end{document}